\documentclass[10pt]{article}
\usepackage{customcommands}

\title{Partial Identification in Nonseparable Binary Response Models with Endogenous Regressors\\~\\ \footnotetext{We are grateful to James Heckman, Marc Henry, Roger Koenker, and to seminar audiences at Columbia University and Michigan State University for helpful feedback. We also thank Martin Weidner and the organizers of the Chamberlain Seminar, and are grateful to Florian Gunsilius, Sukjin Han, Wayne Gao, and Takuya Ura for their questions and feedback, and to Adam Rosen for his thoughtful discussion. Jiaying Gu acknowledges financial support from the Social Sciences and Humanities Research Council of Canada. All errors are our own.}
\vspace{-.4in}
}
\date{\small June 20, 2022}
\author{Jiaying Gu\footnote{Jiaying Gu, Department of Economics, University of Toronto, 150 St. George Street, Toronto, Ontario, M5S3G7, Canada. Email: jiaying.gu@utoronto.ca.} \\ \textit{University of Toronto}
\and 
Thomas M. Russell\footnote{Thomas M. Russell, Department of Economics, Carleton University, 1125 Colonel By Drive, Ottawa, Ontario, K1S5B6, Canada. Email: thomas.russell3@carleton.ca. 
}\\\textit{Carleton University}}

\begin{document}
\maketitle
\vspace{-.4in}
\begin{abstract}
\noindent  This paper considers (partial) identification of a variety of counterfactual parameters in binary response models with possibly endogenous regressors. Our framework allows for nonseparable index functions with multi-dimensional latent variables, and does not require parametric distributional assumptions. We leverage results on hyperplane arrangements and cell enumeration from the literature on computational geometry in order to provide a tractable means of computing the identified set. We demonstrate how various functional form, independence, and monotonicity assumptions can be imposed as constraints in our optimization procedure to tighten the identified set. Finally, we apply our method to study the effects of health insurance on the decision to seek medical treatment. 
\end{abstract}

\medskip



\thispagestyle{empty}
\clearpage

\section{Introduction}

This paper considers partial identification of counterfactual parameters in a class of binary response models of the form:
\begin{align}
Y= \mathbbm{1}\{\tilde{\varphi}_{1}(X,\theta)^\top U + \tilde{\varphi}_{2}(X,\theta)\geq 0 \},\label{eq_main_model}
\end{align}
where $Y \in \{0,1\}$ is a binary outcome variable, $X \in \mathcal{X}\subset \mathbb{R}^{d_{x}}$ is vector of (possibly endogenous) covariates with finite support, $U \in \mathcal{U}:=\mathbb{R}^{d_{u}}$ is a vector of latent variables, $\theta \in \Theta \subset \mathbb{R}^{d_{\theta}}$ is a vector of structural parameters, and $\tilde{\varphi}_{1}:\mathcal{X}\times \Theta \to \mathbb{R}^{d_{u}}$ and $\tilde{\varphi}_{2}: \mathcal{X}\times \Theta \to \mathbb{R}$ are known functions. Our approach does not require any parametric distributional assumptions on the latent variables. We then consider counterfactuals that can be expressed as a function $\gamma: \mathcal{X} \to \mathcal{X}$ that reassigns each $x \in \mathcal{X}$ to a new value $\gamma(x) \in \mathcal{X}$. Our main focus is on bounding linear functionals of counterfactual probabilities.  

In our setting, nonparametric point-identification of the distribution of latent variables occurs only under restrictive conditions, including strong independence assumptions and large support conditions.\footnote{E.g. \cite{ichimura1998maximum}.} Control function approaches are often used to address the issue of endogenous regressors, but if endogenous regressors are discrete or the mechanism generating the endogenous regressors is poorly understood, then many of these approaches are not applicable. Partial identification arises as a natural alternative to methods for point-identification in the presence of endogenous and discrete covariates. This paper uses an optimization-based approach to bound counterfactual quantities, which allows researchers to easily construct sharp bounds on counterfactual quantities under a variety of different assumptions by simply altering the constraints in our optimization problems.

In the absence of parametric distributional assumptions, our analysis reveals the importance of a special partition of the latent variable space into \textit{response types} that have identical responses in all possible counterfactuals.  We are not the first to emphasize the importance of response types, and our discussion echoes the insights of \cite{balke1994counterfactual} and \cite{heckman2018unordered}, among others. Similar to these works, we show that functional form and monotonicity assumptions amount to assigning zero probability to certain response types. Before bounding counterfactual quantities, it is thus necessary to determine which response type are possible/impossible in a given model. In our particular class of models, the latent variable space admits a partition into cells defined by a collection of hyperplanes, each cell corresponding to a unique response type. Using the cell enumeration algorithm of \cite{gu2020nonparametric}, we show how to enumerate all response types in a time polynomial in the number of input hyperplanes. After enumerating response types, we then demonstrate how various counterfactual quantities can be bounded by solving a sequence of linear programming problems. We also study the case when the index function is linear in parameters, in which case we show how sharp bounds on counterfactual probabilities can be computed without the need to grid over the entire parameter space. Finally, we present a consistency result, we show how to adapt the inference procedure of \cite{cho2021simple} to our setting, and we apply our method to study the effects of private health insurance on the decision to seek medical treatment.  

\subsection{A Review of the Relevant Literature} 

Binary response models with endogenous regressors have been studied extensively. Point identified approaches include linear probability models, maximum likelihood estimation (e.g.\! the bivariate probit), control function approaches, or approaches based on special regressors. All of these approaches have well-documented limitations.\footnote{See \cite{lewbel2012comparing} for a review. }  
Nonparametric identification was studied in binary choice and threshold crossing models by \cite{matzkin1992nonparametric}, and in more general nonseparable models by \cite{matzkin2003nonparametric} and \cite{chernozhukov2005iv}, among many others. \cite{vytlacil2007dummy} studied nonparametric identification of the average treatment effect in a discrete triangular system with a binary endogenous variable under a weak separability assumption in the outcome equation. Since then a number of paper have studied point identification in similar models with discrete endogenous variables (e.g. \cite{han2017identification}, \cite{vuong2017counterfactual}, \cite{chen2020identification} and \cite{khan2021informational}) and continuous endogenous variables (e.g. \cite{imbens2009identification}, \cite{d2015identification}, \cite{torgovitsky2015identification}). Important precedents to the work presented here from the literature on point identification in random coefficient models include \cite{ichimura1998maximum}, \cite{gautier2013nonparametric} and \cite{gu2020nonparametric}. However, these papers focus almost exclusively on the point-identified case with a linear index function and exogenous covariates with large support. 

Many authors have also used partial identification methods to relax the assumptions required for point-identification in these models. In a relevant series of papers, \cite{chesher2013instrumental} and \cite{chesher2014instrumental} show how to use random set theory to characterize the identified set of structures in discrete choice models.\footnote{The general formulation of their approach is presented in \cite{chesher2017generalized}.} Similar to the current paper, they do not provide a model for the endogenous explanatory variables, rendering the discrete choice model \textit{incomplete}. They construct bounds on structural parameters using a characterization of the sharp set of constraints based on the results of \cite{artstein1983distributions}.\footnote{See also \cite{norberg1992existence} and \cite{molchanov2017theory} Corollary 1.4.11.} 
Although our paper focuses primarily on computational issues that arise when bounding counterfactual parameters, we present a comparison of our approach with the approach of \cite{chesher2013instrumental} and \cite{chesher2014instrumental} in Appendix \ref{appendix_artstein_comparison}. 

Our approach to identification is closely related to approaches in \cite{galichon2011set}, \cite{laffers2019identification} and \cite{torgovitsky2019partial}, who demonstrate how to construct sharp bounds on various parameters in models with discrete variables by partitioning the latent variable space and discretizing the latent variables. We also use an identification argument based on partitioning the latent variable space, and we demonstrate how to practically compute the relevant partition using results from the literature on computational geometry. In certain cases, we also show how to avoid griding over the entire parameter space when computing the identified set.

There are a number of other relevant papers in the literature on partial identification in discrete choice models. \cite{manski2007partial} considers counterfactual choice probabilities in a setting with partial identification, and shows how these counterfactual choice probabilities can be bounded using optimization problems. However, the response type approach in our paper is quite different. We also show how to practically incorporate different assumptions, and we allow for endogenous explanatory variables.\footnote{This latter point differentiates our work from \cite{chiong2017counterfactual} and \cite{allen2019identification}.} In other related work, \cite{tebaldi2019nonparametric} study the problem of computing various counterfactual quantities in a nonparametric discrete choice model with an application to consumer choice of health insurance in California. However, they focus on quasi-linear utility functions and use the particular structure of their setting to resolve the issue of endogeneity by conditioning on a set of covariates. 
 Computational considerations are also not the main focus of these papers. 

This paper is also related to papers that bound treatment effect parameters. Nonparametric bounds were proposed by \cite{manski1990nonparametric}, and since then contributions have been made by \cite{manski1997monotone}, \cite{manski2000monotone}, \cite{heckman2001instrumental}, \cite{bhattacharya2008treatment}, \cite{manski2009more}, \cite{chiburis2010semiparametric}, \cite{shaikh2011partial}, \cite{bhattacharya2012treatment} and \cite{mourifie2015sharp}, among many others. However, most papers derive closed-form bounds and provide corresponding proofs of sharpness for a fixed target parameter and under a fixed set of assumptions. Similar to other optimization-based approaches, the advantage of our procedure is its flexibility: the researcher can easily modify the model or target parameter and obtain new bounds without the need to derive closed form expressions, or to provide a corresponding proof of sharpness.\footnote{Examples of optimization-based approaches to bounding treatment effect parameters include \cite{balke1997bounds}, \cite{laffers2019bounding}, \cite{russell2021sharp}, \cite{mogstad2018using} and \cite{gunsilius2020path}. } Our optimization-based bounds are shown to be sharp under a variety of different assumptions, including flexible functional form assumptions, and any mix of fixed and random coefficients on either exogenous or endogenous covariates. These assumptions are of direct interest to those familiar with structural models of binary outcome variables, and allow us to obtain many previous results as a special case. We hope our approach may also serve as a useful method of sensitivity analysis for those primarily interested in point identified models.

Finally, our paper makes numerous connections to the literature on computational geometry. Computation of our bounds requires the analysis of a partition of the latent space determined by a finite collection of hyperplanes. This turns out to be a well studied subject in combinatorial geometry, and leads us to consideration of the enumeration algorithm proposed by \cite{gu2020nonparametric}, who build on the work of \cite{rada2018new}. Our profiling procedure also makes use of the double-description algorithm proposed by \cite{Fukuda}. 

\subsection{Paper Outline and Notation} 

The remainder of the paper proceeds as follows. Section \ref{section_theoretical_considerations} introduces the main theoretical framework and main assumptions. Section \ref{section_practical_considerations} studies practical implementation of the theoretical framework from Section \ref{section_theoretical_considerations} and introduces our optimization-based bounding procedure for counterfactual probabilities. Section \ref{section_additional_assumptions} then demonstrates how to introduce functional form, independence, and monotonicity assumptions into our bounding procedure. Section \ref{section_consistency_bias_correction_inference} discusses estimation and inference, and Section \ref{section_application} applies our methodology to study the impact of health insurance on utilization of health care services. Section \ref{section_conclusion} concludes. All proofs can be found in Appendix \ref{appendix_proofs}. Appendix \ref{appendix_additional_definitions_and_results} provides some additional discussion of the results presented in the main text. Appendix \ref{appendix_artstein_comparison} compares our procedure to an approach based on Artstein's inequalities, and Appendix \ref{appendix_figure} contains supplementary material for our application. \\~\\

\noindent \textbf{Notation:} The following notation is relevant for both the main text and the appendices. Given a subset $\mathcal{X}$ of Euclidean space, we use $\mathfrak{B}(\mathcal{X})$ to denote the Borel $\sigma-$algebra on $\mathcal{X}$. For two measurable spaces $(\mathcal{X},\mathfrak{B}(\mathcal{X}))$ and $(\mathcal{X}',\mathfrak{B}(\mathcal{X}'))$, the product $\sigma-$algebra on $\mathcal{X}\times \mathcal{X}'$ is denoted by $\mathfrak{B}(\mathcal{X})\otimes \mathfrak{B}(\mathcal{X}')$. Random variables are denoted using capital letters, and if $X: (\Omega,\mathfrak{A})\to (\mathcal{X},\mathfrak{B}(\mathcal{X}))$ is a random variable defined on the probability space $(\Omega,\mathfrak{A},P)$, then we use $P_{X}$ to denote the probability measure induced on $\mathcal{X}$ by $X$; that is, for any $A \in \mathfrak{B}(\mathcal{X})$, $P_{X}(A) := P(X^{-1}(A))$. Furthermore, we interpret $P_{X\mid X'}(X \in A \mid X' =x')$ as a regular conditional probability measure. Finally, $P_{X\mid X'}$ is used as shorthand for the collection $P_{X\mid X'}:=\{P_{X\mid X'}(\,\cdot\, \mid X'=x') : x' \in \mathcal{X}'\}$. We do not explicitly differentiate between scalars and vectors, or random variables and random vectors. To keep the notation clean, we sometimes omit the transpose when combining column vectors; that is, if $v_{1}$ and $v_{2}$ are two column vectors, rather than write $v=(v_{1}^\top,v_{2}^\top)^\top$ we instead write $v=(v_{1},v_{2})$, where it is understood that $v$ is a column vector unless otherwise specified. The cardinality of a set $\mathcal{X}\subset \mathbb{R}^{d}$ is given by $|\mathcal{X}|$.

\section{General Framework: Theoretical Considerations}\label{section_theoretical_considerations}

\subsection{Main Assumptions and Definitions}

We begin by introducing our main assumptions on the binary response models under consideration.
\begin{assumption}\label{assumption_basic}
There exists a complete probability space $(\Omega, \mathfrak{A},P)$, a random variable $Y:\Omega \to \{0,1\}$, and random vectors $X:\Omega\to \mathcal{X} \subset \mathbb{R}^{d_{x}}$ and $U:\Omega\to \mathcal{U} = \mathbb{R}^{d_{u}}$ satisfying: 
\begin{align}
Y= \mathbbm{1}\{\varphi(X,U,\theta_{0})\geq 0  \}\,\, a.s.,\label{eq_model}
\end{align}
for some function $\varphi(\,\cdot\,,\theta_{0}): \mathcal{X}\times \mathcal{U} \to \mathbb{R}$ parameterized by $\theta_{0} \in \Theta\subset \mathbb{R}^{d_{\theta}}$ with:
\begin{align}
\varphi(x,u,\theta)=\tilde{\varphi}_{1}(x,\theta)^\top u + \tilde{\varphi}_{2}(x,\theta),\label{eq_model2}
\end{align}
where $\tilde{\varphi}_{1}(\,\cdot\,,\theta):\mathcal{X} \to \mathbb{R}^{d_{u}}$ and $\tilde{\varphi}_{2}(\,\cdot\,,\theta): \mathcal{X} \to \mathbb{R}$ are measurable for each $\theta$. Furthermore, $|\mathcal{X}|=:m < \infty$, the spaces $\mathcal{X}$ and $\mathcal{U}$ are equipped with the Borel $\sigma-$algebra,  and the distribution of U assigns zero probability to all sets of the form $\{ u \in \mathcal{U} : \varphi(x,u,\theta_{0}) =0\}$.
\end{assumption}
In Assumption \ref{assumption_basic}, $U \in \mathcal{U}$ is a vector of latent variables, $\theta \in \Theta$ is a vector of fixed coefficients, and $X\in \mathcal{X}$ is a vector of covariates. From \eqref{eq_model2} we restrict the index function to be linear in the latent variables $U \in \mathcal{U}$, although the model in Assumption \ref{assumption_basic} still allows for general nonseparability between covariates and latent variables. Importantly, Assumption \ref{assumption_basic} imposes that the random vector $X$ has finite support. 

In this model the latent variables can also be interpreted as random coefficients. A special case of linearity occurs when the function $\varphi$ is additively separable in a scalar latent variable $U$, which occurs, for instance, when $\tilde{\varphi}_{1}(x,\theta)=1$. A full analysis of this special case using the framework in this paper is taken up in Appendix \ref{appendix_additively_separable}. 
Finally, assuming the distribution of $U$ assigns zero probability to sets of the form $\{u \in \mathcal{U} : \varphi(x,u,\theta)=0\}$ allows for a simplification of the cell enumeration algorithm introduced in the next section. This assumption is implied by absolute continuity of the distribution of latent variables with respect to the Lebesgue measure, which is a standard assumption in this literature.

We assume that the researcher's objective throughout is to obtain a sharp set of constraints defining the identified set of latent variable distributions, and to use these constraints to bound various counterfactual quantities, such as counterfactual conditional probabilities.\footnote{ 
Similar to previous works, we take the \textit{selection} relation as a primitive relation on which to construct a definition of the identified set. The close connection between the selection relation from random set theory and the concept of \textit{observational equivalence} from the work in econometrics on identification has been appreciated in \cite{beresteanu2011sharp}, \cite{beresteanu2012partial}, \cite{chesher2013instrumental}, \cite{chesher2014instrumental}, and \cite{chesher2017generalized}, among many others. We continue this work here.} Define the set:
\begin{align}
 \mathcal{U}(y,x,\theta)&:=\left\{u \in \mathcal{U} : y = \mathbbm{1}\{\varphi(x,u,\theta) \geq 0  \} \right\}.\label{eq_random_set}
\end{align}
\cite{chesher2017generalized} call this set the $U-$level set; intuitively, it delivers all possible values of the latent variables $u$ consistent with the vector $(y,x,\theta)$ given the binary response model in \eqref{eq_model}.  A \textit{measurable selection} from the random set $ \mathcal{U}(Y,X,\theta)$ is a random vector $U: \Omega \to \mathcal{U}$ satisfying $U \in  \mathcal{U}(Y,X,\theta)$ a.s.\footnote{A general definition of a selection and a random set is provided in Appendix \ref{appendix_measurability_results}. In Appendix \ref{appendix_measurability_results} we prove that $\mathcal{U}(Y,X,\theta)$ is suitably measurable and thus is a random set under our assumptions (see Lemma \ref{lemma_effros_measurability}). We also prove the existence of a universally measurable selection (see Lemma \ref{lemma_selection}).} Importantly, given a distribution of the observable random vectors $(Y,X)$, a structural function $\varphi$ and a fixed coefficient $\theta \in \Theta$, any two measurable selections $U$ and $U'$ from the random set $\mathcal{U}(Y,X,\theta)$ are \textit{observationally equivalent} in the sense that both latent variable vectors $U$ and $U'$ are consistent with the observed distribution of $Y$ and $X$ for the vector of parameters $\theta \in \Theta$ through the model \eqref{eq_model}. This idea can be used to define the identified set. 

\begin{definition}[Identified Set]\label{definition_identified_set}
Under Assumption \ref{assumption_basic}, the (joint) identified set $\mathcal{I}_{Y,X}^{*}$ of conditional latent variable distributions $P_{U\mid Y,X}$ and fixed coefficients $\theta$ is the set of all pairs $(P_{U\mid Y,X},\theta)$ satisfying:
\begin{align}
P_{U\mid Y,X}(U \in  \mathcal{U}(Y,X,\theta) \mid Y=y, X=x)=1, \,\, P_{Y,X}-a.s., \label{eq_identified_set}
\end{align}
and such that the distribution $P_{U} = P_{U\mid Y,X} P_{Y,X}$ assigns zero probability to all sets of the form $\{u \in \mathcal{U} : \varphi(x,u,\theta)=0\}$.
\end{definition}
This definition depends on the observed distribution of $(Y,X)$ through the almost-sure relation in \eqref{eq_identified_set}. Conditioning the latent variable distribution on the vector $(Y,X)$ is carried throughout the paper, and we show in Section \ref{section_application} that it allows us to bound counterfactual parameters that may be relevant to policy analysis. Definition \ref{definition_identified_set} can also be used to define other related identified sets, including identified sets for conditional latent variable distributions of the form $P_{U\mid X}$ or $P_{U}$. 

\begin{example}
Consider the simple additively separable threshold crossing model:
\begin{align*}
Y = \mathbbm{1}\{ X  \theta_{0} \geq U \},
\end{align*}
where $X \in \mathcal{X} \subset \mathbb{R}$ has finite support, and $U \in \mathbb{R}$. This is a special case of the model we consider, and is explored in detail in Appendix \ref{appendix_additively_separable}. Here we have:
\begin{align*}
\mathcal{U}(1,x,\theta) &:= \left\{ u \in \mathbb{R} : u \leq x  \theta \right\},\\
\mathcal{U}(0,x,\theta) &:= \left\{ u \in \mathbb{R} : u > x  \theta \right\}.
\end{align*}
Now suppose that $P(Y=y, X=x)>0$ for all $(y,x) \in \mathcal{Y} \times \mathcal{X}$. From Definition \ref{definition_identified_set}, the (joint) identified set $\mathcal{I}_{Y,X}^{*}$ is the set of all pairs $(P_{U\mid Y,X},\theta)$ satisfying the conditions:
\begin{align}
P_{U\mid Y,X}( U \leq X  \theta \mid Y=1, X=x)=1,\qquad \forall x \in \mathcal{X},\\
P_{U\mid Y,X}( U > X  \theta\mid Y=0, X=x)=1, \qquad \forall x \in \mathcal{X},
\end{align}
and such that the distribution $P_{U} = P_{U\mid Y,X} P_{Y,X}$ assigns zero probability to the sets $\{ \{u \in \mathbb{R} : u = x\theta\} : x \in \mathcal{X}\}$.
\end{example}

In this paper, we consider counterfactuals characterized by the occurrence of an exogenous \textit{intervention} that modifies at least one of the explanatory variables. 

\begin{assumption}[Counterfactual Domain]\label{assumption_counterfactual_domain}
For some collection of functions $\Gamma$ with typical element $\gamma: \mathcal{X}  \to \mathcal{X}$, there exists a collection of random variables $\{Y(\,\cdot\,,\gamma): \Omega \to \{0,1\} \mid \gamma \in \Gamma\}$, abbreviated $Y_{\gamma}:=Y(\,\cdot\,,\gamma)$, representing counterfactual choices for each $\gamma$ such that $Y_{\gamma}: \Omega \to \{0,1\}$ is measurable for each $\gamma$, and:
\begin{align}
P_{Y_{\gamma}\mid Y,X,U}\left(Y_{\gamma} = \mathbbm{1}\{\varphi(\gamma(X),U,\theta_{0}) \geq 0  \}\mid Y=y,X=x,U=u\right)=1,  \label{eq_pygamma}
\end{align}
$P_{Y,X,U}-$a.s. for the same $\theta_{0} \in \Theta$ as in Assumption \ref{assumption_basic}, and for all $\gamma \in \Gamma$.
\end{assumption}
Assumption \ref{assumption_counterfactual_domain} implies that (i) counterfactual response variables indexed by $\gamma \in \Gamma$ exist on the common probability space from Assumption \ref{assumption_basic}, and (ii) such counterfactual response variables are equal (almost surely) to the values that would arise after an intervention on the system represented by \eqref{eq_model}. 
Here each counterfactual is represented by a function $\gamma: \mathcal{X} \to \mathcal{X}$, which allows us to consider a general class of counterfactuals. Although each function $\gamma$ is seen as a map from $\mathcal{X}$ to itself, this does not prevent consideration of counterfactuals where $\gamma$ selects values of $x$ that have never been observed in the data. Such cases can be accommodated by simply extending the support $\mathcal{X}$ from Assumption \ref{assumption_basic} to include any counterfactual pair $x$ of interest.\footnote{In particular, if $x \in \mathcal{X}$ but $\gamma(x)=x' \notin \mathcal{X}$, then redefine $\mathcal{X} \leftarrow \mathcal{X} \cup \{x'\}$. This approach does not affect anything we present in this paper, since we always require any relation to the observed distribution of $(Y,X)$ to hold only almost-surely.  } 

\begin{definition}[Identified Set of Counterfactual Conditional Distributions]\label{definition_identified_set_counterfactual}
Under Assumptions \ref{assumption_basic} and \ref{assumption_counterfactual_domain}, the identified set of counterfactual conditional distributions $\mathcal{P}_{Y_{\gamma}\mid Y,X,U}^{*}$ is the set of all conditional distributions $P_{Y_{\gamma}\mid Y,X,U}$ satisfying:
\begin{align}
P_{Y_{\gamma}\mid Y,X,U}\left(Y_{\gamma} = \mathbbm{1}\{\varphi(\gamma(X),U,\theta) \geq 0  \}\mid Y=y, X=x,U=u\right)=1, \label{eq_cf_xzt}
\end{align}
$P_{Y,X,U}-$a.s. for some $(P_{U\mid Y,X},\theta)\in \mathcal{I}_{Y,X}^{*}$. 
\end{definition}
Note this definition references the identified set $\mathcal{I}_{Y,X}^{*}$ from Definition \ref{definition_identified_set}. This definition can also be used to define other related identified sets, including for counterfactual distributions $P_{Y_{\gamma}\mid Y,X}$ or $P_{Y_{\gamma}\mid X}$, as well as identified sets for average structural functions and average treatment effects. 

\setcounter{example}{0}
\begin{example}[\textbf{cont'd}]
Consider again the simple additively separable threshold crossing model:
\begin{align*}
Y = \mathbbm{1}\{ X  \theta_{0} \geq U \},
\end{align*}
where $X \in \mathcal{X} \subset \mathbb{R}$ has finite support, and $U \in \mathbb{R}$. Now fix some $x' \in \mathcal{X}$ and consider the map $\gamma(x) = x'$ for all $x  \in \mathcal{X}$. This corresponds to a counterfactual map $\gamma:\mathcal{X}\to \mathcal{X}$ where all values of $X$ are set to the value $X=x'$. The counterfactual outcome variable $Y_{\gamma}$ satisfies:
\begin{align}
Y_{\gamma} = \mathbbm{1}\{x' \theta_{0} \geq U  \},
\end{align}
a.s. for the same $\theta_{0} \in \Theta$ as in Assumption \ref{assumption_basic}. The identified set of counterfactual conditional distributions $\mathcal{P}_{Y_{\gamma}\mid Y,X,U}^{*}$ is the set of all conditional distributions $P_{Y_{\gamma}\mid Y,X,U}$ satisfying:
\begin{align}
P_{Y_{\gamma}\mid Y,X,U}\left(Y_{\gamma} = \mathbbm{1}\{x' \theta \geq U \}\mid Y=y, X=x,U=u\right)=1, 
\end{align}
$P_{Y,X,U}-$a.s. for some $(P_{U\mid Y,X},\theta)\in \mathcal{I}_{Y,X}^{*}$. That is, the identified set contains all (and only those) conditional counterfactual distributions $P_{Y_{\gamma}\mid Y,X,U}$ consistent with the support restriction $Y_{\gamma} = \mathbbm{1}\{\varphi(x',U,\theta) \geq 0  \}$ a.s. for some pair $(P_{U\mid Y,X},\theta)$ belonging to the (joint) identified set of conditional latent variable distributions and fixed coefficients.
\end{example}

The following result connects the definitions and assumptions in this section.
\begin{theorem}\label{thm_counterfactual_identified_set}
Suppose that Assumptions \ref{assumption_basic} and \ref{assumption_counterfactual_domain} hold. Then a counterfactual conditional distribution $P_{Y_{\gamma}\mid Y,X}$ satisfies $P_{Y_{\gamma}\mid Y,X} \in \mathcal{P}_{Y_{\gamma}\mid Y,X}^{*}$ if and only if there exists a pair $(P_{U\mid Y,X},\theta) \in \mathcal{I}_{Y,X}^{*}$ satisfying:
\begin{align}
P_{Y_{\gamma}\mid Y,X}\left(Y_{\gamma} = 1\mid Y=y,X=x\right)= P_{U\mid Y,X}\left( \varphi(\gamma(X),U,\theta) \geq 0 \mid Y=y,X=x\right),\label{eq_theorem_counterfactual}
\end{align}
$P_{Y,X}-$a.s.
\end{theorem}

Theorem \ref{thm_counterfactual_identified_set} provides the theoretical link between the identified set of counterfactual conditional distributions, and the identified set for the pair $(P_{U\mid Y,X},\theta)$. 
While theoretically straightforward, it hides some important practical difficulties. In particular, verifying the existence of a pair $(P_{U\mid Y,X},\theta)$ that satisfies the conditions from Definition \ref{definition_identified_set} is a nontrivial task. This is at least partly due to the fact that $P_{U\mid Y,X}$ is an infinite dimensional object, even in the case when $X$ has finite support. This \textit{infinite dimensional existence problem} is exacerbated in practice by the fact that $P_{U\mid Y,X}$ must satisfy a number of constraints to ensure it is consistent with the binary response model through \eqref{eq_identified_set}, and to ensure it is a proper conditional probability measure. We consider these practical difficulties in detail in the next section. 

\section{General Framework: Practical Considerations}\label{section_practical_considerations}

To make progress, define the following vector-valued function:
\begin{align}
r(u,\theta) := \begin{bmatrix} \mathbbm{1}\{\varphi(x_{1},u,\theta)\geq 0\} &
\mathbbm{1}\{\varphi(x_{2},u,\theta) \geq 0\} &
 \ldots &
\mathbbm{1}\{\varphi(x_{m},u,\theta) \geq 0\} \end{bmatrix}^\top,\label{eq_response_type}
\end{align}
and for a fixed binary vector $s \in \{0,1\}^{m}$ define the set:
\begin{align}
\mathcal{U}(s,\theta):= \left\{ u \in \mathcal{U} : r(u,\theta) = s \right\}.\label{eq_theta_partition}
\end{align}
The sets from \eqref{eq_theta_partition} partition the space $\mathcal{U}$ into at most $2^{m}$ sets, with each set uniquely associated with a binary vector $s \in \{0,1\}^{m}$. The binary vectors $r(u,\theta)$ represent \textit{response types} (c.f. \cite{balke1994counterfactual} and \cite{heckman2018unordered}).\footnote{The collection of sets defining response types appears to be similar to the ``minimal relevant partition'' in \cite{tebaldi2019nonparametric}, as well as the partition described in \cite{chesher2014instrumental} Appendix B. } 
In the discrete choice setting, these response types tell us the choices that an individual with type indexed by $(u,\theta)$ \textit{would have made} had they been assigned an alternate value of $x$. Any two individuals characterized by values of $u$ from the same set $\mathcal{U}(s,\theta)$ make identical choices under every possible counterfactual assignment $\gamma: \mathcal{X} \to \mathcal{X}$, making this a natural grouping of latent types. 

After partitioning the space of latent variables using response types, various counterfactual objects of interest can be written as a disjoint union of the sets $\mathcal{U}(s,\theta)$ from \eqref{eq_theta_partition} that comprise the partition. For the sake of illustration, consider the binary vectors:
\begin{align}
S_{j} = \{ s \in \{0,1\}^{m} : s_{j} =1 \},\label{eq_sj}
\end{align}
for $j=1,\ldots,m$. Note that each set $S_{j}$ is comprised of all binary vectors that have a $j^{th}$ entry equal to $1$, and thus contain exactly $2^{m-1}$ elements.\footnote{It is useful to note that the sets $\{S_{j}\}_{j=1}^{m}$ are not disjoint; indeed, it is easy to show that $S_{j} \cap S_{k} \neq \emptyset$ and $S_{j} \cap S_{k}^{c} \neq \emptyset$ for every $j\neq k$.} By definition of the sets $\mathcal{U}(s,\theta)$ and $S_{j}$ we have:
\begin{align}
\{u \in \mathcal{U} : \varphi(x_{j},u,\theta) \geq 0\} = \bigcup_{s \in S_{j}} \mathcal{U}(s,\theta). 
\end{align}
Furthermore, for $s' \neq s$  we have $\mathcal{U}(s',\theta) \cap \mathcal{U}(s,\theta) = \emptyset$, so that this union is a disjoint union. Thus, we have the following decomposition:
\begin{align}
P_{U\mid Y,X}\left( \varphi(x_{j},u,\theta) \geq 0 \mid Y=y,X=x\right)= \sum_{s \in S_{j}}  P_{U\mid Y, X}\left( \mathcal{U}(s,\theta) \mid Y=y, X=x\right).\label{eq_specific_counterfactual}
\end{align}
Such a decomposition holds for any $j=1,\ldots,m$. When the conditioning value $x$ differs from the value $x_{j}$ in the structural function, an application of Theorem \ref{thm_counterfactual_identified_set} shows that the left hand side of this display represents a counterfactual conditional distribution, illustrating the connection between response types and counterfactual choices.

The next result shows that, in order to rationalize a given collection of counterfactual conditional distribution, for each $\theta$ it is both necessary and sufficient to construct a probability measure on sets of the form $\mathcal{U}(s,\theta)$ from \eqref{eq_theta_partition} that satisfies the constraints of Theorem \ref{thm_counterfactual_identified_set}. In the statement of the result we redefine $\gamma:\mathbb{N}\to \mathbb{N}$ to denote the index of the point in $\{x_{1},,\ldots,x_{m}\}$ assigned under counterfactual $\gamma$, and we set $S_{\gamma(j)} := \{ s \in \{0,1\}^{m} : s_{\gamma(j)} =1 \}$, the analog of $S_{j}$ from \eqref{eq_sj}. 
\begin{theorem}\label{theorem_infinite_to_finite}
Suppose Assumptions \ref{assumption_basic} and \ref{assumption_counterfactual_domain} hold. Fix some $\theta \in \Theta$ and consider the collection of sets:
\begin{align}
\mathcal{A}(\theta):= \{\text{int}(\mathcal{U}(s,\theta)) : s \in \{0,1\}^{m}\}. \label{eq_sufficient_sets}
\end{align}
Then for any collection of counterfactual conditional distributions $P_{Y_{\gamma} \mid Y,X}$, there exists a collection of Borel conditional probability measures $P_{U\mid Y,X}$ satisfying \eqref{eq_theorem_counterfactual} with $(P_{U\mid Y,X},\theta) \in \mathcal{I}_{Y,X}^{*}$ if and only if there exists a collection $P_{U \mid Y,X}$ of probability measures on the sets in $\mathcal{A}(\theta)$ satisfying:
\begin{align}
\sum_{s \in S_{j}}  P_{U\mid Y,X}\left( \text{int}(\mathcal{U}(s,\theta)) \mid Y=1,X=x_{j}\right)&=1,\label{eq_finite_existence1}\\
\sum_{s \in S_{j}^{c}}  P_{U\mid Y,X}\left( \text{int}(\mathcal{U}(s,\theta)) \mid Y=0,X=x_{j}\right)&=1,\label{eq_finite_existence2}\\
\sum_{s \in S_{\gamma(j)}}  P_{U\mid Y,X}\left( \text{int}(\mathcal{U}(s,\theta)) \mid Y=y,X=x_{j}\right)&=P_{Y_{\gamma}\mid Y,X}\left( Y_{\gamma}=1 \mid Y=y,X=x_{j}\right),\label{eq_finite_existence3}
\end{align}
for all $y \in \{0,1\}$ and $j \in \{1,\ldots,m\}$ assigned positive probability.
\end{theorem}
Theorem \ref{theorem_infinite_to_finite} reduces our infinite dimensional existence problem to a finite dimensional existence problem. For a fixed $\theta \in \Theta$, instead of verifying there exists a Borel probability measure $P_{U \mid Y, X}$ satisfying \eqref{eq_identified_set} and \eqref{eq_theorem_counterfactual} (an infinite-dimensional object), Theorem \ref{theorem_infinite_to_finite} shows it suffices to verify there exists a finite dimensional probability vector with typical element $P_{U\mid Y,X}\left( \text{int}(\mathcal{U}(s,\theta)) \mid Y=y, X=x\right)$ satisfying \eqref{eq_finite_existence1} - \eqref{eq_finite_existence3}. Note that this result relies crucially on the finiteness of $\mathcal{X}$. There is a close connection between Theorem \ref{theorem_infinite_to_finite} and the bounding approach based on Artstein's inequalities (e.g. \cite{chesher2017generalized}) and optimal transportation (e.g. \cite{galichon2011set}).\footnote{In Appendix \ref{appendix_artstein_comparison} we show that Theorem \ref{theorem_infinite_to_finite} is equivalent to a characterization based on Artstein's inequalities after conditioning on the value of the endogenous variables. This conditioning allows us to obtain a much smaller number of \textit{equality} constraints when compared to the full set of unconditional constraints arising from Artstein's inequalities. However, Artstein's inequalities can handle continuous instruments.  It is also well known that Artstein's inequalities are equivalent to the existence of a certain zero-cost optimal transport problem (see \cite{galichon2009test}, \cite{galichon2011set}, and \cite{galichon2016optimal}). } Related results have also appeared in \cite{laffers2019identification} and \cite{torgovitsky2019partial}. The finite number of linear constraints from Theorem \ref{theorem_infinite_to_finite} leads naturally to a linear programming approach to bounds on counterfactual quantities.

\subsection{Optimization Formulation}\label{section_optimization_formulation}
For simplicity, we suppose throughout this subsection that our objective is to bound the counterfactual probability:
\begin{align}
P_{Y_{\gamma}\mid Y,X}(Y_{\gamma} = 1\mid Y=y, X=x_{j}),\label{eq_counterfactual_choice_probability}
\end{align}
for some $j \in \{1,\ldots,m\}$. All the results in this section are immediately applicable to the case when we wish to bound some linear function of these counterfactual probabilities. 
From Theorem \ref{theorem_infinite_to_finite} our counterfactual probability can be written as:
\begin{align}
P_{Y_{\gamma}\mid Y,X}(Y_{\gamma} = 1\mid Y=y, X=x_{j}) = \sum_{s \in S_{\gamma(j)}}  P_{U\mid Y, X}\left( \text{int}(\mathcal{U}(s,\theta)) \mid Y=y, X=x_{j}\right),
\end{align}
where $\gamma(j)$ is the index in $\{1,\ldots,m\}$ assigned to $j$ under counterfactual $\gamma$. Define the parameter:
\begin{align}
\pi(y,x,s,\theta)=P_{U\mid Y,X}\left( \text{int}(\mathcal{U}(s,\theta)) \mid Y=y,X=x\right).\label{eq_nu_parameter_vector}
\end{align}
Furthermore, define:
\begin{align}
\pi(y,s,\theta) &:= 
\begin{bmatrix}
\pi(y,x_{1},s,\theta) &
\pi(y,x_{2},s,\theta) &
\ldots &
\pi(y,x_{m},s,\theta)
\end{bmatrix}^\top,
\end{align}
\begin{align}
\pi(y,\theta):=  \begin{bmatrix}
\pi(y,s_{1},\theta)^\top & \pi(y,s_{2},\theta)^\top & \ldots & \pi(y,s_{2^{m}},\theta)^\top
\end{bmatrix}^\top, &&\pi(\theta):=  \begin{bmatrix} \pi(0,\theta)^\top & \pi(1,\theta)^\top\end{bmatrix}^\top.
\end{align}
The vector of parameters $\pi(\theta)$ represents the variable over which we optimize in our result ahead, and has dimension $d_{\pi} = 2m2^m$. Without loss of generality, we suppose that each $(y,x)$ is assigned positive probability by the observed distribution. From conditions \eqref{eq_finite_existence1} and \eqref{eq_finite_existence2} in Theorem \ref{theorem_infinite_to_finite}, we have the constraints:
\begin{align}
\sum_{s \in S_{j}}  \pi(1,x_{j},s,\theta)=1,&&\sum_{s \in S_{j}^{c}}  \pi(0,x_{j},s,\theta)=1,\label{eq_moment_conditions_finite_case}
\end{align}
for $j=1,\ldots,m$. Finally, we require the constraints:
\begin{align}
\pi(y,x_{j},s,\theta) \in \begin{cases}
	\{0\}, &\text{ if }\text{int}(\mathcal{U}(s,\theta))=\emptyset,\\
	[0,1], &\text{ otherwise},
\end{cases}\label{eq_non_negative}
\end{align}
for all $y\in\{0,1\}$ and $j=1,\ldots,m$ and $s \in \{0,1\}^{m}$, and:
\begin{align}
\sum_{s \in \{0,1\}^{m}} \pi(y,x_{j},s,\theta) =1,\label{eq_adding_up}
\end{align}
for all $y \in \{0,1\}$ and $j=1,\ldots,m$. 
Note that to impose the constraints from \eqref{eq_non_negative}, the researcher must first determine which sets $\mathcal{U}(s,\theta)$ have nonempty interior. We return to this point in the next subsection. We are now ready to state one of the main results. 
\begin{theorem}\label{thm_linprog_finite_case}
Under Assumptions \ref{assumption_basic} and \ref{assumption_counterfactual_domain}, the identified set for the counterfactual conditional probability $P_{Y_{\gamma}\mid Y,X}(Y_{\gamma} = 1\mid Y=y, X=x_{j})$ is given by:
\begin{align}
\bigcup_{\theta \in \Theta} [\pi_{\ell b}(y,x_{j},\theta),\pi_{u b}(y,x_{j},\theta)], \label{eq_thm_linprog_finite_union}
\end{align}
where $\pi_{\ell b}(y,x_{j},\theta)$ and $\pi_{u b}(y,x_{j},\theta)$ are determined by the optimization problems:
\begin{align}
\pi_{\ell b}(y,x_{j},\theta) &:= \min_{\pi(\theta)\in \mathbb{R}^{d_{\pi}}} \sum_{s \in S_{\gamma(j)}} \pi(y,x_{j},s,\theta), \text{ subject to \eqref{eq_moment_conditions_finite_case}, \eqref{eq_non_negative}, and \eqref{eq_adding_up},}\label{eq_thm_linprog_finite_LB}\\
\pi_{u b}(y,x_{j},\theta) &:=  \max_{\pi(\theta)\in \mathbb{R}^{d_{\pi}}} \sum_{s \in S_{\gamma(j)}} \pi(y,x_{j},s,\theta), \text{ subject to \eqref{eq_moment_conditions_finite_case}, \eqref{eq_non_negative}, and \eqref{eq_adding_up}. }\label{eq_thm_linprog_finite_UB}
\end{align}
\end{theorem} 
\begin{remark}
If $\Theta^{*}$ is the identified set for $\theta$, then the interval $[\pi_{\ell b}(y,x_{j},\theta),\pi_{u b}(y,x_{j},\theta)]$ is empty for all values of $\theta \in \Theta\setminus \Theta^{*}$. Thus, it is equivalent to take the union in \eqref{eq_thm_linprog_finite_union} over $\Theta$ or $\Theta^{*}$. The statement of Theorem \ref{thm_linprog_finite_case} does not assume the identified set $\Theta^{*}$ is known. 
\end{remark}
In one direction, Theorem \ref{thm_linprog_finite_case} implies that any counterfactual conditional probability of the form \eqref{eq_counterfactual_choice_probability} belonging to the identified set can be written as:
\begin{align}
P_{Y_{\gamma}\mid Y,X}\left(Y_{\gamma} = 1\mid Y=y, X=x_{j}\right) = \sum_{s \in S_{\gamma(j)}} \pi(y,x_{j},\theta,s),
\end{align}
for some $\theta$ and some vector $\pi(\theta)$ satisfying the constraints \eqref{eq_moment_conditions_finite_case}, \eqref{eq_non_negative}, and \eqref{eq_adding_up}. In the opposite direction, the Theorem implies that if for some $\theta$ the vector $\pi(\theta)$ satisfies the constraints \eqref{eq_moment_conditions_finite_case}, \eqref{eq_non_negative}, and \eqref{eq_adding_up} then the conditional probability measure on $\mathcal{U}$ represented by $\pi(\theta)$ can be extended to a (not necessarily unique) Borel probability measure on all of $\mathfrak{B}(\mathcal{U})$ that satisfies the conditions of Theorem \ref{thm_counterfactual_identified_set}. This result can be easily modified to bound any linear function of counterfactual conditional probabilities by simply modifying the objective function in Theorem \ref{thm_linprog_finite_case}. We make use of this fact in the application section. 

After determining which of the sets $\mathcal{U}(s,\theta)$ have nonempty interior, all the constraints in \eqref{eq_thm_linprog_finite_LB} and \eqref{eq_thm_linprog_finite_UB} can be written as linear equality/inequality constraints, so that the optimization problems in \eqref{eq_thm_linprog_finite_LB} and \eqref{eq_thm_linprog_finite_UB} are linear programming problems. This is very beneficial, since linear programs can be efficiently solved even in cases with thousands of parameters and constraints. 

Interestingly, the proofs for Theorems \ref{theorem_infinite_to_finite} and \ref{thm_linprog_finite_case} do not require linearity of the index function in $U$ from Assumption \ref{assumption_basic}, although this assumption is used starting in the next subsection. This implies that Theorem \ref{thm_linprog_finite_case} can be used to bound counterfactual parameters for models of the form:
\begin{align}
Y = \mathbbm{1}\{\varphi(X,U,\theta) \geq 0 \},\label{eq_nonparametric_case}
\end{align}
without imposing any restrictions on the index function. Without any restrictions, it is always possible to construct a function $\varphi$ such that all regions $\mathcal{U}(s,\theta)$ have nonempty interior, implying that there are $2^m$ response types. In this case, constraint \eqref{eq_non_negative} only imposes that all probabilities are bounded between zero and one, and the rest of Theorem \ref{thm_linprog_finite_case} remains unchanged. We illustrate our procedure using the model \eqref{eq_nonparametric_case} in the application section. 

In the general case, elements of $\pi(\theta)$ corresponding to sets $\mathcal{U}(s,\theta)$ with empty interior can be removed from the parameter vector $\pi(\theta)$ without altering the optimal solutions to the linear programs in \eqref{eq_thm_linprog_finite_LB} and \eqref{eq_thm_linprog_finite_UB}. This allows for further reduction of the dimension of these linear programs. 
Although the number of sets $\mathcal{U}(s,\theta)$ appear to grow exponentially in $m$, in the subsections ahead we show that under linearity of the index function in $U$ the number of sets $\mathcal{U}(s,\theta)$ that have nonempty interior grows at a rate that is polynomial in $m$, substantially reducing the computational burden. 

\subsection{Hyperplane Arrangements and Cell Enumeration}

To use Theorem \ref{thm_linprog_finite_case}, we must first determine which set $\mathcal{U}(s,\theta)$ have nonempty interior in order to impose the constraints from \eqref{eq_non_negative}. This section explains how this can be done in our setting by using linearity of the index function in the latent variables in combination with the cell enumeration algorithm of \cite{gu2020nonparametric}.

First we demonstrate that linearity of the index function in the latent variables imposes restrictions on the model by limiting the number of sets $\mathcal{U}(s,\theta)$ that can be assigned positive probability. Constraining sets of the form $\mathcal{U}(s,\theta)$ to be assigned zero probability is called \textit{eliminating response types}. Response types corresponding to sets $\mathcal{U}(s,\theta)$ that survive elimination are called \textit{admissible}, and response types corresponding to sets $\mathcal{U}(s,\theta)$ that are eliminated are called \textit{inadmissible}. The following simple example shows how we can eliminate response types under Assumption \ref{assumption_basic}.  

\begin{example}\label{example_linear}
Suppose we have a variable $X\in\{0.5,1,2\}$ and latent variables $U \in \mathbb{R}^{2}$. Assume there are no fixed coefficients $\theta$. Then the index function from \eqref{eq_model} can be written as $\varphi(X,U)$ and the binary response vector $r(u,\theta)$ can be written as $r(u)$, where:
\begin{align*}
r(u) = 
\begin{bmatrix}
\mathbbm{1}\{\varphi(0.5,u) \geq 0 \}\\
\mathbbm{1}\{\varphi(1,u) \geq 0 \}\\
\mathbbm{1}\{\varphi(2,u) \geq 0 \}
\end{bmatrix}.
\end{align*}
Without any additional restrictions on $\varphi$ there are a total of $2^{|\mathcal{X}|}=8$ possible response types.\footnote{For example, if $U=(U_{1},U_{2})$ take $\varphi(X,U)=\sin(U_{1}X+U_{2})$ and fix $U_{2}=0$. Then it is straightforward to find eight values of the parameter $U_{1} \in [-1,1]$ to rationalize each of the $8$ response types.} That is, $r(u) \in \{s_{1}, \ldots,s_{8}\}$, where:
\begin{align*}
s_{1} = \begin{bmatrix}
0\\
0\\
0
\end{bmatrix}, &&s_{2} = \begin{bmatrix}
1\\
0\\
0
\end{bmatrix},&&s_{3} = \begin{bmatrix}
0\\
1\\
0
\end{bmatrix},&&s_{4} = \begin{bmatrix}
1\\
1\\
0
\end{bmatrix},&&s_{5} = \begin{bmatrix}
0\\
0\\
1
\end{bmatrix}, &&s_{6} = \begin{bmatrix}
1\\
0\\
1
\end{bmatrix},&&s_{7} = \begin{bmatrix}
0\\
1\\
1
\end{bmatrix},&&s_{8} = \begin{bmatrix}
1\\
1\\
1
\end{bmatrix}.
\end{align*}
Now suppose instead that the structural function from \eqref{eq_model} can be written as:
\begin{align}
\varphi(X,U) = XU_{1} - U_{2}.  \label{eq_example_functional_form}
\end{align}
Then the binary response vector $r(u_{1},u_{2})$ is given by:
\begin{align*}
r(u_{1},u_{2}) = 
\begin{bmatrix}
\mathbbm{1}\{u_{1} \geq 2u_{2} \}\\
\mathbbm{1}\{u_{1} \geq u_{2} \}\\
\mathbbm{1}\{2u_{1} \geq u_{2} \}
\end{bmatrix}.
\end{align*}
As is illustrated in Figure \ref{fig_example_linear_arrangement}, now only $6$ response types are admissible. For a distribution of the latent variables to be admissible in this context, it must assign zero probability to the sets:
\begin{align*}
\mathcal{U}(\theta,s_{3}) &= \left\{(u_{1},u_{2}) : r(u_{1},u_{2}) = s_{3} \right\},\\
\mathcal{U}(\theta,s_{6})&=\left\{(u_{1},u_{2}) : r(u_{1},u_{2}) = s_{6} \right\}.
\end{align*}
These additional constraints must be imposed in our optimization problems from Theorem \ref{thm_linprog_finite_case}. 
\end{example}   

\begin{figure}[!t]
\centering
\includegraphics[scale=0.7]{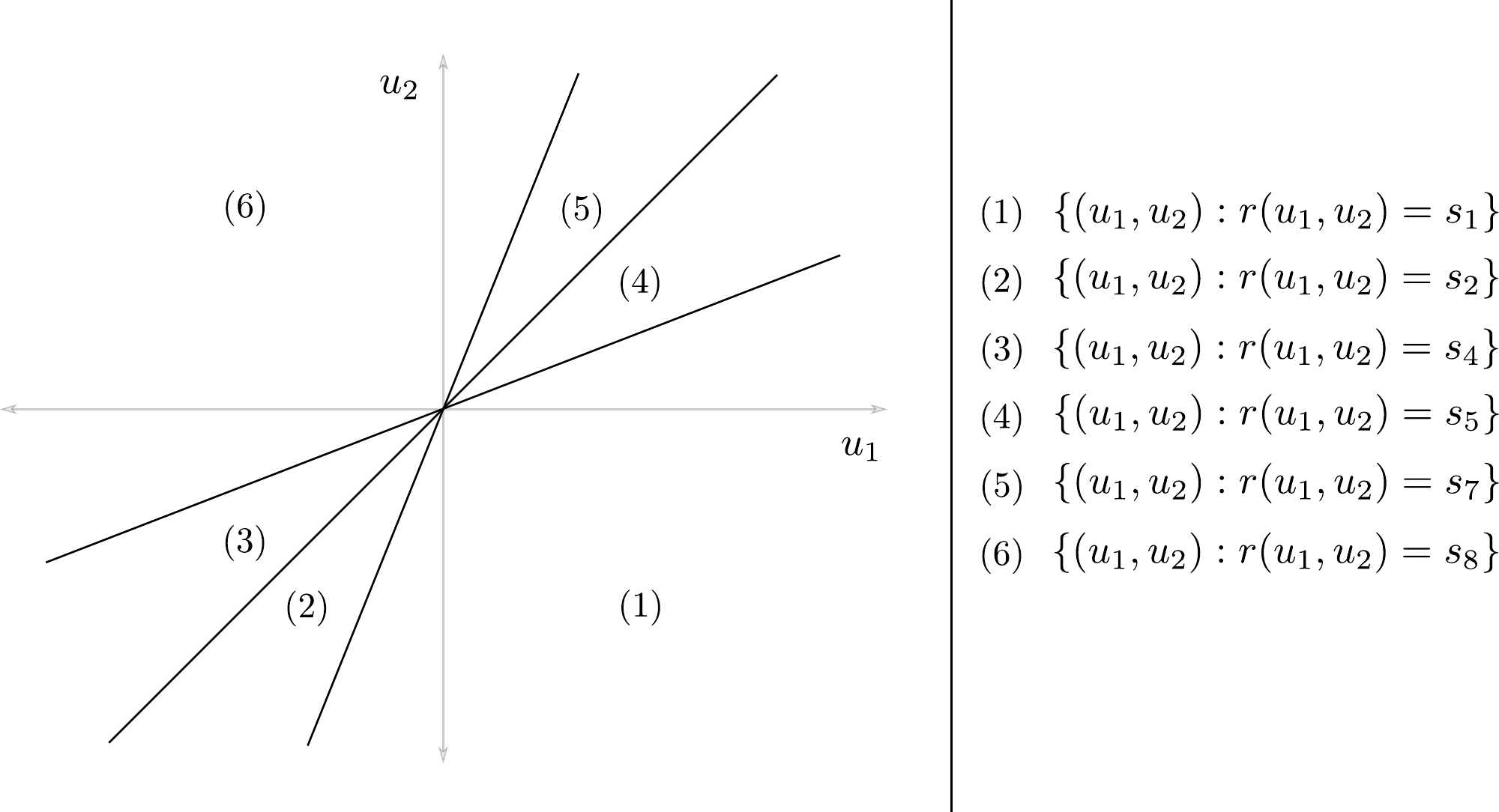}
\caption{A figure corresponding to Example \ref{example_linear} illustrating the partition of the latent variable space according to response types in the case when the index function is linear. Without functional form restrictions, Example \ref{example_linear} shows $8$ response types are possible; however, when the index function is linear in the latent variables there are only $6$ possible response types, as illustrated in the figure. In particular, the response types corresponding to binary vectors $s_{3}$ and $s_{6}$ from Example \ref{example_linear} are not possible.} \label{fig_example_linear_arrangement}
\end{figure}

In the general case, it can be shown that when $\varphi$ is restricted to be linear in $U$, there is an upper bound on the number of nonempty sets $\mathcal{U}(s,\theta)$ that grows at a rate that is polynomial in $m$ rather than exponential, which is the case when $\varphi$ is unrestricted.  
\begin{proposition}\label{proposition_vc}
Suppose that Assumption \ref{assumption_basic} is satisfied. Then for each $\theta \in \Theta$, there are at most $\sum_{j=0}^{d_{u}} \binom{m}{j}$ admissible response types.
\end{proposition}

This result is implied by results in the literature on combinatorial geometry. In particular, linearity of the function $\varphi(\,\cdot\,,u)$ means that for each instance of $(x,\theta)$ the function $\varphi(x,u,\theta)$ defines a hyperplane in $d_{u}-$dimensional space. In the case when the vectors defining these hyperplanes are in \textit{general position} the upper bound in Proposition \ref{proposition_vc} is obtained.\footnote{ A collection of $m$ hyperplanes in $d-$dimensional space are considered to be in general position when any collection of $k$ out of the $m$ hyperplanes intersect in a $d-k$ dimensional space for $1 < k \leq d$, and any collection of $k$ out of $m$ hyperplanes has an empty intersection for $k > d$.} This latter result was initially proven by \cite{Buck}. Straightforward calculation shows that, for $m>d_{u}+1$, the function $\sum_{j=0}^{d_{u}} \binom{m}{j}$ is bounded above by $(e \cdot m/d_{u})^{d_{u}} = O(m^{d_{u}})$, a polynomial in the number of hyperplanes.  

\begin{remark}
We conjecture that proposition \ref{proposition_vc} is a special case of a more general result. In particular, define the class of functions:
\begin{align}
\Phi_{\theta} := \left\{ \phi:\mathcal{X} \to \mathbb{R} : \phi(x):= \phi(x,u,\theta),   u \in \mathcal{U}\right\}.\label{eq_phi_class}
\end{align}
Now define the function:
\begin{align*}
\Pi_{\Phi_{\theta}}(m):=\sup_{x_{1},x_{2},\ldots,x_{m}} \left| \left\{\begin{bmatrix} \mathbbm{1}\{\phi(x_{1}) \geq 0 \} \\ \mathbbm{1}\{\phi(x_{2}) \geq 0 \} \\ \vdots \\ \mathbbm{1}\{\phi(x_{m}) \geq 0 \}\end{bmatrix} : \phi \in \Phi_{\theta} \right\} \right|.
\end{align*}
If $\mathcal{X}=\{x_{1},\ldots,x_{m}\}$, then $\Pi_{\Phi_{\theta}}(m)$ delivers the largest (over all inputs $\mathcal{X}$ with $|\mathcal{X}|=m$) number of admissible response types consistent with the class of index functions $\Phi_{\theta}$. The Shelah-Sauer Lemma (c.f. \cite{anthony1999neural} Theorem 3.6) then shows that:
\begin{align*}
\Pi_{\Phi_{\theta}}(m) \leq \sum_{i=0}^{v} \binom{m}{i}, 
\end{align*}
where $v$ is the Vapnik-Chervonenkis (VC) dimension of $\Phi_{\theta}$.\footnote{The definition of the VC dimension varies depending on the reference; for instance, the definition of the VC dimension (index) in \cite{van1996weak} is $+1$ larger than the definition of the VC dimension in either \cite{anthony1999neural} or \cite{mohri2012foundations}. See the beginning of Section 3.3 in \cite{anthony1999neural} for the definition of the VC dimension used here.  } In the special case when each function $\phi(x,u,\theta)$ from $\Phi_{\theta}$ is linear in $u \in \mathbb{R}^{d_{u}}$, the VC-dimension of $\Phi_{\theta}$ is $d_{u}$, suggesting Proposition \ref{proposition_vc} is a special case of the Shelah-Sauer Lemma.\footnote{C.f. Theorem 3.4 in \cite{anthony1999neural} (note Theorem 3.4 in \cite{anthony1999neural} is applicable to affine functions rather than linear functions, which increases the VC dimension by 1 relative to our case).}
\end{remark}
Define the collection:
\begin{align}
S_{\varphi}(\theta) := \{ s \in \{0,1\}^{m} : \text{int}(\mathcal{U}(s,\theta)) \neq \emptyset \}.\label{eq_s_theta}
\end{align}
To impose linearity in the latent variables we must determine which sets $\mathcal{U}(s,\theta)$ have nonempty interior, and then ensure that any distribution of the latent variables assigns zero probability to these sets. To compute the collection $S_{\varphi}(\theta)$ we propose to use the enumeration algorithm of \cite{gu2020nonparametric}. 

When the index function $\varphi$ is linear in $U$, for each fixed $\theta$ and $s \in \{0,1\}^{m}$ the set $\mathcal{U}(s,\theta)$ is a convex polyhedron formed by the intersection of at most $m$ halfspaces whose boundaries are hyperplanes of the form $\{u \in \mathcal{U} : \varphi(x,u,\theta)=0\}$. The hyperplane arrangement algorithm of \cite{gu2020nonparametric} accepts $m$ hyperplanes as an input, and outputs the binary vectors $s$ corresponding to the sets $\mathcal{U}(s,\theta)$ that have nonempty interior, as well as a point from each of these sets. \cite{avis1996reverse} were the first to provide an enumeration algorithm that runs in a time proportional to the maximum number of sets with nonempty interior. Improvements to the enumeration algorithm were made by \cite{sleumer1999output} and \cite{rada2018new}. The algorithm of \cite{gu2020nonparametric} was developed for the problem of nonparametric maximum likelihood in a linear index random coefficient model. It runs in a time proportional to $O(m^{d_{u}})$, which is near-optimal when the hyperplanes are in general position, according to Proposition \ref{proposition_vc}. 

To understand the algorithm, note that for each $s \in \{0, 1\}^m$ and fixed $\theta$, we can verify using a linear program whether there exists a point in the interior of $\mathcal{U}(s,\theta)$. Consider the following problem:
\begin{align}
\max_{u,\varepsilon} \,\,\varepsilon \qquad \text{ s.t. } \qquad (2 s_j -1) \left(\tilde{\varphi}_{1}(x,\theta)^\top u + \tilde{\varphi}_{2}(x,\theta)\right)  \geq \varepsilon, \quad j = 1, \dots, m,\label{eq_linprog_feasible}
\end{align}
where $s_{j}$ is the $j^{th}$ element of our fixed binary vector $s$. If $\varepsilon^{*}$ and $u^{*}$ are the optimal values of the program \eqref{eq_linprog_feasible}, then a value $\varepsilon^{*}>0$ indicates that $u^{*}$ is an interior point to the polyhedron $\mathcal{U}(s,\theta)$. 

Since the linear program \eqref{eq_linprog_feasible} must be solved for each $s \in \{0, 1\}^m$, checking whether each $\mathcal{U}(s,\theta)$ has nonempty interior requires solving $2^{m}$ linear programs, despite the fact that we know the number of nonempty subsets $\mathcal{U}(s,\theta)$ is polynomial in $m$. To address this issue, the algorithm proposed in \cite{gu2020nonparametric} builds upon the algorithm in \cite{rada2018new}. The idea is to add one hyperplane at a time. At step $k$ we start with a collection of $k-1$ hyperplanes from the previous steps, as well as all existing response types found up to step $k-1$. We then introduce a new hyperplane into the arrangement, and determine all newly created response types by solving a collection of linear programs. When a new hyperplane is added the only new cells are those that are created when the existing cells are crossed by the last hyperplane. By efficiently locating those crossed cells at each step, all cells in the full arrangement can be enumerated in polynomial time.

In summary, the hyperplane arrangement algorithm is used as a pre-processing step under Assumption \ref{assumption_basic} to determine which sets $\mathcal{U}(s,\theta)$ have nonempty interior in a given application. Eliminating the inadmissible sets $\mathcal{U}(s,\theta)$ can also dramatically reduce the dimension of the parameter vector $\pi(\theta)$ in the bounding optimization problems.  In particular, under Assumption \ref{assumption_basic} we need only consider a parameter vector $\pi(\theta)$ with typical element $\pi(y,x,s,\theta)$ defined for $s$ corresponding to subsets $\mathcal{U}(s,\theta)$ with nonempty interior. The dimension of the revised parameter vector $\pi(\theta)$ constructed in this way is always upper-bounded by a polynomial in $m$ under Assumption \ref{assumption_basic}. 
In the next subsection we show how the assumption of linearity in parameters $\theta \in \Theta$ can be combined with the hyperplane arrangement algorithm to dramatically simplify the bounding procedure suggested by Theorem \ref{thm_linprog_finite_case}. 

\subsection{Profiling Under Linearity in the Fixed Coefficients}\label{section_profiling}

Constructing bounds on counterfactual probabilities using Theorem \ref{thm_linprog_finite_case} requires evaluating the linear programs \eqref{eq_thm_linprog_finite_LB} and \eqref{eq_thm_linprog_finite_UB} at all values of $\theta \in \Theta$ in the parameter space. In practice this procedure is infeasible, and instead the identified set must be constructed using Theorem \ref{thm_linprog_finite_case} by establishing a grid over $\Theta$. The following proposition demonstrates that, theoretically speaking, the researcher need only consider a finite grid. 
\begin{proposition}\label{proposiTion_fInite_B}
Suppose that Assumptions \ref{assumption_basic} and \ref{assumption_counterfactual_domain} hold. Then there exists a (not necessarily unique) finite subset $\Theta' \subset \Theta$ such that:
\begin{align*}
&\left\{ \overline{\pi} \in \mathbb{R}^{d_{\pi}} : \exists \theta \in \Theta \text{ s.t. } \pi(\theta) \text{ satisfies  \eqref{eq_moment_conditions_finite_case}, \eqref{eq_non_negative}, \eqref{eq_adding_up}, and }\overline{\pi}=\pi(\theta) \right\}\\
&\qquad\qquad\qquad= \left\{ \overline{\pi} \in \mathbb{R}^{d_{\pi}} : \exists \theta \in \Theta' \text{ s.t. } \pi(\theta) \text{ satisfies  \eqref{eq_moment_conditions_finite_case}, \eqref{eq_non_negative}, \eqref{eq_adding_up}, and }\overline{\pi}=\pi(\theta) \right\}.
\end{align*}
\end{proposition}

We call the points in $\Theta'$ the \textit{representative points}, although these points are generally not unique. If the representative points can be determined by the researcher, Proposition \ref{proposiTion_fInite_B} implies that the union over $\theta \in \Theta$ in \eqref{eq_thm_linprog_finite_union} can be replaced with a union over $\theta \in \Theta'$. That is, the linear programs in \eqref{eq_thm_linprog_finite_LB} and \eqref{eq_thm_linprog_finite_UB} need only be solved at the representative points.\footnote{See the discussion at the end of Section 3 in \cite{torgovitsky2019partial} for the same idea. Proposition \ref{proposiTion_fInite_B} also implies that the identified set for counterfactual conditional distributions in Theorem \ref{thm_linprog_finite_case} is always a closed (but possibly disconnected) set.} 

In the case when $\varphi$ is also linear in $\theta$, we provide a polynomial-time algorithm for finding a collection of representative points.\footnote{Linearity in $\theta$ is not necessarily restrictive. Since $X$ is discrete, the researcher can construct a saturated specification with an index function $\sum_{x \in \mathcal{X}} \mathbbm{1}\{X=x\}\theta_{x} + U$. Defining $\theta :=(\theta_{x})_{x \in \mathcal{X}}$, this index function is linear in $(\theta,U)$, and so all of our computational results are applicable.} To introduce the approach, recall the set $S_{\varphi}(\theta)$ from \eqref{eq_s_theta}. Note that for any two values of $\theta, \theta' \in \Theta$ with $\theta \neq \theta'$, if $S_{\varphi}(\theta) = S_{\varphi}(\theta')$ then the linear programming problems in Theorem \ref{thm_linprog_finite_case} at $\theta$ and $\theta'$ are identical, since they have an identical set of constraints. The points $\theta$ and $\theta'$ are thus equivalent in the sense that we only need to solve the linear programming problems for one of them. Extending this idea, we can define an equivalence class by the set of all $\theta \in \Theta$ delivering the same collection $S_{\varphi}(\theta)$. We then only need to solve the linear programming problems at one value of $\theta$ belonging to each equivalence class. These values of $\theta$ selected from each equivalence class are exactly what we call representative points.

To see how to find the representative points, partition $x = (x_{r},x_{f})$ and suppose the index function is of the form:
\begin{align*}
\varphi(x,u,\theta) = x_{r}^\top u +  x_{f}^\top \theta,
\end{align*}
where $x_{r}$ has dimension $d_{u}$ and $x_{f}$ has dimension $d_{f}=d_{x} - d_{u}$. Here $x_{r}$ is a subvector of $x$ associated with a random coefficient $u \in \mathcal{U}$, and $x_{f}$ is a subvector of $x$ associated with a fixed coefficient $\theta \in \Theta$. For any binary vector $s \in \{ 0,1\}^{m}$, define:
\begin{align}
\mathcal{R}(s) := \left\{ (u,\theta) : \begin{bmatrix} \mathbbm{1}\{x_{1r}^\top u +  x_{1f}^\top \theta  \geq 0\}\\  \mathbbm{1}\{x_{2r}^\top u +  x_{2f}^\top \theta  \geq 0\}\\ \vdots\\ \mathbbm{1}\{x_{mr}^\top u + x_{mf}^\top \theta  \geq 0\} \end{bmatrix} = s \right\}.\label{eq_cell}
\end{align}
These sets form a unique partition of the space $\mathcal{U}\times \Theta$ into cells defined by $m$ hyperplanes of the form:
\begin{align}
x_{ir}^\top u + x_{if}^\top \theta = 0. \label{eq_hyperplane}
\end{align}
To find the representative points, we project the sets $\mathcal{R}(s)$ onto the parameter space $\Theta$, and intersect the projections of each set $\mathcal{R}(s)$ across $s$. This produces a collection of sets on $\Theta$ corresponding exactly to the equivalence classes discussed above. The main challenge of the procedure is to obtain a tractable characterization of the projected sets. 

Define the set:
\begin{align*}
S_{\varphi} := \{ s \in \{0,1\}^{m} : \exists \theta \in \Theta \text{ s.t. }\text{int}(\mathcal{U}(s,\theta)) \neq \emptyset \}.
\end{align*}
That is, $S_{\varphi}$ collects all vectors $s \in \{0,1\}^{m}$ corresponding to sets $\mathcal{R}(s)$ with nonempty interior. We begin by determining all vectors $s \in S_{\varphi}$ by running the hyperplane arrangement algorithm of \cite{gu2020nonparametric} on the $m$ hyperplanes of the form \eqref{eq_hyperplane}. 
These $m$ hyperplanes can be combined as $X_{r} u + X_{f} \theta =0$, where $X_{r}$ is an $m \times d_{u}$ matrix and $X_{f}$ is an $m\times d_{f}$ matrix. Now fix any $s \in S_{\varphi}$, and let $D(s) = \text{diag}(2s-1)$ denote the $m\times m$ diagonal matrix with the sign vector $2s-1$ along its main diagonal. Define $X_{r}(s) := D(s) X_{r}$ and $X_{f}(s):= D(s) X_{f}$. Then $\mathcal{R}(s)$ can be rewritten as: 
\begin{align}
\mathcal{R}(s) = \{ (u,\theta) : X_{r}(s)u + X_{f}(s)\theta \geq 0 \}.
\end{align}
The projection of $\mathcal{R}(s)$ on $\Theta$ is:
\begin{align}
\Theta(s) := \{ \theta \in \Theta : \exists u \in \mathcal{U} \text{ s.t. } X_{r}(s)u + X_{f}(s)\theta \geq 0\}.\label{eq_projected_beta_set}
\end{align}
The objective is now to write $\Theta(s)$ only in terms of linear inequality constraints in $\theta$; in other words, to ``eliminate'' the latent variables $u$ from the system of inequalities in \eqref{eq_projected_beta_set}.\footnote{It is possible to use Fourier-Motzkin elimination for this purpose, which was explored in a similar context in Section 8.2 of \cite{chesher2019generalized}. In practice we find that Fourier-Motzkin elimination leads to a prohibitively large number of constraints defining \eqref{eq_projected_beta_set}, most of which are redundant. This remains true even when the number of inequalities defining the set \eqref{eq_projected_beta_set} is small.}  
To this end, consider the set:
\begin{align}
\mathcal{C}(s):= \{ c \in \mathbb{R}^{m}: c^\top X_{r}(s) =0,\,\,c\geq 0\}. \label{eq_projected_Hrepresentation} 
\end{align}
\cite{kohler1967projections} showed that the projected set \eqref{eq_projected_beta_set} can be rewritten as: 
\begin{align*}
\Theta(s)=\left\{ \theta \in \Theta : c^\top  X_{r}(s)u + c^\top  X_{f}(s)\theta \geq 0, \text{ $\forall c \in \mathcal{C}(s)$}\right\}.
\end{align*}
Furthermore, the Minkowski-Weyl Theorem also allows us to re-write the set $\mathcal{C}(s)$ as:
\begin{align}
\mathcal{C}(s)= \left\{ c \in \mathbb{R}^{m}: c= R(s)a,\text{ for some $a\geq 0$} \right\},  \label{eq_projected_Vrepresentation} 
\end{align}
where $R(s)$ is some matrix.\footnote{For a general convex set defined by $\Lambda = \{\lambda \in \mathbb{R}^d: A\lambda \leq b\}$, the Minkowski-Weyl Theorem states that every vector $\lambda \in \Lambda$ can be written as $\lambda = \lambda_{1}+\lambda_{2}$, where $\lambda_{1} \in \text{conv} \{ v_1, \dots, v_k\}$ and $\lambda_{2} \in \text{cone}\{v_{k+1}, \dots, v_n\}$. Here $v_1, \dots, v_k$ are called vertices of $\Lambda$ and $v_{k+1}, \dots, v_n$ are the extreme rays of $\Lambda$. In the special case of $b = 0$, where all hyperplanes are through the origin, then $\Lambda$ becomes a polyhedral cone and $k=0$, so that $\Lambda = \text{cone}\{v_1, \dots, v_n\}$. This latter case is what is relevant for us, and the columns of the matrix $R(s)$ are the collections of these extreme rays.  } That is, every element belonging to the polyhedral cone $\mathcal{C}(s)$ can be written as a nonnegative linear combination of the columns of $R(s)$. The matrix $R(s)$ is called the \textit{generating matrix} of the polyhedral cone $\mathcal{C}(s)$, and the problem of finding the minimal generating matrix $R(s)$ is called the \textit{extreme ray enumeration problem}.\footnote{A minimal generating matrix for $\mathcal{C}(s)$ is a generating matrix with the property that no proper submatrix also generates $\mathcal{C}(s)$ (c.f. \cite{Fukuda}). Note the minimal generating matrix is unique only up to multiplication by a positive scalar.} After obtaining the matrix $R(s)$, we have the following representation of the projection of $\mathcal{R}(s)$:
\begin{align}
\Theta(s)=\left\{ \theta \in \Theta : R(s)^\top X_f(s) \theta \geq 0\right\}.\label{eq_H_representation}
\end{align}
The two representations of the cone $\mathcal{C}(s)$ from \eqref{eq_projected_Hrepresentation} and \eqref{eq_projected_Vrepresentation} are called its H-representation and its V-representation, respectively. Converting from one representation of a convex polyhedron to another is called the \textit{double description problem} in computational geometry. An efficient double description algorithm was proposed by \cite{Fukuda}, and we use the \texttt{R} implementation in the package \texttt{Rcdd} by \cite{Geyer}. \cite{avis1997good} provides a comparison of different algorithms. The projection procedure is then repeated for all $s \in S_{\varphi}$, resulting in the collection $\{ \Theta(s) : s \in S_{\varphi}\}$. To obtain the final representative points, we stack $R(s)^\top X_f(s) \theta = 0$ for all $s \in S_{\varphi}$, and then run the enumeration algorithm of \cite{gu2020nonparametric} on this final collection of hyperplanes. 

This procedure also sheds light on how to construct the identified set for $\theta$. Since our final arrangement involves only hyperplanes through the origin, each representative point is selected from a polyhedral cone. For some of these representative points, the linear programming problems in our bounding procedure from Theorem \ref{thm_linprog_finite_case} may have an empty feasible region. In these case, the representative points---and all points belonging the corresponding cone---lie outside of the identified set $\Theta^{*}$. Thus, the identified set $\Theta^*$ is the union of the polyhedral cones with representative points associated with nonempty feasible regions in our linear programming problems from Theorem \ref{thm_linprog_finite_case}.\footnote{This implies that the identified set $\Theta^*$ may not be connected, and for any $\theta \in \Theta^*$, we also have $\lambda \theta \in \Theta^*$ for all $\lambda \geq 0$.  An appropriate normalization---for example, fixing $||\theta||=1$---leads to a bounded identified set $\Theta^*$.}

\section{Additional Assumptions}\label{section_additional_assumptions}

In this section we describe how to impose additional independence and monotonicity assumptions in our framework. 

\subsection{Independence Assumptions}

In some cases the researcher may have access to an observed variable that is independent of the latent variables. If such a variable enters as an argument in the index function, it induces variation in the observed conditional probabilities without affecting the distribution of the latent variables. We refer to such variables as \textit{exogenous covariates}. A similar intuition applies if the variable is independent of the distribution of latent variables, does not enter as an argument in the structural function, but has nontrivial dependence with the variables that do enter the index function. We refer to such variables as \textit{instruments}. Both exogenous covariates and instruments can be used to produce a smaller identified set for counterfactual parameters. To introduce our independence assumption, we partition $X=(W,Z)$, where $W \in \mathcal{W} \subset \mathbb{R}^{d_{w}}$ and $Z \in \mathcal{Z} \subset \mathbb{R}^{d_{z}}$ and $d_{w}+d_{z}=d_{x}$. 

\begin{assumption}[Independence]\label{assumption_independence}
For all $A \in \mathfrak{B}(\mathcal{U})$ we have $P_{U\mid Z}(A\mid Z=z) = P_{U}(A)$, $P_{Z}-$a.s.
\end{assumption} 
The independence assumption constrains the set of admissible latent variable distributions, and links the conditional distributions of $U \mid Z=z$ across values of $z \in \mathcal{Z}$. Although Assumption \ref{assumption_independence} posits full independence between $Z$ and the vector of latent variables $U$, the assumption can be easily modified for the case when a subvector of $Z$, say $Z_{1}$, is conditionally independent of $U$ given some other subvector of $Z$, say $Z_{2}$.  


Definition \ref{definition_identified_set_independence} in Appendix \ref{subappendix_independence} provides the extension of Definition \ref{definition_identified_set} to the case when Assumption \ref{assumption_independence} also holds, and Corollary \ref{corollary_infinite_to_finite_independence} in Appendix \ref{subappendix_independence} provides the analogous extension of Theorem \ref{theorem_infinite_to_finite}. To extend the linear programming result of Theorem \ref{thm_linprog_finite_case} we must include additional constraints. Without loss of generality we assume all values of $(y,w,z)$ are assigned positive probability. Then the appropriate constraints on the vector $\pi(\theta)$ are given by:
\begin{align}
&\sum_{y \in \{0,1\}} \sum_{w \in \mathcal{W}} \pi(y,w,z_{k},s,\theta)P(Y=y,W=w \mid Z=z_{k})\nonumber\\
&\qquad\qquad\qquad\qquad= \sum_{y \in \{0,1\}}\sum_{w \in \mathcal{W}} \pi(y,w,z_{k+1},s,\theta)P(Y=y,W=w \mid Z=z_{k+1}),\label{eq_independence_constraint}
\end{align}
for $k=1,\ldots,m_{z}-1$, where $m_{z} = |\mathcal{Z}|$. A formal statement of the extension of Theorem \ref{thm_linprog_finite_case} to the case when the constraints \eqref{eq_independence_constraint} are also imposed is provided by Corollary \ref{corollary_linprog_finite_case_independence} in Appendix \ref{subappendix_independence}. 


\subsection{Monotonicity Assumptions}

Let $\mathcal{M} \subset \{1,\ldots,m\}\times \{1,\ldots,m\}$ denote any collection of integer tuples $(j,k)$, where $1 \leq j,k \leq m$. 

\begin{assumption}[Monotonicity]\label{assumption_monotonicity}
For each $\theta \in \Theta$ and each tuple $(j,k)$ in the set $\mathcal{M}$, we have $\varphi(x_{j},u,\theta) \leq \varphi(x_{k},u,\theta)$.
\end{assumption} 
This monotonicity assumption states that, when comparing two points $x_{j}$ and $x_{k}$, the value of the structural function can be ordered by the researcher. For instance, if the order determined by the researcher's monotonicity assumption for the points $x_{j}$ and $x_{k}$ is $\varphi(x_{j},u,\theta) \leq \varphi(x_{k},u,\theta)$, then the researcher automatically rules out response types with $\mathbbm{1}\{\varphi(x_{j},u,\theta) \geq 0\} > \mathbbm{1}\{\varphi(x_{k},u,\theta)\geq 0\}$. 

\begin{example}\label{example_monotonicity}
Suppose again that we have a binary variable $X\in\{0,1\}$ a latent variable $U$, and no fixed coefficients. Then the structural function from \eqref{eq_model} can be written as $\varphi(X,U)$ and the binary response vector $r(u,\theta)$ can be written as $r(u)$, where:
\begin{align*}
r(u) = 
\begin{bmatrix}
\mathbbm{1}\{\varphi(0,u) \geq 0 \}\\
\mathbbm{1}\{\varphi(1,u) \geq 0 \}
\end{bmatrix}.
\end{align*}
Note that there are only four response types; that is, $r(u) \in \{s_{1}, s_{2}, s_{3}, s_{4}\}$ where:
\begin{align*}
s_{1} = \begin{bmatrix}
1\\
1\\
\end{bmatrix}, &&s_{2} = \begin{bmatrix}
1\\
0\\
\end{bmatrix},&&s_{3} = \begin{bmatrix}
0\\
1\\
\end{bmatrix},&&s_{4} = \begin{bmatrix}
0\\
0\\
\end{bmatrix}.
\end{align*}
Without any additional restrictions, all response types---and thus all sets of the form $\mathcal{U}(s,\theta)$ for $s \in \{0,1\}^{2}$---can be assigned positive probability by the optimization problems in Theorem \ref{thm_linprog_finite_case}. Now suppose we entertain the monotonicity assumption $\varphi(0,u) \leq \varphi(1,u)$. Imposing this constraint rules out $r(u) = s_{2}$, so the set $\mathcal{U}(\theta, s_{2}) = \{u : r(u) = s_{2}\}$ must be assigned probability zero in any solution to the optimization problems in Theorem \ref{thm_linprog_finite_case}. \
\end{example}   

Similar monotonicity assumptions in triangular systems have been extensively explored by \cite{heckman2018unordered}. In particular, \cite{heckman2018unordered} explore how choice theory can be used to impose monotonicity assumptions and to eliminate response types, and many of their insights are applicable here. 

Let $S_{M}$ be the collection of vectors $s \in \{0,1\}^{m}$ that respect the monotonicity relations from Assumption \ref{assumption_monotonicity}. Definition \ref{definition_identified_set_monotonicity} in Appendix \ref{subappendix_monotonicity} provides the extension of Definition \ref{definition_identified_set} to the case when Assumption \ref{assumption_monotonicity} is also imposed.  The extension of Theorem \ref{theorem_infinite_to_finite} to the case when Assumption \ref{assumption_monotonicity} is imposed is provided by Corollary \ref{corollary_infinite_to_finite_monotonicity} in Appendix \ref{subappendix_monotonicity}. To extend the results of Theorem \ref{thm_linprog_finite_case} we must simply include the set of constraints imposed by Assumption \ref{assumption_monotonicity} in our optimization problems. These constraints are provided in Corollary \ref{corollary_infinite_to_finite_monotonicity}, and can be written in terms of the parameter vector $\pi(\theta)$ as:
\begin{align}
\sum_{s \in S_{M}^{c}} \pi(y,x_{j},s,\theta) = 0, \label{eq_monotonicity_constraint}
\end{align}
for all $y\in \{0,1\}$ and $j=1,\ldots,m$ occurring with positive probability. Corollary \ref{corollary_linprog_finite_case_monotonicity} in Appendix \ref{subappendix_monotonicity} then shows the extension of Theorem \ref{thm_linprog_finite_case} to the case when Assumption \ref{assumption_monotonicity} is imposed using the constraints \eqref{eq_monotonicity_constraint}. 


\section{Consistency, Inference, and Bias Correction}\label{section_consistency_bias_correction_inference}

The previous section focused on identification and computation issues, and set up the main optimization procedure for bounding counterfactual quantities. In this section we briefly describe estimation and inference in our setting. Our inference procedure is adapted from the procedure of \cite{cho2021simple}. 

In most applications of our procedure, we expected both the constraints and objective function in the linear programming problems \eqref{eq_thm_linprog_finite_LB} and \eqref{eq_thm_linprog_finite_UB} to be data dependent, and our results in this section are designed with this case in mind. We make use of this in the application section. Although motivated by the previous sections, the results in this section do not require any of our previous assumptions to hold. In this sense this section is independent of the previous sections, and may be of interest to researchers working with a similar class of problems. 

Let $\mathscr{P}$ denote the set of all probability measures on $\mathcal{Y}\times \mathcal{X}$, and let $\Pi:=\{\pi \in \mathbb{R}^{d_{\pi}} : A \pi \leq b\}$ denote a compact and convex polytope defined by some matrix $A$ and some vector $b$. In our context, $\Pi$ represents the parameter space for the conditional probability vector $\pi$ from Theorem \ref{thm_linprog_finite_case}. To introduce our results, we convert each of the equality constraints in the linear programs from Theorem \ref{thm_linprog_finite_case} to two equivalent inequality constraints. Then for a fixed $\theta \in \Theta$, the constraints in the programs \eqref{eq_thm_linprog_finite_LB} and \eqref{eq_thm_linprog_finite_UB} can be written as a finite number of moment inequality constraints:
\begin{align*}
\E_{P}[m_{j}(Y_{i},X_{i},\pi,\theta)] \leq 0, \,\,\forall j \in \mathcal{J}(\theta),  
\end{align*}
where $\mathcal{J}(\theta)$ is an index set that may depend on $\theta$.\footnote{These moment functions should also include parameter space constraints.} Note this formulation also allows the use of data-dependent constraints in the programs \eqref{eq_thm_linprog_finite_LB} and \eqref{eq_thm_linprog_finite_UB}, so long as these constraints can be written as a collection of moment inequalities. Furthermore, the objective function in the programs \eqref{eq_thm_linprog_finite_LB} and \eqref{eq_thm_linprog_finite_UB} can be written as $\E_{P}[\psi(Y_{i},X_{i},\pi,\theta)]$, which also permits the use of a data-dependent objective function.  We require the following assumption. 
\begin{assumption}\label{assumption_consistency_in_text}
The parameter space $(\Theta, \Pi,\mathcal{P})$ satisfies the following: (i) for each $\theta \in \Theta$, the function $\psi(\,\cdot\,,\theta): \mathcal{Y}\times\mathcal{X} \times \Pi \to \mathbb{R}$ is measurable in $(Y,X)$ and linear in $\pi$ with a (possibly data-dependent) Lipschitz constant $C(\theta)$ satisfying $\sup_{\theta \in \Theta} C(\theta)<\infty$ a.s. (ii) For each $\theta \in \Theta$ the set $\mathcal{J}(\theta)$ is finite, and for each $j\in \mathcal{J}(\theta)$ the functions $m_{j}(\,\cdot\,,\theta): \mathcal{Y}\times \mathcal{X} \times \Pi \to \mathbb{R}$ are measurable in $(Y,X)$ and linear in $\pi \in \Pi$. (iii) For every $P \in \mathcal{P}\subset \mathscr{P}$ there exists $(\pi,\theta) \in \Pi\times \Theta$ such that $\E_{P}[m_{j}(Y,X,\pi,\theta)] \leq 0$ for all $j\in \mathcal{J}(\theta)$. (iv) There exists a $B<\infty$ and a value $\delta>0$ such that:
\begin{align*}
\sup_{P \in \mathcal{P}} \int ||(Y,X)||^{2+\delta} \,dP \leq B. 
\end{align*}
(v) There exists some constant $C<\infty$ such that, for each $\theta \in \Theta$:
\begin{align*}
\sup_{P \in \mathcal{P}} \max\left\{\E_{P}||\nabla_{\pi} \psi(Y,X,\pi,\theta)||^{2},\max_{j \in \mathcal{J}(\theta)} \E_{P}||\nabla_{\pi} m_{j}(Y,X,\pi,\theta)||^{2} \right\} \leq C; 
\end{align*}
(vi) There exists a finite subset $\Theta' \subset \Theta$ such that:
\begin{align*}
&\left\{ \pi \in \Pi : \exists \theta \in \Theta \text{ s.t. } \E_{P}[m_{j}(Y,X,\pi,\theta)] \leq 0 \text{ for $j\in \mathcal{J}(\theta)$} \right\}\\
&\qquad\qquad\qquad\qquad\qquad= \left\{ \pi \in \Pi : \exists \theta \in \Theta' \text{ s.t. } \E_{P}[m_{j}(Y,X,\pi,\theta)] \leq 0 \text{ for $j\in \mathcal{J}(\theta)$} \right\}. 
\end{align*}
(vii) For each $\theta \in \Theta'$ and for some $\varepsilon>0$, define the vector-valued class of functions:
\begin{align*}
\mathcal{F}(\theta)&:= \left\{\left(\psi(\,\cdot\,,\pi,\theta),\left( m_{j}(\,\cdot\,,\pi,\theta)\right)_{j \in \mathcal{J}(\theta)} \right) : \pi \in \Pi_{\varepsilon} \right\},\\
\Pi_{\varepsilon} &:= \left\{\pi \in \mathbb{R}^{d_{\pi}} : A \pi \leq b + \varepsilon  \right\},
\end{align*}
where $A\pi \leq b$ are the constraints defining $\Pi$. Then there exists an element-wise measurable envelope $F_{\theta}:\mathcal{Y}\times \mathcal{X} \to \mathbb{R}$ for $\mathcal{F}(\theta)$ that is bounded on $\mathcal{Y}\times \mathcal{X}$. 
(viii) The researcher has a sample $\{(Y_{i},X_{i})\}_{i=1}^{n}$, with each $(Y_{i},X_{i})$ an i.i.d. draw from some $P \in \mathcal{P}$.
\end{assumption}
Part (i) restricts the objective function to be Lipschitz continuous in $\pi$ for each $\theta$. This assumption is satisfied, for instance, when bounding counterfactual probabilities or average treatment effects, as in the application section. Part (ii) requires each moment function to be linear in $\pi$, an assumption which is also satisfied for the model and assumptions in this paper. Assumption (iii) is standard in the moment inequalities literature, and constrains $\mathcal{P}$ to be the collection of data generating processes that satisfy the moment conditions. Part (iv) imposes a uniform integrability requirement on the vector $(Y,X)$ needed for the procedure of \cite{cho2021simple}, and part (v) imposes a uniform integrability condition on the gradients. These are both easily satisfied with discrete random variables $Y$ and $X$. Part (vi) allows us to replace $\Theta$ with a finite subset $\Theta'$ without impacting the bounding problem. Proposition \ref{proposiTion_fInite_B} shows this is the case in our setting. Part (vii) assumes the existence of a bounded envelope function on a slight expansion of $\Pi$. This is also easily verified, for instance, when bounding counterfactual probabilities or average treatment effects. Finally part (viii) assumes that the sample under consideration is i.i.d. from some $P \in \mathcal{P}$ satisfying the other conditions. 

Using these assumptions we will prove a consistency result, and demonstrate how to do inference in our class of problems. For any $c \in \mathbb{R}_{+}$ and any $P \in \mathscr{P}$ let us define:
\begin{align*}
\Pi^{*}(\theta,P,c)&:= \left\{ \pi \in \Pi : \E_{P}[m_{j}(Y_{i},X_{i},\pi,\theta)] \leq c \text{ $\forall j \in \mathcal{J}(\theta)$} \right\},&&\Psi^{*}(P,c) := \bigcup_{\theta \in \Theta'} [\Psi_{\ell b}(\theta,P,c), \Psi_{u b}(\theta,P,c)],
\end{align*}
where:
\begin{align*}
\Psi_{\ell b}(\theta,P,c):= \min_{\pi \in \Pi^{*}(\theta,P,c)} \E_{P}[\psi(Y_{i},X_{i},\pi,\theta)], && \Psi_{u b}(\theta,P,c):= \max_{\pi \in \Pi^{*}(\theta,P,c)} \E_{P}[\psi(Y_{i},X_{i},\pi,\theta)].
\end{align*}
The value functions $\Psi_{\ell b}(\theta,P,0)$ and $\Psi_{u b}(\theta,P,0)$ are analogous to the value functions $\pi_{\ell b}(y,x,\theta)$ and $\pi_{u b}(y,x,\theta)$ from Theorem \ref{thm_linprog_finite_case}. Furthermore, let $\Psi^{*}(P):=\Psi^{*}(P,0)$, and denote the empirical measure as $\mathbb{P}_{n} \in \mathscr{P}$. The following theorem shows that a slight (shrinking) enlargement of the set $\Psi^{*}(\P_{n})$ is a consistent estimator for the set $\Psi^{*}(P)$, where consistency is defined using the Hausdorff metric.\footnote{For two sets $A, B \subset \mathbb{R}^{d}$, the Hausdorff metric is defined as:
\begin{align*}
d_{H}(A,B) := \max\left\{\sup_{a \in A}\inf_{b \in B} ||a-b||, \sup_{b \in B}\inf_{a \in A} ||a-b|| \right\}. 
\end{align*}
  }
\begin{proposition}\label{proposition_consistency_in_text}
Suppose that Assumption \ref{assumption_consistency_in_text} holds. Then $d_{H}(\Psi^{*}(\P_{n},b_{n}),\Psi^{*}(P)) = o_{P}(1)$, where $b_{n}$ is any positive user-specified sequence satisfying $b_{n} = O(1/\sqrt{\log(n)})$.
\end{proposition}
Proposition \ref{proposition_consistency_in_text} is related to results and discussions found in \cite{molchanov1998limit}, \cite{manski2002inference}, and \cite{chernozhukov2007estimation}. It is also a special case of a more general consistency result presented in Appendix \ref{appendix_consistency}, and suggests a consistent estimator for the identified set from Theorem \ref{thm_linprog_finite_case}. In the application section we take $b_{n} = b/\sqrt{\log(n)}$ for some $b>0$, and call $\Psi^{*}(\P_{n},b_{n})$ the ``plug-in'' estimate of the identified set. 

In Section \ref{section_application} we also use the inference method of \cite{cho2021simple}, designed for uniform inference on value functions in stochastic linear programming problems. In general, this inference problem is highly irregular, but \cite{cho2021simple} show that regularity can be restored by introducing infinitesimal random perturbations to the constraints and objective function. After perturbing the problem, they prove consistency of a simple and fast nonparametric bootstrap procedure to construct a confidence set. However, a modification of the procedure of \cite{cho2021simple} is needed to fit our setting to allow for the profiling points $\theta \in \Theta'$. 

Let $\xi\sim P_{\xi}$ denote the random perturbation vector from the procedure of \cite{cho2021simple}. By Proposition \ref{proposiTion_fInite_B}, there exists a finite set $\Theta' \subset \Theta$ of representative points satisfying:
\begin{align*}
\Psi^{*}(P) = \bigcup_{\theta \in \Theta'} \Psi^{*}(\theta,P,0).
\end{align*}
To apply the procedure of \cite{cho2021simple}, we take each $\theta \in \Theta'$, and use the procedure of \cite{cho2021simple} to construct a set $C_{n}(1-\alpha,\theta)$ satisfying:
\begin{align}
\liminf_{n\to \infty}  \inf_{\{(\psi,P): \psi\in \Psi^{*}(\theta,P,0), P \in \mathcal{P}\}} (\text{Pr}_{P}\times P_{\xi}) \left(\psi \in CS_{n}(1-\alpha,\theta)  \right) \geq 1-\alpha.\label{eq_sub_CS}
\end{align}
Here, probability is taken with respect to the product measure $\text{Pr}_{P}\times P_{\xi}$ to account for the random perturbations $\xi \sim P_{\xi}$ introduced to restore regularity. We refer to \cite{cho2021simple} for additional discussion. We then set:
\begin{align*}
CS_{n}(1-\alpha) = \bigcup_{\theta \in \Theta'} CS_{n}(1-\alpha,\theta),
\end{align*}
The following result shows that this confidence set is uniformly valid over $\mathcal{P}$. The proof of the result proceeds by verifying the assumptions of \cite{cho2021simple} for \eqref{eq_sub_CS}, and then showing that the modified confidence set $CS_{n}(1-\alpha)$ has the correct coverage. 
\begin{proposition}\label{proposition_inference_in_text}
Suppose that Assumption \ref{assumption_consistency_in_text} holds. Then:
\begin{align}
\liminf_{n\to \infty}  \inf_{\{(\psi,P): \psi\in \Psi^{*}(P), P \in \mathcal{P}\}} (\text{Pr}_{P}\times P_{\xi}) \left(\psi \in CS_{n}(1-\alpha)  \right) \geq 1-\alpha.
\end{align}
\end{proposition}
In our application ahead, we also report bias-corrected estimates of the lower and upper endpoints of the identified set. Convexity of the minimum in combination with Jensen's inequality shows that the sample analog lower bound is biased upward. Similarly, the sample analog upper bound is biased downward. This leads to an identified set that is on average too narrow. In response, \cite{chernozhukov2013intersection} proposed the use of half-median unbiased estimators. Half-median unbiased estimates $\hat{\Psi}_{\ell b}$ and $\hat{\Psi}_{ub}$ satisfy $\hat{\Psi}_{\ell b} \leq \Psi_{\ell b}(P)$ and $\Psi_{ub}(P)\leq \hat{\Psi}_{u b}$, each holding with probability at least $1/2$. We construct a bias-corrected estimate of $\Psi_{\ell b}(P)$ ($\Psi_{u b}(P)$) using the inference procedure of \cite{cho2021simple} by setting $\alpha=1/2$ and by taking our estimate to be the endpoint of the $1-\alpha$ lower (upper) confidence set for $\Psi_{\ell b}(P)$ ($\Psi_{u b}(P)$). In our application we report both the plug-in estimates of our bounds based on Theorem \ref{proposition_consistency_in_text} and the half-median unbiased estimates.

\section{Application}\label{section_application}

In this section we apply our method to study the impact of private health insurance on an individual's decision to visit a doctor. In general, insurance markets are plagued by problems arising from asymmetric information between consumers and insurance providers (c.f. \cite{rothschild1978equilibrium}). For example, adverse selection occurs in the health insurance market when individuals have more information about their latent health determinants than the providers of health insurance. A robust prediction of the classical theory of asymmetric information is that those who are more likely to purchase insurance are also those who are more likely to experience the insured risk.\footnote{The ``insured risk'' refers to the event for which insurance was purchased. In our context, it is any event that would typically require a visit to the doctor.} On the other hand, there has been little and mixed empirical evidence of adverse selection in health insurance markets (see \cite{cardon2001asymmetric} for a discussion). 

In this section we compute various counterfactual parameters while remaining agnostic on the exact nature of the latent variables linking health insurance and health care utilization decisions. We take the decision to visit a doctor as our binary outcome variable of interest, and consider the individuals' private health insurance status as an endogenous explanatory variable. This is consistent with the idea that private insurance status may be dependent with individual-specific latent factors---most importantly, unobserved health determinants and attitudes towards risk---that influence an individual's propensity to visit a doctor. We use data from the 2010 wave of the Medical Expenditure Panel Survey (MEPS), which has also been recently analyzed by \cite{han2019estimation} and \cite{acerenza2021testing}. We focus on the same sub-sample considered in these papers. In particular, we focus on the month of January 2010, consider only individuals between ages $25$ and $64$, and drop individuals who obtain either federal or state insurance in 2010 and individuals who are self-employed or unemployed. These restrictions leave us with a sample of $7555$ individuals. 

In all specifications $W$ is a binary endogenous variable representing an individual's private insurance status, and we consider a binary health status variable ($Z_{1}$) and a binary marital status variable ($Z_{2}$) as regressors.\footnote{The MEPS data includes information on self-reported health status on a scale from $1- 5$, and we consider values less than or equal to $2$ as being ``unhealthy.''} Finally, we use the number of employees working for the individual's firm ($Z_{3}$) as an instrument. This variable provides a measure of the size of a firm and has discrete support in the range $[1,500]$, which we further discretize into 11 bins.\footnote{Variable $Z_3$ is supported on the range $[1, 500]$ and is clearly top-coded. We notice that there is bunching of observations at firm sizes in multiples of five, and some regions of the support of $Z_3$ contain very few observations. In order to get reliable estimates of the conditional choice probabilities, we further discretize the firm size into 11 bins. The bins are respectively $[1, 5]$, $(5, 10]$, $(10,20]$, $(20,30]$, $(30,40]$, $(40,50]$, $(50,60]$, $(60,70]$, $(70,100]$, $(100,200]$ and $(200,500]$.} Using firm size as an instrument is consistent with the evidence that larger firms are more likely to provide health insurance benefits, but do not directly influence an individual's decision to visit a doctor.\footnote{From \cite{cardon2001asymmetric} p.408: ``Another observed symptom, consistent with the theoretical predictions, is that the uninsured tend to work for small employers. Large employers can overcome adverse selection by risk pooling.'' }  

A possible concern with using firm size as an instrument is that risk averse individuals may be more likely to select into a job with a larger firm size. In an attempt to address this issue, we include an alternate independence assumption that assumes the firm size $Z_3$ is conditionally independent of $U$ given $(Z_1, Z_2)$ only when $Z_{3}$ lies within a certain range. The idea is that once we condition on a particular range of firm size, the remaining variation in firm size is independent of $U$ conditional on $(Z_1, Z_2)$. We consider four ranges, given by $(1, 10]$, $(10, 50]$, $(50, 100]$ and $(100,500]$, and impose our conditional independence assumption for each range separately.

The first parameter we consider is the average treatment effect, defined as:
\begin{align}
	\mu_{ate} & := \sum_{(y, w, z) \in \{0,1\} \times \mathcal{W} \times \mathcal{Z}}P_{U\mid Y, W, Z}(\varphi(1, z, U, \theta) \geq 0\mid  Y = y, W = w, Z = z) P(Y =y, W = w, Z = z) \label{eq_mu_ate}\\ 
	&\qquad\qquad- \sum_{(y, w, z) \in \{0,1\} \times \mathcal{W} \times \mathcal{Z}}P_{U\mid Y, W, Z}(\varphi(0, z, U, \theta) \geq 0\mid  Y = y, W = w, Z = z) P(Y =y, W = w, Z = z). \nonumber
\end{align}
This parameter provides the average causal effect of obtaining health insurance on the decision to visit a doctor. Second, we consider the counterfactual choice probability: 
\begin{align*}
\mu_{ccp}(y)  & := \sum_{z \in \mathcal{Z}} P_{U\mid Y,W,Z}(\varphi(1, z, u, \theta) \geq 0 \mid  Y = y, W = 0, Z = z) P(Z = z\mid Y = y, W = 0),
\end{align*}
for $y \in \{0,1\}$. We focus on the parameter $\mu_{ccp}(0)$ for simplicity, which represents the counterfactual choice probability of visiting a doctor when given private health insurance for the set of individuals who have no insurance and who have chosen not to visit a doctor, averaged across health and marital status. We construct our bounds under the following set of assumptions: 
\begin{enumerate}[label=(A\arabic*)]
	\item Only Assumptions \ref{assumption_basic} and \ref{assumption_counterfactual_domain}.
	\item (A1) and monotonicity (Assumption \ref{assumption_monotonicity}). See below for further details.
	\item (A1) and independence between $(Z_{1},Z_{2})$ and $U$ (Assumption \ref{assumption_independence}).
	\item (A1), (A2) and (A3) together. 
	\item (A1) and independence between $(Z_{1},Z_{2}, Z_{3})$ and $U$ (Assumption \ref{assumption_independence}). 
	\item (A1), (A2) and (A5) together. 
	\item (A1) and $U \independent (Z_{1}, Z_{2},Z_{3}) \mid Z_{3} \in [a,b]$ for various intervals $[a,b]$ (Assumption \ref{assumption_independence}). See the discussion above. 
	\item (A1), (A2) and (A7) together. 
\end{enumerate}
Note that the general index function takes the form $\varphi(w,z_{1},z_{2},u,\theta)$. When monotonicity is imposed in (A2), we impose:
\begin{align*}
	\varphi(1, 0, z_{2},u,\theta) \geq \varphi(0, 0, z_{2},u,\theta), 
\end{align*}
for each $z_{2} \in \{0,1\}$. This implies that for an unhealthy individual, the propensity to visit a doctor when the person has private insurance is always weakly greater than without insurance, regardless of marital status. Finally we consider three different models for the binary outcome variable $Y$:
\begin{align}
	Y &= 1\{ \varphi(W, Z_{1},Z_{2}, U)\geq 0\},\tag{M1} \label{eq_app_nonseparable}\\
	Y &= 1 \{ W U_{1} + Z_{1}\theta_{1}+Z_{2}\theta_{2} \geq U_{2}\},\tag{M2} \label{eq_app_linear1}\\
	Y &= 1 \{ W \theta_{1} + Z_{1}\theta_{2} + Z_{2}\theta_{3} \geq U\}.\tag{M3} \label{eq_app_linear2}
\end{align}
Recall that the extension of our procedure to cover model \eqref{eq_app_nonseparable} was discussed briefly at the end of Section \ref{section_optimization_formulation}. Indeed, under model \eqref{eq_app_nonseparable} the index function $\varphi$ need not be explicitly specified and it may not satisfy the linearity assumption made under Assumption \eqref{assumption_basic}. This makes model \eqref{eq_app_nonseparable} the most flexible.  Models \eqref{eq_app_linear1} and \eqref{eq_app_linear2} impose linearity of $\varphi$ in the latent variables and in the parameters. Here we distinguish two cases. In the first case, \eqref{eq_app_linear1} regards $(U_{1},U_{2})$ as the latent variables in the model. Model \eqref{eq_app_linear2} is the same as \eqref{eq_app_linear1} except that we have replaced the random slope coefficient $U_{1}$ from \eqref{eq_app_linear1} with a fixed coefficient. Model \eqref{eq_app_linear2} represents the additively separable linear index model that is commonly used in the empirical literature, except for the fact that we do not assume a parametric distribution for $U$ and do not have a model for the endogenous variable $W$. 

The identified sets for $\mu_{ate}$ under assumptions (A1) - (A8) and models \eqref{eq_app_nonseparable} - \eqref{eq_app_linear2} are reported in Table \ref{table_ate_results_cv}. For simplicity, we report the convex hull of the estimated identified set for each specification. Table \ref{table_ate_results_cv} also reports our modified plug-in estimator (see Section \ref{section_consistency_bias_correction_inference}) as well as half-median unbiased estimators and $90\%$ confidence sets constructed using the modified procedure of \cite{cho2021simple}. Due to a confluence of factors---including the dimension of the empirical choice probability vector, the large number of constraints, and the sample size---we find that the bootstrap standard errors are small, resulting in half-median unbiased estimates that are only slightly more narrow than the 90\% confidence sets. 

Unsurprisingly, the plug-in bounds on $\mu_{ate}$ shrink as the strength of our assumptions increase. The most flexible model is \eqref{eq_app_nonseparable} under assumption (A1), in which case the length of the bound on $\mu_{ate}$ is one.\footnote{Note that in a potential outcome framework with a binary treatment and binary outcome (and no other additional assumptions), the worst-case bound on the average treatment effect always has a length of one (c.f. \cite{manski1990nonparametric}). } The identified set for $\mu_{ate}$ also always overlaps zero for model \eqref{eq_app_nonseparable}. Results in Table \ref{table_ate_results_cv} suggest that full independence of $Z_3$ (A5 and A6) produces more informative bounds for $\mu_{ate}$ than our alternate conditional independence assumption (A7 and A8). In fact, our alternate conditional independence assumption does not provide much identifying power (compare the results under Assumptions (A3) and (A7)). On the other hand, full independence of $Z_3$ does induce a noticeable narrowing of the identified set for $\mu_{ate}$ (compare the results under Assumptions (A3) and (A5)). The results for this model are a useful benchmark to compare with cases where we impose linearity on the index function.  

\begin{table}[!tbp]
	\begin{center}
		\resizebox{\textwidth}{!}{\begin{tabular}{lcccccccc}
			\hline\hline
			\multicolumn{1}{l}{}&\multicolumn{1}{c}{(A1)}&\multicolumn{1}{c}{(A2)}&\multicolumn{1}{c}{(A3)}&\multicolumn{1}{c}{(A4)}&\multicolumn{1}{c}{(A5)}&\multicolumn{1}{c}{(A6)}&\multicolumn{1}{c}{(A7)}&\multicolumn{1}{c}{(A8)}\tabularnewline
			\hline
		&	\multicolumn{6}{c}{\eqref{eq_app_nonseparable}: Nonseparability of $\varphi$}\\
		\hline 
			Plug-in&$[-0.57, 0.43]$&$[-0.37, 0.43]$&$[-0.57, 0.43]$&$ [-0.36, 0.43]$&$[-0.28, 0.26]$&$[-0.13, 0.26]$&$[-0.49, 0.38]$&$[-0.29, 0.38]$\tabularnewline
			Half-Median&$[-0.57, 0.43]$&$[-0.37, 0.43]$&$[-0.57, 0.43]$&$[-0.36, 0.43]$&$[-0.30, 0.28]$&$[-0.14, 0.27]$&$[-0.73, 0.52]$&$[-0.48, 0.52]$\tabularnewline
			90\% c.s.&$[-0.58, 0.44]$&$[-0.37, 0.44]$&$[-0.58, 0.44]$&$[-0.37, 0.44]$&$[-0.33, 0.30]$&$[-0.17, 0.30]$&$[-0.76, 0.55]$&$[-0.51, 0.54]$\tabularnewline
			 									\hline 
			& \multicolumn{6}{c}{\eqref{eq_app_linear1}: Linearity of $\varphi$ (with random coefficients)}\\
			\hline
			Plug-in&$[-0.57, 0.43]$&$[-0.37, 0.43]$&$[-0.54, 0.43]$&$ [-0.26, 0.43]$&$[-0.17, 0.25]$&$[0.17, 0.25]$&$[-0.42, 0.37]$&$[-0.18, 0.37]$\tabularnewline
			Half-Median&$[-0.57, 0.43]$&$[-0.37, 0.43]$&$[-0.54, 0.43]$&$[-0.25, 0.43]$&$[-0.20, 0.28]$&$[0.14, 0.27]$&$[-0.39, 0.36]$&$[-0.16, 0.36]$\tabularnewline
			90\% c.s.&$[-0.58, 0.44]$&$[-0.37, 0.44]$&$[-0.56, 0.44]$&$[-0.27, 0.44]$&$[-0.24, 0.31]$&$[0.11, 0.29]$&$[-0.49, 0.61]$&$[-0.19, 0.37]$\tabularnewline
			 									\hline 
			& \multicolumn{6}{c}{\eqref{eq_app_linear2}: Linearity of $\varphi$ (with fixed coefficients)}\\
			\hline 
			Plug-in&$[-0.57, 0.43]$&$[-0.37, 0.43]$&$[-0.53, 0.43]$&$ [0.00, 0.43]$&$[0.17, 0.25]$&$[0.17, 0.25]$&$[0.07, 0.37]$&$[0.07, 0.37]$\tabularnewline
			Half-Median&$[-0.57, 0.43]$&$[-0.37, 0.43]$&$[-0.54, 0.43]$&$ [0.01, 0.43]$&$[0.15, 0.27]$&$[0.15, 0.27]$&$[0.10, 0.36]$&$[0.10, 0.36]$\tabularnewline
			90\% c.s.&$[-0.58, 0.44]$&$[-0.37, 0.44]$&$[-0.55, 0.44]$&$[0.00, 0.44]$&$[0.12, 0.29]$&$[0.12, 0.29]$&$[0.07, 0.37]$&$[0.06, 0.38]$\tabularnewline
			\hline
		\end{tabular}}
		\caption{Identified sets for the average treatment effect under different specifications and under various assumptions. For the plug-in estimates, we convert all equality constraints to two inequality constraints and introduce a small slackness $b_n = 0.0001/\sqrt{\log(n)}$ which is needed for consistency (see Section \ref{section_consistency_bias_correction_inference}). Half-median unbiased estimates and a 90\% confidence set are also reported. These sets are computed using 999 bootstrap samples using the inference approach in \cite{cho2021simple}.\label{table_ate_results_cv}}\end{center}
\end{table}

Next, we see in Table \ref{table_ate_results_cv} that the linear models from \eqref{eq_app_linear1} and \eqref{eq_app_linear2} narrow the bounds relative to the case of the general index function under some of the assumptions. Unsurprisingly, the smallest interval for $\mu_{ate}$ for model \eqref{eq_app_linear1} is obtained under Assumption (A6), in which case the sign of $\mu_{ate}$ is identified. For models \eqref{eq_app_linear1} and \eqref{eq_app_linear2} we make use of our method for profiling $\theta$, as described in Section \ref{section_profiling}. In model \eqref{eq_app_linear1} we must profile on $\theta \in \mathbb{R}^2$ and there are 8 representative points. 
Interestingly, we find that under Assumptions (A1) - (A4) and (A7) - (A8), the identified set for $\theta$ is the entire Euclidean space $\mathbb{R}^2$. This illustrates that non-trivial bounds on $\mu_{ate}$ are possible even when the structural parameters are unidentified. Figure \ref{ATE_hole1} in Appendix \ref{appendix_figure} shows the intervals computed using the linear programs of the form \eqref{eq_thm_linprog_finite_LB} and \eqref{eq_thm_linprog_finite_UB} for each representative point of $\theta$ under our various assumptions. 

In the second linear model \eqref{eq_app_linear2}, all coefficients are fixed. Thus, we need to profile on a parameter vector $\theta\in \mathbb{R}^3$. Our profiling procedure from Section \ref{section_profiling} returns $96$ representative points, each associated with a polyhedral cone in $\mathbb{R}^3$. 
Under Assumptions (A1) and (A2), the identified set for $\theta$ is $\mathbb{R}^3$, while for all other assumptions (A3) - (A8) we get a more informative identified set for $\theta$. In Figure \ref{ATE_hole2} in Appendix \ref{appendix_figure} we also show the intervals computed using the linear programs of the form \eqref{eq_thm_linprog_finite_LB} and \eqref{eq_thm_linprog_finite_UB} for each representative point of $\theta$ under our various assumptions.  Interestingly, the ATE bounds under (A1) for model (M2) and (M3) are the same as those under model (M1), which suggests that the functional form restrictions only become informative when combined with other modelling assumptions. The sign of $\mu_{ate}$ is identified for model \eqref{eq_app_linear2} under assumptions (A4) - (A8). The narrowest bounds for $\mu_{ate}$ under model \eqref{eq_app_linear2} are obtained under Assumption (A6), where the $90\%$ confidence interval is $[0.12,0.29]$. For comparison, using the same data but a slightly different model, \cite{acerenza2021testing} also bound the ATE and obtain a $95\%$ confidence interval of $[0.02,0.29]$.\footnote{In addition to using a different model, \cite{acerenza2021testing} also use the number of employees (without discretization) as their instrument, and they use inference procedure of \cite{chernozhukov2013intersection} combined with a sample splitting procedure, which is valid under a very different set of assumptions than those presented in the current paper.  }

Next we consider the counterfactual choice probability $\mu_{ccp}(0)$. Table \ref{table_ccp_results_cv} reports the convex hull of the estimated identified set for $\mu_{ccp}(0)$ under various model specifications and under various assumptions. Similar to the bounds for $\mu_{ate}$, the half-median unbiased estimates are only slightly more narrow than the 90\% confidence sets. We also see that the bounds on counterfactual choice probabilities tend to be wide and uninformative for most assumptions. Note that under Assumption (A1) to (A4) we always obtain the interval $[0,1]$ for the estimated identified set. 
The narrowest bounds are found in model \eqref{eq_app_linear2} under Assumptions (A5) and (A6). These bounds allow us to conclude that the probability an individual visits a doctor when provided private health insurance, given that they have no private health insurance and did not visit a doctor increases noticeably, with a magnitude in the interval $[0.37, 0.55]$ and a 90\% confidence set of $[0.22, 0.62]$.

\begin{table}[!tbp]
	\begin{center}
		\resizebox{\textwidth}{!}{	\begin{tabular}{lcccccccc}
			\hline\hline
			\multicolumn{1}{l}{}&\multicolumn{1}{c}{(A1)}&\multicolumn{1}{c}{(A2)}&\multicolumn{1}{c}{(A3)}&\multicolumn{1}{c}{(A4)}&\multicolumn{1}{c}{(A5)}&\multicolumn{1}{c}{(A6)}&\multicolumn{1}{c}{(A7)}&\multicolumn{1}{c}{(A8)}\tabularnewline
			\hline
&	\multicolumn{6}{c}{\eqref{eq_app_nonseparable}: Nonseparability of $\varphi$}\\
\hline 
			Plug-in &$[0.00,1.00]$&$[0.00,1.00]$&$[0.00,1.00]$&$[0.00, 1.00]$&$[0.10, 0.70]$&$[0.18, 0.62]$&$[0.04, 0.95]$&$[0.07, 0.91]$\tabularnewline
			Half-Median&$[0.00,1.00]$&$[0.00,1.00]$&$[0.00,1.00]$&$[0.00, 1.00]$&$[0.06, 0.71]$&$[0.14, 0.64]$&$[0.00,1.00]$&$[0.00,1.00]$\tabularnewline
			90\% c.s.&$[0.00,1.00]$&$[0.00,1.00]$&$[0.00, 1.00]$&$[0.00, 1.00]$&$[0.00, 0.76]$&$[0.07, 0.68]$&$[0.00, 1.00]$&$[0.00, 1.00]$\tabularnewline
		 									\hline 
	& \multicolumn{6}{c}{\eqref{eq_app_linear1}: Linearity of $\varphi$ (with random coefficients)}\\
	\hline
			Plug-in &$[0.00,1.00]$&$[0.00,1.00]$&$[0.00,1.00]$&$[0.00, 1.00]$&$[0.13, 0.68]$&$[0.37, 0.56]$&$[0.04, 0.94]$&$[0.08, 0.91]$\tabularnewline
			Half-Median&$[0.00,1.00]$&$[0.00,1.00]$&$[0.00,1.00]$&$[0.00, 1.00]$&$[0.06, 0.67]$&$[0.29, 0.58]$&$[0.05, 0.93]$&$[0.10, 0.89]$\tabularnewline
			90\% c.s.&$[0.00,1.00]$&$[0.00,1.00]$&$[0.00, 1.00$&$[0.00, 1.00]$&$[0.00, 0.73]$&$[0.22, 0.63]$&$[0.00, 1.00]$&$[0.07, 0.92]$\tabularnewline
						 									\hline 
			& \multicolumn{6}{c}{\eqref{eq_app_linear2}: Linearity of $\varphi$ (with fixed coefficients)}\\
			\hline 
			Plug-in &$[0.00,1.00]$&$[0.00,1.00]$&$[0.00,1.00]$&$[0.00, 1.00]$&$[0.37, 0.55]$&$[0.37, 0.55]$&$[0.10, 0.86]$&$[0.10, 0.86]$\tabularnewline
			Half-Median&$[0.00,1.00]$&$[0.00,1.00]$&$[0.00, 1.00]$&$[0.00, 1.00]$&$[0.30, 0.58]$&$[0.30, 0.58]$&$[0.12, 0.83]$&$[0.12, 0.83]$\tabularnewline
			90\% c.s.&$[0.00,1.00]$&$[0.00,1.00]$&$[0.00,1.00]$&$[0.00, 1.00]$&$[0.22, 0.62]$&$[0.22, 0.62]$&$[0.08, 0.86]$&$[0.06, 0.87]$\tabularnewline
			\hline
		\end{tabular}}
		\caption{This table reports the convex hull of the estimated bounds on $\mu_{ccp}(0)$, the counterfactual choice probability of visiting doctor when granted insurance for those who chose not to visit a doctor without insurance. For the plug-in estimates, we convert all equality constraints into two inequality constraints and introduce a small slackness $b_n = 0.0001/\sqrt{\log(n)}$, which is needed for consistency (see Section \ref{section_consistency_bias_correction_inference}). Half-median unbiased estimates and a 90\% confidence set are also reported. These sets are computed using 999 bootstrap samples using the inference approach in \cite{cho2021simple}.\label{table_ccp_results_cv}}\end{center}
\end{table}

For the sake of comparison, we estimate the following bivariate probit:
\begin{align*}
	Y &= 1\{W \theta_{1} + Z_1 \theta_2 + Z_2 \theta_3 \geq \varepsilon_1\},\\
	W &= 1\{ Z_1 \gamma_{1} + Z_2 \gamma_{2}+ Z_{3}\gamma_{3} \geq \varepsilon_2\},
\end{align*}
where $(Z_1, Z_2, Z_3)$ are assumed to be independent from $(\varepsilon_1, \varepsilon_2)$, which are bivariate normal with mean zero, unit variance and correlation $\rho$. This model was estimated with our data using maximum likelihood, and $\mu_{ate}$ was estimated as $0.16$ with a bootstrapped confidence interval of $[0.11, 0.20]$. This value for $\mu_{ate}$ is not in the plug-in bounds under (M3) and (A5) and (A6), but it does lie within all of the 90\% confidence sets in Table \ref{table_ate_results_cv}, and seems to suggest strong evidence of a positive causal effect of health insurance on the decision to visit the doctor.\footnote{\cite{han2019estimation} obtain a similar result in a model allowing for $\varepsilon_{1}$ and $\varepsilon_{2}$ to have unrestricted marginals, and a flexible dependence structure. However, they consider a different model from us, and the average treatment effect in \cite{han2019estimation} is different from ours; we consider the average treatment effect averaged over all values of $(w,z)$, while they report the average treatment effect at the average value of their conditioning variables. They also report the average treatment effect at various quantiles of their conditioning variables.}  However, the bivariate probit model is highly parameterized, and the results from Table \ref{table_ate_results_cv} suggest that under weaker assumptions the sign of $\mu_{ate}$ may not be identified.\footnote{In fact, \cite{acerenza2021testing} reject the assumptions of the bivariate probit model using the same data but with a slightly different specification.} 

The previous literature studying the effects of health insurance on the utilization of health care services is full of mixed results, and Table \ref{table_ate_results_cv} suggests that highly parameterized models may give highly significant, but possibly misleading results relative to models that make weaker assumptions.

\section{Conclusion}\label{section_conclusion}

This paper considers (partial) identification of a variety of counterfactual parameters in binary response models with possibly endogenous regressors. Importantly, our class of models allows for nonseparability of the index function in latent variables, and does not require any parametric distributional assumptions. Our specific partition of the latent variable space is key to our procedure, and we show how to enumerate the sets in this partition using results from the literature on computational geometry and hyperplane arrangements. In doing so, we provide a feasible method of constructing bounds on counterfactual quantities under a variety of different assumptions with multi-dimensional and nonseparable latent variables. We also thoroughly study the special case when the index function is linear in parameters, and show how to compute exact (i.e. not approximate) sharp bounds on counterfactual quantities. We also show how to adapt a recent inference procedure to the setting in this paper in order to construct confidence sets and bias-corrected estimates of the identified set.  Finally, we show how to impose independence and monotonicity assumptions, and we present an application of our method to study the effects of private health insurance on the utilization of health care services. 

The consideration of multinomial choice models, triangular systems, or general simultaneous discrete choice models (e.g. games, network formation, or models of social interactions) are all natural future extensions of the framework presented here which we intend to pursue. In addition, this paper emphasizes computational issues that arise in models that are partially identified. We believe exploring applications of state-of-the-art algorithms in computer science to problems in econometrics---as we have attempted here---is a fruitful avenue of future research.

\bibliographystyle{apa}
\bibliography{bibfile}
\pagebreak
\begin{appendix}
\section{Proofs}\label{appendix_proofs}

\subsection{Proofs of Results in the Main Text}

\begin{proof}[Proof of Theorem \ref{thm_counterfactual_identified_set}]
Let $\mathcal{P}_{Y_{\gamma}\mid Y,X}^{**}$ denote the set of all conditional distributions $P_{Y_{\gamma}\mid Y,X}$ such that there exists a pair $(P_{U\mid Y,X},\theta) \in \mathcal{I}_{Y,X}^{*}$ satisfying:
\begin{align}
P_{Y_{\gamma}\mid Y,X}\left(Y_{\gamma} = 1\mid Y=y, X=x\right)= P_{U\mid Y,X}\left( \varphi(\gamma(X),U,\theta) \geq 0 \mid Y=y, X=x\right),
\end{align}
$P_{Y,X}-$a.s. To prove the result it suffices to show $\mathcal{P}_{Y_{\gamma}\mid Y, X}^{*}=\mathcal{P}_{Y_{\gamma}\mid Y, X}^{**}$. To do this, we show that $\mathcal{P}_{Y_{\gamma}\mid Y, X}^{*}\subset\mathcal{P}_{Y_{\gamma}\mid Y, X}^{**}$ and $\mathcal{P}_{Y_{\gamma}\mid Y, X}^{**}\subset\mathcal{P}_{Y_{\gamma}\mid Y,X}^{*}$. To this end, begin by fixing an arbitrary $P_{Y_{\gamma}\mid Y, X} \in \mathcal{P}_{Y_{\gamma}\mid Y, X}^{*}$. By Definition \ref{definition_identified_set_counterfactual} we have:
\begin{align}
P_{Y_{\gamma}\mid Y,X,U}\left(Y_{\gamma} = \mathbbm{1}\{\varphi(\gamma(X),U,\theta) \geq 0  \}\mid Y=y,X=x,U=u\right)=1, \label{eq_thm_cf1}
\end{align}
$P_{Y,X,U}-$a.s. for some $(P_{U\mid Y,X},\theta)\in \mathcal{I}_{Y,X}^{*}$. For this pair $(P_{U\mid Y,X},\theta)$ we have:
\begin{align*}
&P_{Y_{\gamma}\mid Y,X,U}\left(Y_{\gamma} = 1\mid Y=y,X=x,U=u\right)\\
&= P_{Y_{\gamma}\mid Y,X,U}\left(Y_{\gamma} = 1,Y_{\gamma} = \mathbbm{1}\{\varphi(\gamma(X),U,\theta) \geq 0  \}\mid Y=y,X=x, U=u\right),
\end{align*}
$P_{Y,X,U}-$a.s., which follows from \eqref{eq_thm_cf1}. Now note:
\begin{align*}
P_{Y_{\gamma}\mid Y,X,U}\left(Y_{\gamma} =1, Y_{\gamma} = \mathbbm{1}\{\varphi(\gamma(X),U,\theta) \geq 0  \}\mid Y=y,X=x,U=u\right)= \mathbbm{1}\{\varphi(\gamma(x),u,\theta) \geq 0 \},
\end{align*}
$P_{Y,X,U}-$a.s. Thus we have:
\begin{align*}
P_{Y_{\gamma}\mid Y,X}\left(Y_{\gamma} = 1\mid Y=y,X=x\right) &= \int P_{Y_{\gamma}\mid Y,X,U}\left(Y_{\gamma} = 1\mid Y=y,=x,U=u\right) \, dP_{U\mid Y,X}\\
&= \int \mathbbm{1}\{\varphi(\gamma(x),u,\theta) \geq 0 \} \, dP_{U\mid Y,X}\\
&= P_{U\mid Y,X}(\varphi(\gamma(X),U,\theta) \geq 0\mid Y=y,X=x),
\end{align*}
$P_{Y,X}-$a.s. This proves $P_{Y_{\gamma}\mid Y,X} \in \mathcal{P}_{Y_{\gamma}\mid Y,X}^{**}$, and since $P_{Y_{\gamma}\mid Y,X} \in \mathcal{P}_{Y_{\gamma}\mid Y,X}^{*}$ was arbitrary we conclude that $\mathcal{P}_{Y_{\gamma}\mid Y,X}^{*} \subset \mathcal{P}_{Y_{\gamma}\mid Y,X}^{**}$. 

For the reverse inclusion, fix any arbitrary $P_{Y_{\gamma}\mid Y,X} \in \mathcal{P}_{Y_{\gamma}\mid Y,X}^{**}$. Then by definition there exists a pair $(P_{U\mid Y,X},\theta) \in \mathcal{I}_{Y,X}^{*}$ satisfying:
\begin{align}
P_{Y_{\gamma}\mid Y,X}\left(Y_{\gamma} = 1\mid Y=y,X=x\right)= P_{U\mid Y,X}\left( \varphi(\gamma(X),U,\theta) \geq 0 \mid Y=y, X=x\right),\label{eq_cond_prob}
\end{align}
$P_{Y,X}-$a.s. It suffices to show that for this pair $(P_{U\mid Y,X},\theta)$ there exists $P_{Y_{\gamma}\mid Y,X,U}$ satisfying:
\begin{align}
P_{Y_{\gamma}\mid Y,X,U}\left(Y_{\gamma} = \mathbbm{1}\{\varphi(\gamma(X),U,\theta) \geq 0  \}\mid Y=y,X=x,U=u\right)=1,\label{eq_conclusion2}
\end{align}
$P_{Y,X,U}-$a.s. By the Radon-Nikodym Theorem, the existence of a (version of) $P_{Y_{\gamma}\mid Y,X,U}$ is guaranteed by the fact that $P_{Y_{\gamma},U\mid Y,X}\ll P_{U \mid Y, X}$ for all $(y,x)$ occurring with positive probability. Since all spaces involved are Euclidean, we can choose the version to be an almost surely unique regular conditional distribution (c.f. \cite{durrett2010probability} Theorem 5.1.9). By construction this $P_{Y_{\gamma}\mid Y,X,U}$ satisfies:
\begin{align*}
&P_{Y_{\gamma},U\mid Y,X}(Y_{\gamma} \in A, U \in B \mid Y=y,X=x)\\
&\qquad\qquad= \int_{B} P_{Y_{\gamma}\mid Y,X,U}(Y_{\gamma} \in A \mid Y=y,X=x,U=u) \, dP_{U \mid Y,X},
\end{align*}
$P_{Y,X}-$a.s. for every $A \subset \{0,1\}$ and $B \in \mathfrak{B}(\mathcal{U})$. Now note that:
\begin{align*}
P_{Y_{\gamma}\mid Y,X,U}(Y_{\gamma}=1, Y_{\gamma}= \mathbbm{1}\{\varphi(\gamma(X),U,\theta) \geq 0  \} \mid Y=y,X=x,U=u) = \mathbbm{1}\{\varphi(\gamma(x),u,\theta) \geq 0  \},\\
P_{Y_{\gamma}\mid Y,X,U}(Y_{\gamma}=0, Y_{\gamma}= \mathbbm{1}\{\varphi(\gamma(X),U,\theta) \geq 0  \} \mid Y=y,X=x,U=u) = \mathbbm{1}\{\varphi(\gamma(x),u,\theta) < 0  \}.
\end{align*}
$P_{Y,X}-$a.s. Thus:
\begin{align*}
&P_{Y_{\gamma}\mid Y,X}\left(Y_{\gamma} = \mathbbm{1}\{\varphi(\gamma(X),U,\theta) \geq 0  \}\mid Y=y,X=x\right)\\
&=  \int_{\mathcal{U}} P_{Y_{\gamma}\mid Y,X,U}(Y_{\gamma} = \mathbbm{1}\{\varphi(\gamma(X),U,\theta) \geq 0  \} \mid Y=y,X=x,U=u) \, dP_{U \mid Y,X}\\
&=  \int_{\mathcal{U}} P_{Y_{\gamma}\mid Y,X,U}(Y_{\gamma}=1, Y_{\gamma}= \mathbbm{1}\{\varphi(\gamma(X),U,\theta) \geq 0  \} \mid Y=y,X=x,U=u) \, dP_{U \mid Y,X}\\
&\qquad\qquad\qquad+  \int_{\mathcal{U}} P_{Y_{\gamma}\mid Y,X,U}(Y_{\gamma}=0, Y_{\gamma}= \mathbbm{1}\{\varphi(\gamma(X),U,\theta) \geq 0  \} \mid Y=y,X=x,U=u) \, dP_{U \mid Y,X}\\
&=  \int_{\mathcal{U}} \mathbbm{1}\{\varphi(\gamma(x),u,\theta) \geq 0  \} \, dP_{U \mid Y,X}+  \int_{\mathcal{U}} \mathbbm{1}\{\varphi(\gamma(x),u,\theta) < 0  \} \, dP_{U \mid Y,X}\\
&= P_{U\mid Y,X}(\varphi(\gamma(x),u,\theta) \geq 0 \mid Y=y, X=x) + P_{U\mid Y,X}(\varphi(\gamma(x),u,\theta) < 0 \mid Y=y, X=x)\\
&=1,
\end{align*}
$P_{Y,X}-$a.s. This proves \eqref{eq_conclusion2} and thus shows $P_{Y_{\gamma}\mid Y,X} \in \mathcal{P}_{Y_{\gamma}\mid Y,X}^{*}$. Since $P_{Y_{\gamma}\mid Y,X} \in \mathcal{P}_{Y_{\gamma}\mid Y,X}^{**}$ was arbitrary we can conclude that $\mathcal{P}_{Y_{\gamma}\mid Y,X}^{**} \subset \mathcal{P}_{Y_{\gamma}\mid Y,X}^{*}$. Combining the two inclusions, we have $\mathcal{P}_{Y_{\gamma}\mid Y,X}^{*} =\mathcal{P}_{Y_{\gamma}\mid Y,X}^{**}$. This completes the proof.

\end{proof}

\begin{proof}[Proof of Theorem \ref{theorem_infinite_to_finite}]
Let $P_{Y_{\gamma} \mid Y,X}$ be a collection of conditional distributions, and suppose there exists $(P_{U \mid Y,X},\theta) \in \mathcal{I}_{Y,X}^{*}$ satisfying \eqref{eq_theorem_counterfactual}. Since $P_{U \mid Y,X} \ll P_{U}$ for all $(y,x)$ assigned positive probability, and since $P_{U} = P_{U \mid Y,X} P_{Y,X}$ assigns zero probability to sets of the form $\{u \in \mathcal{U} : \varphi(x,u,\theta)=0\}$, \eqref{eq_finite_existence3} is equivalent to \eqref{eq_theorem_counterfactual}, so we can conclude that $(P_{U \mid Y,X},\theta)$ satisfies \eqref{eq_finite_existence3}. Furthermore, by definition $(P_{U \mid Y,X},\theta) \in \mathcal{I}_{Y,X}^{*}$ implies that:
\begin{align}
P_{U\mid Y,X}(U \in  \mathcal{U}(Y,X,\theta) \mid Y=y, X=x)=1, \,\, P_{Y,X}-a.s.,
\end{align}
Again, since $P_{U \mid Y,X} \ll P_{U}$ for all $(y,x)$ assigned positive probability, and since $P_{U}$ assigns zero probability to sets of the form $\{u \in \mathcal{U} : \varphi(x,u,\theta)=0\}$, the previous display is equivalent to conditions \eqref{eq_finite_existence1} and \eqref{eq_finite_existence2}. This shows that any pair $(P_{U \mid Y,X},\theta) \in \mathcal{I}_{Y,X}^{*}$ satisfying \eqref{eq_theorem_counterfactual} satisfies \eqref{eq_finite_existence1} - \eqref{eq_finite_existence3}. 

For the reverse, fix any $\theta\in \Theta$ and any collection $P_{U \mid Y,X}$ of conditional probability measures on the sets in $\mathcal{A}(\theta)$ satisfying \eqref{eq_finite_existence1} - \eqref{eq_finite_existence3}. We show that $P_{U \mid Y,X}$ can be extended to a (not necessarily unique) probability measure $\tilde{P}_{U \mid Y,X}$ on $\mathfrak{B}(\mathcal{U})$ in a manner that ensures $\tilde{P}_{U \mid Y,X}$ satisfies \eqref{eq_theorem_counterfactual} and such that $(\tilde{P}_{U \mid Y,X},\theta) \in \mathcal{I}_{Y,X}^{*}$. Furthermore, by the definition of an extension, $\tilde{P}_{U \mid Y,X}$ agrees with $P_{U \mid Y,X}$ on all sets in $\mathcal{A}(\theta)$. To construct the extension, for each $s \in \{0,1\}^{m}$ select a single point $u(s,\theta)$ from $\text{int}(\mathcal{U}(s,\theta))$ if $\text{int}(\mathcal{U}(s,\theta))\neq \emptyset$; otherwise choose $u(s,\theta)$ as an arbitrary point from $\mathcal{U}$. For any set $A \subset \mathcal{U}$, define the indicator:
\begin{align*}
\mathbbm{1}(A,\theta,s) = \mathbbm{1}\{  u(s,\theta)\in A\cap \text{int}(\mathcal{U}(s,\theta)) \}.
\end{align*}
Now define the function $\mu_{y,x}:\mathfrak{B}(\mathcal{U})\to \mathbb{R}$ as:
\begin{align*}
\mu_{y,x}(B) := \sum_{s \in \{0,1\}^{m}} \mathbbm{1}(B,\theta,s) P_{U\mid Y,X}\left( \text{int}(\mathcal{U}(s,\theta))\mid Y=y,X=x\right).
\end{align*}
To verify that this is a proper probability measure on $\mathfrak{B}(\mathcal{U})$, we must show that (i) $\mu_{y,x}(B) \geq \mu_{y,x}(\emptyset) =0$ for every $B \in \mathfrak{B}(\mathcal{U})$, (ii) $\mu_{y,x}(\mathcal{U})=1$, and (iii) for any countable sequence of disjoint sets $\{A_{i}\}_{i=1}^{\infty}$ in $\mathfrak{B}(\mathcal{U})$, we have:
\begin{align*}
\mu_{y,x} \left( \bigcup_{i=1}^{\infty} A_{i} \right) = \sum_{i=1}^{\infty} \mu_{y,x}(A_{i}).
\end{align*}
The first property holds since $\mathbbm{1}(\emptyset,\theta,s)=0$ for all $s$. To verify the second property, note that $\mathbbm{1}(\mathcal{U},\theta,s)=1$ for all $s$ with $\text{int}(\mathcal{U}(s,\theta))\neq \emptyset$, so that:
\begin{align*}
\mu_{y,x}(\mathcal{U}) &= \sum_{s \in \{0,1\}^{m}} \mathbbm{1}(\mathcal{U},\theta,s) P_{U\mid Y,X}\left( \text{int}(\mathcal{U}(s,\theta)) \mid Y=y,X=x\right)\\
&= \sum_{s : \text{int}(\mathcal{U}(s,\theta)) \neq \emptyset} P_{U\mid Y,X}\left( \text{int}(\mathcal{U}(s,\theta)) \mid Y=y,X=x\right)\\
&= 1,
\end{align*}
where the last line holds since $P_{U\mid Y,X}$ is a probability measure on $\mathcal{A}(\theta)$. For the third property, note that for two disjoint Borel sets $A_{1},A_{2} \in \mathfrak{B}(\mathcal{U})$ we have:
\begin{align*}
\mathbbm{1}(A_{1}\cup A_{2},\theta,s) = \mathbbm{1}(A_{1},\theta,s) + \mathbbm{1}(A_{2},\theta,s).
\end{align*}
Inducting on this formula, we conclude that for countable disjoint sets $\{A_{i}\}_{i=1}^{\infty}$ in $\mathfrak{B}(\mathcal{U})$, we have:
\begin{align*}
\mathbbm{1} \left( \bigcup_{i=1}^{\infty} A_{i}, \theta, s \right) = \sum_{i=1}^{\infty} \mathbbm{1}(A_{i},\theta,s), 
\end{align*}
Thus we can conclude:
\begin{align*}
\mu_{y,x} \left( \bigcup_{i=1}^{\infty} A_{i} \right) &= \sum_{s \in \{0,1\}^{m}} \mathbbm{1}\left(\bigcup_{i=1}^{\infty} A_{i},\theta,s\right) P_{U\mid Y,X}\left( \text{int}(\mathcal{U}(s,\theta)) \mid Y=y,X=x\right)\\
&= \sum_{s \in \{0,1\}^{m}}  \sum_{i=1}^{\infty} \mathbbm{1}(A_{i},\theta,s) P_{U\mid Y,X}\left( \text{int}(\mathcal{U}(s,\theta)) \mid Y=y,X=x\right)\\
&= \sum_{i=1}^{\infty} \sum_{s \in \{0,1\}^{m}}   \mathbbm{1}(A_{i},\theta,s) P_{U\mid Y,X}\left( \text{int}(\mathcal{U}(s,\theta)) \mid Y=y,X=x\right)\\
&= \sum_{i=1}^{\infty} \mu_{y,x}(A_{i}).
\end{align*}
Thus, our measure satisfies countable additivity. We conclude that $\mu_{y,x}$ is a proper probability measure. Note that the argument above has been completed for a single pair $(y,x)$ indexing the conditioning variables. However, we can repeat the same argument as above for all $(y,x)$ assigned positive probability, and thus can construct a corresponding probability measure $\mu_{y,x}$ satisfying all the conditions described above for each such $(y,x)$. 

Now we define $\tilde{P}_{U\mid Y,X} : \mathfrak{B}(\mathcal{U}) \to [0,1]$ by $\tilde{P}_{U\mid Y,X}(B\mid Y=y,X=x) = \mu_{y,x}(B)$ for all $B \in \mathfrak{B}(\mathcal{U})$ and all $(y,x)$ assigned positive probability. By the above, $\tilde{P}_{U\mid Y,X}(\,\cdot\,\mid Y=y,X=x)$ is a proper probability measure on $\mathfrak{B}(\mathcal{U})$ for each $(y,x)$. Also note that for any pair $(1,x)$ assigned positive probability, the pair $(\tilde{P}_{U\mid Y,X},\theta)$ satisfies:
\begin{align}
&\tilde{P}_{U\mid Y,X}(\mathcal{U}(1,x,\theta) \mid Y=1,X=x)\nonumber\\
&= \sum_{s \in S_{j}} \tilde{P}_{U\mid Y,X}(\mathcal{U}(s,\theta)\mid Y=1,X=x)\nonumber\\
&= \sum_{s \in S_{j}} \sum_{s' \in \{0,1\}^{n}} \mathbbm{1}(\mathcal{U}(s,\theta),\theta,s') P_{U\mid Y,X}\left( \text{int}(\mathcal{U}(s,\theta))\mid Y=1,X=x\right)\nonumber\\
&= \sum_{s \in S_{j}}  \mathbbm{1}(\mathcal{U}(s,\theta),\theta,s) P_{U\mid Y,X}\left( \text{int}(\mathcal{U}(s,\theta)) \mid Y=1,X=x\right)\nonumber\\
&= 1,
\end{align}
which follows from \eqref{eq_finite_existence1}. Furthermore, for any pair $(0,x)$ assigned positive probability, the pair $(\tilde{P}_{U\mid Y,X},\theta)$ also satisfies:
\begin{align}
&\tilde{P}_{U\mid Y,X}(\mathcal{U}(0,x,\theta) \mid Y=0,X=x) \nonumber\\
&= \sum_{s \in S_{j}^{c}} \tilde{P}_{U\mid Y,X}(\mathcal{U}(s,\theta)\mid Y=0,X=x)\nonumber\\
&= \sum_{s \in S_{j}^{c}} \sum_{s' \in \{0,1\}^{n}} \mathbbm{1}(\mathcal{U}(s,\theta),\theta,s') P_{U\mid Y,X}\left( \text{int}(\mathcal{U}(s,\theta)) \mid Y=0,X=x\right)\nonumber\\
&= \sum_{s \in S_{j}^{c}}  \mathbbm{1}(\mathcal{U}(s,\theta),\theta,s) P_{U\mid Y,X}\left( \text{int}(\mathcal{U}(s,\theta)) \mid Y=0,X=x\right)\nonumber\\
&= 1,
\end{align}
which follows from \eqref{eq_finite_existence2}. Conclude that:
\begin{align*}
\tilde{P}_{U\mid Y,X}(U \in \mathcal{U}(Y,X,\theta) \mid Y=y,X=x) &=1, \,\,a.s.
\end{align*}
It is also straightforward to see that $\tilde{P}_{U}:=\tilde{P}_{U\mid Y,X} P_{Y,X}$ assigns zero probability to all sets of the form $\{u \in \mathcal{U} : \varphi(x,u,\theta) = 0\}$, since these sets have empty intersection with $\text{int}(\mathcal{U}(s,\theta))$ for all $s\in \{0,1\}^{m}$. Combining everything, this shows that $(\tilde{P}_{U\mid Y, X},\theta) \in \mathcal{I}_{Y,X}^{*}$. Finally, setting $C:=\{ u \in \mathcal{U} : \varphi(\gamma(x),u,\theta) \geq 0 \}$, it is straightforward to show that:
\begin{align*}
\tilde{P}_{U\mid Y,X}\left( C \mid Y=y, X=x_{j}\right)&= \sum_{s \in \{0,1\}^{m}} \mathbbm{1}(C,\theta,s) P_{U\mid Y,X}\left( \text{int}(\mathcal{U}(s,\theta)) \mid Y=y,X=x_{j}\right)\\
&=\sum_{s \in S_{\gamma(j)}}  P_{U\mid Y,X}\left( \text{int}(\mathcal{U}(s,\theta)) \mid Y=y,X=x_{j}\right)\\
&=P_{Y_{\gamma}\mid Y,X}\left( Y_{\gamma}=1 \mid Y=y,X=x_{j}\right),
\end{align*}
for all $(y,x_{j})$ assigned positive probability, which follows from \eqref{eq_finite_existence3}. This is exactly condition \eqref{eq_theorem_counterfactual}. Conclude that $(\tilde{P}_{U\mid Y, X},\theta) \in \mathcal{I}_{Y,X}^{*}$ and that $(\tilde{P}_{U\mid Y,X},\theta)$ satisfies \eqref{eq_theorem_counterfactual}. This completes the proof.

\end{proof}

\begin{proof}[Proof of Theorem \ref{thm_linprog_finite_case}]
Note that the constraints in \eqref{eq_moment_conditions_finite_case} are equivalent to the constraints in \eqref{eq_finite_existence1} and \eqref{eq_finite_existence2}. Furthermore, the objective function in the optimization problems in Theorem \ref{thm_linprog_finite_case} enforce \eqref{eq_finite_existence3}. Thus, using Theorem \ref{theorem_infinite_to_finite}, a distribution $\pi(\theta)$ is feasible in the optimization problems from Theorem \ref{thm_linprog_finite_case} if and only if there exists a collection of Borel conditional probability measures $P_{U\mid Y,X}$ satisfying \eqref{eq_theorem_counterfactual} with $(P_{U\mid Y,X},\theta) \in \mathcal{I}_{Y,X}^{*}$. However, by Theorem \ref{thm_counterfactual_identified_set}, there exists a collection of Borel conditional probability measures $P_{U\mid Y,X}$ satisfying \eqref{eq_theorem_counterfactual} with $(P_{U\mid Y,X},\theta) \in \mathcal{I}_{Y,X}^{*}$ if and only if $P_{Y_{\gamma} \mid Y,X} \in \mathcal{P}_{Y_{\gamma} \mid Y,X}^{*}$, where $P_{Y_{\gamma} \mid Y,X}$ is the (collection of) conditional distribution(s) satisfying \eqref{eq_theorem_counterfactual}. 
\end{proof}

\begin{proof}[Proof of Proposition \ref{proposiTion_fInite_B}]
First note that $\theta \in \Theta$ enters the constraints in Theorem \ref{thm_linprog_finite_case} only through the constraints \eqref{eq_non_negative}; in particular, only through its determination of which sets $\text{int}(\mathcal{U}(s,\theta))$ are empty versus nonempty. Now define:
\begin{align*}
\mathcal{S}_{\varphi}(\theta):= \{s \in \{0,1\}^{m} : \text{int}(\mathcal{U}(s,\theta)) \neq \emptyset \}.
\end{align*}
Now define an equivalence relation $\sim$ on $\Theta$ as follows: $\theta \sim \theta'$ if and only if $\mathcal{S}_{\varphi}(\theta) = \mathcal{S}_{\varphi}(\theta')$. This equivalence relation will partition $\Theta$ into at most $2^{2^{m}}$ equivalence classes (which is the total number of ways of choosing $k$ vectors from $\{0,1\}^{m}$ for $k=0,1,\ldots,2^{m}$). Furthermore, any two values $\theta$ and $\theta'$ belonging to the same equivalence class will deliver the same values for the linear programs \eqref{eq_thm_linprog_finite_LB} and \eqref{eq_thm_linprog_finite_UB} (by construction of the equivalence class). Thus, it is sufficient to consider only one $\theta$ from each equivalence class in Theorem \ref{thm_linprog_finite_case}, showing there are at most $2^{2^{m}}$ such $\theta$'s to consider.   
\end{proof}

\begin{proof}[Proof of Proposition \ref{proposition_vc}]
This follows immediately from the results of \cite{Buck}.
\end{proof}

\begin{proof}[Proof of Proposition \ref{proposition_consistency_in_text}]
It suffices to verify the assumptions of Theorem \ref{theorem_consistency}; i.e., Assumption \ref{assumption_consistency}. 
\begin{enumerate}[label=(\roman*)]
 	\item In our context, the parameter space $\mathcal{T}$ is given by $\Pi\times \Theta$. The set $\Pi$ is a compact and convex polytope by definition. 
 	\item Part (ii) of Assumption \ref{assumption_consistency} is implied by part (i) of Assumption \ref{assumption_consistency_in_text}.
 	\item Part (iii) of Assumption \ref{assumption_consistency} is implied by part (ii) of Assumption \ref{assumption_consistency_in_text}.
 	\item Part (iv) of Assumption \ref{assumption_consistency} is implied by part (iii) of Assumption \ref{assumption_consistency_in_text}.
 	\item Fix any $\theta \in \Theta$ and define:
 	\begin{align*}
 	\eta_{n}(\theta):= &\bigg\{\max_{j \in \mathcal{J}(\theta)} \sup_{\pi \in \Pi} |\E_{n}[m_{j}(Y_{i},X_{i},\pi,\theta)] - \E_{P}[m_{j}(Y_{i},X_{i},\pi,\theta)]|,\\
 	&\qquad\qquad\qquad\qquad\qquad\qquad\sup_{\pi \in \Pi} |\E_{n}[\psi(Y_{i},X_{i},\pi,\theta)] - \E_{P}[\psi(Y_{i},X_{i},\pi,\theta)]| \bigg\}. 
 	\end{align*}
 	Now define the classes of functions:
 	\begin{align*}
 	\mathcal{M}_{j}(\theta)&:= \left\{m_{j}(\,\cdot\,,\pi,\theta):\mathcal{Y}\times \mathcal{X} \to \mathbb{R} : \pi \in \Pi \right\},\\
	\tilde{\Psi}(\theta)&:= \left\{\psi(\,\cdot\,,\pi,\theta):\mathcal{Y}\times \mathcal{X} \to \mathbb{R} : \pi \in \Pi \right\}. 
 	\end{align*}
 	By Assumption \ref{assumption_consistency_in_text}, $\psi(\,\cdot\,,\theta): \mathcal{Y}\times\mathcal{X} \times \Pi \to \mathbb{R}$ is measurable in $(Y,X)$ and linear in $\pi \in \Pi$, and (for all $j \in \mathcal{J}(\theta)$) the functions $m_{j}(\,\cdot\,,\theta): \mathcal{Y}\times \mathcal{X} \times \Pi\to \mathbb{R}$ are measurable in $(Y,X)$ and linear in $\pi \in \Pi$. These classes are thus all VC-subgraph classes (c.f. Lemma 2.6.15 in \cite{van1996weak}) with a bounded (and thus uniformly square integrable) envelope function. From here, standard arguments show that these classes are Donsker uniformly over $P \in \mathcal{P}$, and thus $\eta_{n}(\theta)$ is $O_{P}(n^{-1/2})$. This argument shows that part (v) of Assumption \ref{assumption_consistency} is satisfied with $a_{n}:= \sqrt{n}$. 
 	\item From the previous part, any sequence satisfying $b_{n} = O(1/\sqrt{\log(n)})$ also satisfies $b_{n} \geq \eta_{n}(\theta)$ w.p.a. 1. 
 	\item Part (vii) of Assumption \ref{assumption_consistency} is implied by part (vi) of Assumption \ref{assumption_consistency_in_text}. Note this assumption also ensures that we can consider at most a finite number of moment inequalities, indexed by by $j \in \cup_{\theta \in \Theta'} \mathcal{J}(\theta)$. 
 	\item The last part of Assumption \ref{assumption_consistency} (i.i.d. data) is implied by part (viii) of Assumption \ref{assumption_consistency_in_text}.
 \end{enumerate}

\end{proof}

\begin{proof}[Proof of Proposition \ref{proposition_inference_in_text}]
Fix some $\theta \in \Theta$ and let $C_{n}(1-\alpha,\theta)$ denote the confidence set constructed using the procedure of \cite{cho2021simple}. We first show:
\begin{align}
\liminf_{n\to \infty}  \inf_{\{(\psi,P): \psi\in \Psi^{*}(\theta,P), P \in \mathcal{P}\}} (\text{Pr}_{P}\times P_{\xi}) \left(\psi_{0} \in CS_{n}(1-\alpha,\theta)  \right) \geq 1-\alpha.\label{eq_sub_CS2}
\end{align}
To do so, it suffices to show that, for our fixed $\theta \in \Theta$, Assumptions 3.1 and 3.2 in \cite{cho2021simple} are satisfied. By Assumption \ref{assumption_consistency_in_text}, $\psi(\,\cdot\,,\theta): \mathcal{Y}\times\mathcal{X} \times \Pi \to \mathbb{R}$ is linear in $\pi$ and the functions $m_{j}(\,\cdot\,,\theta): \mathcal{Y}\times \mathcal{X} \times \Pi \to \mathbb{R}$ are linear in $\pi \in \Pi$. Since $(Y,X)$ has finite support, we can equip $\mathcal{Y}\times\mathcal{X}$ with the discrete topology, in which case every function on $\mathcal{Y}\times \mathcal{X}$ is continuous. This verifies Assumption 3.1 in \cite{cho2021simple}. Parts (i), (ii) and (iii) of Assumption 3.2 in \cite{cho2021simple} are implied by the fact that $\Pi$ is a compact and convex polytope, and parts  (iii) and (viii) of Assumption \ref{assumption_consistency_in_text} (resp.). Parts (iv), (v) and (vi) of Assumption 3.2 in \cite{cho2021simple} are then implied by parts (vii), (iv) and (v) of Assumption \ref{assumption_consistency_in_text} (resp.). This verifies Assumption 3.2 in \cite{cho2021simple}. Now note:
\begin{align*}
&\liminf_{n\to \infty} \inf_{\{(\psi,P): \psi\in \Psi^{*}(P), P \in \mathcal{P}\}} (\text{Pr}_{P}\times P_{\xi}) \left(\psi_{0} \in CS_{n}(1-\alpha)  \right)\\
&=\liminf_{n\to \infty} \inf_{\{(\psi,P): \psi\in \Psi^{*}(P,\theta), \theta \in \Theta', P \in \mathcal{P}\}} (\text{Pr}_{P}\times P_{\xi})\left(\psi_{0} \in CS_{n}(1-\alpha)  \right)\\
&=\liminf_{n\to \infty} \min_{\theta \in \Theta'} \inf_{\{(\psi,P): \psi\in \Psi^{*}(\theta,P), P \in \mathcal{P}\}} (\text{Pr}_{P}\times P_{\xi}) \left(\psi_{0} \in CS_{n}(1-\alpha)  \right)\\
&=\liminf_{n\to \infty} \min_{\theta \in \Theta'} \inf_{\{(\psi,P): \psi\in \Psi^{*}(\theta,P), P \in \mathcal{P}\}} (\text{Pr}_{P}\times P_{\xi}) \left(\psi_{0} \in \bigcup_{\theta \in \Theta'} CS_{n}(1-\alpha,\theta)  \right)\\
&\geq\liminf_{n\to \infty} \min_{\theta \in \Theta'} \inf_{\{(\psi,P): \psi\in \Psi^{*}(\theta,P), P \in \mathcal{P}\}} \min_{\theta \in \Theta'} (\text{Pr}_{P}\times P_{\xi})\left(\psi_{0} \in CS_{n}(1-\alpha,\theta)  \right)\\
&=\liminf_{n\to \infty} \min_{\theta \in \Theta'} \inf_{\{(\psi,P): \psi\in \Psi^{*}(\theta,P), P \in \mathcal{P}\}} (\text{Pr}_{P}\times P_{\xi}) \left(\psi_{0} \in CS_{n}(1-\alpha,\theta)  \right)\\
&=\min_{\theta \in \Theta'} \liminf_{n\to \infty}  \inf_{\{(\psi,P): \psi\in \Psi^{*}(\theta,P), P \in \mathcal{P}\}} (\text{Pr}_{P}\times P_{\xi})\left(\psi_{0} \in CS_{n}(1-\alpha,\theta)  \right)\\
&\geq 1-\alpha,
\end{align*}
where the second last line follows from continuity of the minimum, and the last line follows from \eqref{eq_sub_CS2}. This completes the proof.
\end{proof}

\subsection{Measurability Results}\label{appendix_measurability_results}

\begin{definition}[Weak Measurability, Random Set, Selection]\label{definition_selections}
Let $(\Omega,\mathfrak{A},P)$ be a probability space, let $\mathcal{V}$ be a Polish space, and let $\mathcal{O}_{\mathcal{V}}$ denote the collection of all open sets on $\mathcal{V}$. A multifunction $\mathbb{V}: \Omega \to 2^{\mathcal{V}}$ is called weakly-measurable if for every $A \in \mathcal{O}_{\mathcal{V}}$ we have $\mathbb{V}^{-}(A):=\{ \omega \in \Omega :  \mathbb{V}(\omega) \cap A \neq \emptyset \} \in \mathfrak{A}$. 
A random set is a weakly measurable multifunction defined on a probability space. 
If $\mathbb{V}: \Omega \to 2^{\mathcal{V}}$ is a random set, then a random element $V: \Omega \to \mathcal{V}$ is called a (measurable) selection of $ V$ if $V(\omega) \in  \mathbb{V}(\omega)$ for $P-$almost all $\omega \in \Omega$. 
\end{definition}

\begin{lemma}\label{lemma_effros_measurability}
Suppose Assumption \ref{assumption_basic} holds. Then for each $\theta \in \Theta$, the map $\mathcal{U}(Y(\,\cdot\,),X(\,\cdot\,),\theta): \Omega \to 2^\mathcal{U}$ is a weakly-measurable multifunction, and thus is a random set.
\end{lemma}
\begin{proof}[Proof of Lemma \ref{lemma_effros_measurability}]
Fix any open set $A \in \mathcal{O}_{\mathcal{U}}$. We want to show that: 
\begin{align*}
\{ \omega \in \Omega :  \mathcal{U}(Y(\omega),X(\omega),\theta) \cap A \neq \emptyset \} \in \mathfrak{A}.
\end{align*}
First, define:
\begin{align*}
B(A):= \{ (y,x) \in \mathcal{Y}\times \mathcal{X} :  \mathcal{U}(y,x,\theta) \cap A \neq \emptyset \}. 
\end{align*}
Since $\mathcal{Y} \times\mathcal{X}$ is finite, and is equipped with the discrete topology and the Borel $\sigma-$algebra, we trivially have $B(A) \in \mathfrak{B}(\mathcal{Y})\otimes \mathfrak{B}(\mathcal{X})$. Since, $(Y,X) : \Omega \to \mathcal{Y}\times \mathcal{X}$ is measurable by assumption, we have $(Y,X)^{-1}(B) := \{ \omega : (Y(\omega),X(\omega)) \in B\} \in \mathfrak{A}$. Thus:
\begin{align*}
\{ \omega \in \Omega :  \mathcal{U}(Y(\omega),X(\omega),\theta) \cap A \neq \emptyset \} = (Y,X)^{-1}(B(A)) \in \mathfrak{A},
\end{align*}
as desired.

\end{proof}

Given a $\sigma-$algebra $\mathfrak{F}$ on a space $\mathcal{R}$, the $P$-completion of $\mathfrak{F}$ is the smallest $\sigma-$algebra containing $\mathfrak{F}$ as well as all $P-$null sets of $\mathcal{R}$. The intersection of all $P-$completions of $\mathfrak{F}$ (over all $P$) is called the \textit{universal $\sigma-$algebra}, and functions that are measurable with respect to the universal $\sigma-$algebra are said to be \textit{universally measurable}. The following Lemma shows that the random set $\text{cl } \mathcal{U}(Y,X,\theta)$ admits a universally measurable selection under Assumption \ref{assumption_basic}.
\begin{lemma}\label{lemma_selection}
Suppose Assumption \ref{assumption_basic} holds. Then $\text{cl }\mathcal{U}(Y(\omega),X(\omega),\theta)$ admits a universally measurable selection for every $\theta \in \Theta$ ensuring it is nonempty almost surely.
\end{lemma}
\begin{proof}[Proof of Lemma \ref{lemma_selection}]
Fix some $\theta \in \Theta$ ensuring $ \mathcal{U}(Y(\omega),X(\omega),\theta)$ is almost surely nonempty. We can then revise $\mathcal{U}(Y(\omega),X(\omega),\theta)$ on any null set to ensure it is nonempty for all $\omega \in \Omega$.  By Lemma \ref{lemma_effros_measurability}, $\mathcal{U}(Y(\omega),X(\omega),\theta)$ is weakly-measurable, and by Theorem 18.6 in \cite{aliprantis2006infinite} this implies that the graph of $\text{cl }\mathcal{U}(Y(\omega),X(\omega),\theta)$ belongs to $\mathfrak{A}\times \mathfrak{B}(\mathcal{U})$; that is, $\text{cl }\mathcal{U}(Y(\omega),X(\omega),\theta)$ is graph-measurable. The result then follows immediately from Theorem 3 of \cite{sainte1974extension}.
\end{proof}

Taking the closure in Lemma \ref{lemma_selection} is a technical detail that does not impact any of the identification results since under Assumption \ref{assumption_basic} the paper restricts attention to selections that assign zero probability to the boundary of $\mathcal{U}(Y,X,\theta)$.

\section{Additional Definitions and Results}\label{appendix_additional_definitions_and_results}

\subsection{Independence Assumptions}\label{subappendix_independence}

Under Assumption \ref{assumption_independence}, we have the following definition of the identified set, which is analogous to both Definitions \ref{definition_identified_set} and \ref{definition_identified_set_counterfactual}. 
\begin{definition}\label{definition_identified_set_independence}
Under Assumptions \ref{assumption_basic} and \ref{assumption_independence}, the identified set $\mathcal{I}_{Y,W,Z}^{*}$ is the set of all pairs $(P_{U\mid Y,W,Z},\theta)$ such that:
\begin{enumerate}[label=(\roman*)]
	\item $(P_{U\mid Y,W,Z},\theta)$ satisfies: 
	\begin{align}
	P_{U\mid Y,W,Z}(U \in  \mathcal{U}(Y,W,Z,\theta)  \mid Y=y, W=w, Z=z)&=1,
	\end{align}
	$P_{Y,W,Z}-$a.s.
	\item The distribution $P_{U} = P_{U \mid Y,W,Z} P_{Y,W,Z}$ assigns zero probability to all sets of the form $\{u \in \mathcal{U} : \varphi(w,z,u,\theta) =0\}$.
	\item For all Borel sets $A \in \mathfrak{B}(\mathcal{U})$ we have $P_{U\mid Z}(A \mid Z=z) = P_{U}(A)$, $P_{Z}-$a.s. 
\end{enumerate}
Furthermore, under Assumptions \ref{assumption_basic}, \ref{assumption_counterfactual_domain} and \ref{assumption_independence}, the identified set of counterfactual conditional distributions $\mathcal{P}_{Y_{\gamma}\mid Y,W,Z,U}^{*}$ is the set of all conditional distributions $P_{Y_{\gamma}\mid Y,W,Z,U}$ satisfying:
\begin{align}
P_{Y_{\gamma}\mid Y, W,Z,U}\left(Y_{\gamma} = \mathbbm{1}\{\varphi(\gamma(W,Z),U,\theta) \geq 0  \}\mid Y=y, W=w,Z=z,U=u\right)=1, \label{eq_cf_xzt2}
\end{align}
$P_{Y,W,Z,U}-$a.s. for some pair $(P_{U\mid Y,W,Z},\theta)\in \mathcal{I}_{Y,W,Z}^{*}$. 
\end{definition}
Here we do not consider the case when both Assumptions \ref{assumption_independence} and \ref{assumption_monotonicity} hold, but we again note that this definition (and the results to follow) are easily modified to accommodate the case when any combination of these assumptions hold. We now provide the following Corollary whose proof follows almost identically to that of Theorems \ref{thm_counterfactual_identified_set} and \ref{theorem_infinite_to_finite}, with the exception being that we require condition (ii) of Definition \ref{definition_identified_set_independence} to hold.
\begin{corollary}\label{corollary_infinite_to_finite_independence}
Under Assumptions \ref{assumption_basic}, \ref{assumption_counterfactual_domain} and \ref{assumption_independence}, a counterfactual conditional distribution $P_{Y_{\gamma}\mid Y,W,Z}$ satisfies $P_{Y_{\gamma}\mid Y,W,Z} \in \mathcal{P}_{Y_{\gamma}\mid Y,W,Z}^{*}$ if and only if there exists a pair $(P_{U\mid Y,W,Z},\theta) \in \mathcal{I}_{Y,W,Z}^{*}$ (for $\mathcal{I}_{Y,W,Z}^{*}$ from Definition \ref{definition_identified_set_independence}) satisfying:
\begin{align}
P_{Y_{\gamma}\mid W,Z}\left(Y_{\gamma} = 1\mid Y=y,W=w,Z=z\right)= P_{U\mid Y,W,Z}\left( \varphi(\gamma(W,Z),U,\theta) \geq 0 \mid Y=y,W=w,Z=z\right),\qquad\label{eq_theorem_counterfactual2}
\end{align}
$P_{Y,W,Z}-$a.s. Furthermore, for any collection of counterfactual conditional distributions $P_{Y_{\gamma} \mid Y,W,Z}$, there exists a collection of Borel conditional probability measures $P_{U \mid Y,W,Z}$ satisfying \eqref{eq_theorem_counterfactual2} with $(P_{U \mid Y,W,Z},\theta) \in \mathcal{I}_{Y,W,Z}^{*}$ (for $\mathcal{I}_{Y,W,Z}^{*}$ from Definition \ref{definition_identified_set_independence}) if and only if there exists a collection $P_{U \mid Y,W,Z}$ of probability measures on the sets in $\mathcal{A}(\theta)$ from \eqref{eq_sufficient_sets} satisfying:
\begin{align}
\sum_{s \in S_{j}}  P_{U\mid Y,W,Z}\left( \text{int}(\mathcal{U}(s,\theta)) \mid Y=1,W=w_{j},Z=z_{j}\right)&=1,\label{eq_finite_existence1_2}\\
\sum_{s \in S_{j}^{c}}  P_{U\mid Y,W,Z}\left(\text{int}(\mathcal{U}(s,\theta)) \mid Y=0,W=w_{j},Z=z_{j}\right)&=1,\,\label{eq_finite_existence2_2}\\
\sum_{s \in S_{\gamma(j)}}  P_{U\mid Y,W,Z}\left( \text{int}(\mathcal{U}(s,\theta)) \mid Y=y,W=w_{j},Z=z_{j}\right)&=P_{Y_{\gamma}\mid Y,W,Z}\left( Y_{\gamma}=1 \mid Y=y,W=w_{j},Z=z_{j}\right),\label{eq_finite_existence3_2}
\end{align}
for $y \in \{0,1\}$ and $j\in\{1,\ldots,m\}$ assigned positive probability, and:
\begin{align}
&\sum_{y }\sum_{w} P_{U\mid Y,W,Z}\left( \text{int}(\mathcal{U}(s,\theta)) \mid Y=y,W=w,Z=z_{k}\right) P(Y=y,W=w \mid Z=z_{k})\nonumber\\
&\qquad\qquad=\sum_{y }\sum_{w}  P_{U\mid Y,W,Z}\left( \text{int}(\mathcal{U}(s,\theta)) \mid Y=y,W=w,Z=z_{k+1}\right)P(Y=y,W=w \mid Z=z_{k+1}),\label{eq_finite_existence4_2}
\end{align}
for all $s \in \{0,1\}^{m}$ and all $k=1,\ldots,m_{z}-1$ assigned positive probability. 
\end{corollary}
\begin{proof}[Proof of Corollary \ref{corollary_infinite_to_finite_independence}]

The first statement follows a proof identical to the proof of Theorem \ref{thm_counterfactual_identified_set}. For the second statement, the forward direction is identical to the proof of Theorem \ref{theorem_infinite_to_finite}. The reverse direction is similar to the proof of Theorem \ref{theorem_infinite_to_finite}, with the exception that we must show that the extended measure on $\mathfrak{B}(\mathcal{U})$ satisfies independence if the intial measure on $\mathcal{A}(\mathcal{U})$ satisfies independence. Let $\tilde{P}_{U\mid Y,W,Z}$ be the extension of $P_{U\mid Y,W,Z}$ from the proof of Theorem \ref{theorem_infinite_to_finite}. Then for any $A \in \mathfrak{B}(\mathcal{U})$:
\begin{align*}
&\tilde{P}_{U\mid Z}(A \mid Z=z_{k})\\
&=\sum_{y \in \{0,1\}}\sum_{w \in \mathcal{W}}\sum_{s \in \{0,1\}^{m}} \mathbbm{1}(A,\theta,s) P_{U\mid Y,W,Z}\left( \text{int}(\mathcal{U}(s,\theta)) \mid Y=y,W=w,Z=z_{k}\right) P_{Y,W\mid Z}(Y=y,W=w \mid Z=z_{k})\\
&=\sum_{s \in \{0,1\}^{m}} \mathbbm{1}(A,\theta,s)\sum_{y \in \{0,1\}}\sum_{w \in \mathcal{W}} P_{U\mid Y,W,Z}\left(\text{int}(\mathcal{U}(s,\theta)) \mid Y=y,W=w,Z=z_{k}\right) P_{Y,W\mid Z}(Y=y,W=w \mid Z=z_{k})\\
&=\sum_{s \in \{0,1\}^{m}} \mathbbm{1}(A,\theta,s)\sum_{y \in \{0,1\}}\sum_{w \in \mathcal{W}} P_{U\mid Y,W,Z}\left( \text{int}(\mathcal{U}(s,\theta)) \mid Y=y,W=w,Z=z_{k+1}\right) P_{Y,W\mid Z}(Y=y,W=w \mid Z=z_{k+1})\\
&=\sum_{y \in \{0,1\}}\sum_{w \in \mathcal{W}} \sum_{s \in \{0,1\}^{m}} \mathbbm{1}(A,\theta,s) P_{U\mid Y,W,Z}\left( \text{int}(\mathcal{U}(s,\theta)) \mid Y=y,W=w,Z=z_{k+1}\right)P_{Y,W\mid Z}(Y=y,W=w \mid Z=z_{k+1})\\
&=\tilde{P}_{U\mid Z}(A \mid Z=z_{k+1}),
\end{align*}
for all pairs $z_{k}$ and $z_{k+1}$ assigned positive probability, where the third equality follows from \eqref{eq_finite_existence4_2}. Conclude that $\tilde{P}_{U\mid Z}$ satisfies the second condition in Definition \ref{definition_identified_set_independence}. 
\end{proof}

Analogous to Theorem \ref{thm_counterfactual_identified_set}, the first part of Corollary \ref{corollary_infinite_to_finite_independence} provides the theoretical link between the identified set for counterfactual conditional distributions and the identified set for the pair $(P_{U\mid Y,W,Z},\theta)$ under the additional independence assumption between $U$ and $Z$. Furthermore, analogous to the result in Theorem \ref{theorem_infinite_to_finite}, the second part of Corollary \ref{corollary_infinite_to_finite_independence} reduces an infinite dimensional existence problem to a finite dimensional existence problem. Importantly, the second part of Corollary \ref{corollary_infinite_to_finite_independence} builds on Theorem \ref{theorem_infinite_to_finite} by demonstrating that Assumption \ref{assumption_independence}---which requires $P_{U\mid Z}(A \mid Z=z) = P_{U}(A)$ a.s. for all Borel sets $A$---can be imposed by considering only a finite number of equality constraints on a distribution $P_{U\mid Y,W,Z}$ defined on sets of the form $\mathcal{U}(s,\theta)$.

We have the following Corollary to Theorem \ref{thm_linprog_finite_case}:
\begin{corollary}\label{corollary_linprog_finite_case_independence}
Under Assumptions \ref{assumption_basic}, \ref{assumption_counterfactual_domain}, and \ref{assumption_independence}, the identified set for the counterfactual conditional probability $P_{Y_{\gamma}\mid Y,W,Z}(Y_{\gamma} = 1\mid Y=y, W=w_{j},Z=z_{j})$ is given by:
\begin{align}
\bigcup_{\theta \in \Theta} [\pi_{\ell b}(y,w_{j},z_{j},\theta),\pi_{u b}(y,w_{j},z_{j},\theta)], \label{eq_thm_linprog_finite_union_independence}
\end{align}
where $\pi_{\ell b}(y,w_{j},z_{j},\theta)$ and $\pi_{u b}(y,w_{j},z_{j},\theta)$ are determined by the optimization problems:
\begin{align}
\pi_{\ell b}(y,w_{j},z_{j},\theta) &:= \min_{\pi(\theta)\in \mathbb{R}^{d_{\pi}}} \sum_{s \in S_{\gamma(j)}} \pi(y,w_{j},z_{j},s,\theta),\text{ s.t. \eqref{eq_moment_conditions_finite_case}, \eqref{eq_non_negative}, \eqref{eq_adding_up}, and \eqref{eq_independence_constraint},}\label{eq_thm_linprog_finite_LB_independence}\\
\pi_{u b}(y,w_{j},z_{j},\theta) &:=  \max_{\pi(\theta)\in \mathbb{R}^{d_{\pi}}} \sum_{s \in S_{\gamma(j)}} \pi(y,w_{j},z_{j},s,\theta),\text{ s.t. \eqref{eq_moment_conditions_finite_case}, \eqref{eq_non_negative}, \eqref{eq_adding_up}, and \eqref{eq_independence_constraint}.}\label{eq_thm_linprog_finite_UB_independence}
\end{align}
\end{corollary}

Note that this Corollary is identical to Theorem \ref{thm_linprog_finite_case} with the exception that we have imposed Assumption \ref{assumption_independence}, and thus have included constraints of the form in \eqref{eq_independence_constraint}.  With the exception of these additional constraints, the optimization problems that characterize the bounding problem are the same as before. Again, this result can be easily modified to bound any linear function of counterfactual conditional distributions by simply modifying the objective function in the optimization problems \eqref{eq_thm_linprog_finite_LB_independence} and \eqref{eq_thm_linprog_finite_UB_independence}. 

\subsection{Monotonicity Assumptions}\label{subappendix_monotonicity}

When we entertain Assumption \ref{assumption_monotonicity}, we have the following definition of the identified set, which is analogous to both Definitions \ref{definition_identified_set} and \ref{definition_identified_set_counterfactual}. 
\begin{definition}\label{definition_identified_set_monotonicity}
Under Assumptions \ref{assumption_basic} and \ref{assumption_monotonicity}, the identified set $\mathcal{I}_{Y,X}^{*}$ is the set of all pairs $(P_{U\mid Y,X},\theta)$ such that:
\begin{enumerate}[label=(\roman*)]
	\item $(P_{U\mid Y,X},\theta)$ satisfies: 
	\begin{align}
	P_{U\mid Y,X}(U \in  \mathcal{U}(Y,X,\theta) \mid Y=y, X=x)=1,
	\end{align}
	$P_{Y,X}-$a.s.
	\item The distribution $P_{U} = P_{U\mid Y,X} P_{Y,X}$ assigns zero probability to all sets of the form $\{u \in \mathcal{U} : \varphi(x,u,\theta)=0 \}$.
	\item For all $(j,k) \in \mathcal{M}$ from Assumption \ref{assumption_monotonicity}, we have:
	\begin{align}
	P_{U\mid Y,X}(\varphi(x_{j},U,\theta) \leq \varphi(x_{k},U,\theta) \mid Y=y, X=x)=1 \text{ a.s.}
	\end{align}

\end{enumerate}
Furthermore, under Assumptions \ref{assumption_basic}, \ref{assumption_counterfactual_domain}, and \ref{assumption_monotonicity}, the identified set of counterfactual conditional distributions $\mathcal{P}_{Y_{\gamma}\mid Y,X,U}^{*}$ is the set of all conditional distributions $P_{Y_{\gamma}\mid Y,X,U}$ satisfying:
\begin{align}
P_{Y_{\gamma}\mid Y,X,U}\left(Y_{\gamma} = \mathbbm{1}\{\varphi(\gamma(X),U,\theta) \geq 0  \}\mid Y=y, X=x,U=u\right)=1, \label{eq_cf_xzt3}
\end{align}
$P_{Y,X,U}-$a.s. for some pair $(P_{U\mid Y,X},\theta)\in \mathcal{I}_{Y,X}^{*}$. 
\end{definition}
Again, this definition and the results to follow are easily modified to accommodate the case when any combination of Assumptions \ref{assumption_independence} and \ref{assumption_monotonicity} hold. We now provide the following Corollary whose proof follows almost identically to that of Theorems \ref{thm_counterfactual_identified_set} and \ref{theorem_infinite_to_finite}, with the exception being that we require condition (ii) of Definition \ref{definition_identified_set_monotonicity} to hold.
\begin{corollary}\label{corollary_infinite_to_finite_monotonicity}
Under Assumptions \ref{assumption_basic}, \ref{assumption_counterfactual_domain}, and \ref{assumption_monotonicity}, a counterfactual conditional distribution $P_{Y_{\gamma}\mid Y,X}$ satisfies $P_{Y_{\gamma}\mid Y,X} \in \mathcal{P}_{Y_{\gamma}\mid Y,X}^{*}$ if and only if there exists a pair $(P_{U\mid Y,X},\theta) \in \mathcal{I}_{Y,X}^{*}$ (for $\mathcal{I}_{Y,X}^{*}$ from Definition \ref{definition_identified_set_monotonicity}) satisfying:
\begin{align}
P_{Y_{\gamma}\mid Y,X}\left(Y_{\gamma} = 1\mid Y=y,X=x\right)= P_{U\mid Y,X}\left( \varphi(\gamma(X),U,\theta) \geq 0 \mid Y=y,X=x\right),\label{eq_theorem_counterfactual3}
\end{align}
$P_{Y,X}-$a.s. Furthermore, for any collection of counterfactual conditional distributions $P_{Y_{\gamma} \mid Y,X}$, there exists a collection of Borel conditional probability measures $P_{U \mid Y,X}$ satisfying \eqref{eq_theorem_counterfactual3} with $(P_{U \mid Y,X},\theta) \in \mathcal{I}_{Y,X}^{*}$ (for $\mathcal{I}_{Y,X}^{*}$ from Definition \ref{definition_identified_set_monotonicity}) if and only if there exists a collection $P_{U \mid Y,X}$ of probability measures on the sets in $\mathcal{A}(\theta)$ from \eqref{eq_sufficient_sets} satisfying:
\begin{align}
\sum_{s \in S_{j}}  P_{U\mid Y,X}\left( \text{int}(\mathcal{U}(s,\theta))\mid Y=1,X=x_{j}\right)&=1,\label{eq_finite_existence1_3}\\
\sum_{s \in S_{j}^{c}}  P_{U\mid Y,X}\left( \text{int}(\mathcal{U}(s,\theta)) \mid Y=0,X=x_{j}\right)&=1,\label{eq_finite_existence2_3}\\
\sum_{s \in S_{\gamma(j)}}  P_{U\mid Y,X}\left( \text{int}(\mathcal{U}(s,\theta)) \mid Y=y,X=x_{j}\right)&=P_{Y_{\gamma}\mid Y,X}\left( Y_{\gamma}=1 \mid Y=y,X=x_{j}\right),\label{eq_finite_existence3_3}
\end{align}
for $y \in \{0,1\}$ and $j\in \{1,\ldots,m\}$ assigned positive probability, and:
\begin{align}
\sum_{s \in S_{M}^{c}} P_{U\mid Y,X}\left(\text{int}(\mathcal{U}(s,\theta))\mid Y=y,X=x\right) = 0, \,\,a.s.\label{eq_finite_existence5_3}
\end{align}
for all $(y,x)$ assigned positive probability, where $S_{M}$ is as defined in Section \ref{section_additional_assumptions}.

\end{corollary}
The proof of this corollary is identical to the proof of Theorem \ref{thm_counterfactual_identified_set} and Theorem \ref{theorem_infinite_to_finite}. Analogous to Theorem \ref{thm_counterfactual_identified_set}, the first part of Corollary \ref{corollary_infinite_to_finite_monotonicity} provides the theoretical link between the identified set for counterfactual conditional distributions and the identified set for the pair $(P_{U\mid Y,X},\theta)$ under the additional monotonicity assumption. Analogous to Theorem \ref{theorem_infinite_to_finite}, the second part of Corollary \ref{corollary_infinite_to_finite_monotonicity} reduces an infinite dimensional existence problem to a finite dimensional existence problem amenable to analysis using optimization problems. Building on the intuition provided in Example \ref{example_monotonicity}, the second part of Corollary \ref{corollary_infinite_to_finite_monotonicity} demonstrates that monotonicity as in Assumption \ref{assumption_monotonicity} can be imposed by considering only a finite number of equality constraints on a distribution $P_{U\mid Y,X}$ defined on sets of the form $\mathcal{U}(s,\theta)$. By definition of the set $S_M$, condition \eqref{eq_finite_existence5_3} simply assigns probability zero to all sets $\mathcal{U}(s,\theta)$ that do not satisfy the monotonicity relation from Assumption \ref{assumption_monotonicity}. This leads to the following result.
\begin{corollary}\label{corollary_linprog_finite_case_monotonicity}
Under Assumptions \ref{assumption_basic}, \ref{assumption_counterfactual_domain}, and \ref{assumption_monotonicity}, the identified set for the counterfactual conditional probability $P_{Y_{\gamma}\mid Y,X}(Y_{\gamma} = 1\mid Y=y, X=x_{j})$ is given by:
\begin{align}
\bigcup_{\theta \in \Theta} [\pi_{\ell b}(y,x_{j},\theta),\pi_{u b}(y,x_{j},\theta)], \label{eq_thm_linprog_finite_union_monotonicity}
\end{align}
where $\pi_{\ell b}(y,x_{j},\theta)$ and $\pi_{u b}(y,x_{j},\theta)$ are determined by the optimization problems:
\begin{align}
\pi_{\ell b}(y,x_{j},\theta) &:= \min_{\pi(\theta)\in \mathbb{R}^{d_{\pi}}} \sum_{s \in S_{\gamma(j)}} \pi(y,x_{j},s,\theta),\text{ s.t. \eqref{eq_moment_conditions_finite_case}, \eqref{eq_non_negative}, \eqref{eq_adding_up}, and \eqref{eq_monotonicity_constraint},}\label{eq_thm_linprog_finite_LB_monotonicity}\\
\pi_{u b}(y,x_{j},\theta) &:=  \max_{\pi(\theta)\in \mathbb{R}^{d_{\pi}}} \sum_{s \in S_{\gamma(j)}} \pi(y,x_{j},s,\theta),\text{ s.t. \eqref{eq_moment_conditions_finite_case}, \eqref{eq_non_negative}, \eqref{eq_adding_up}, and \eqref{eq_monotonicity_constraint}.}\label{eq_thm_linprog_finite_UB_monotonicity}
\end{align}
\end{corollary} 
Note that this Corollary is identical to Theorem \ref{thm_linprog_finite_case} with the exception that we have imposed Assumption \ref{assumption_monotonicity}, and thus have included constraints of the form \eqref{eq_monotonicity_constraint}. With the exception of these additional constraints, the optimization problems that characterize the bounding problem are the same as before. Finally, alternative counterfactual quantities can be bounded in the same way by simply modifying the objective function in \eqref{eq_thm_linprog_finite_LB_monotonicity} and \eqref{eq_thm_linprog_finite_UB_monotonicity}. 

\subsection{Consistency}\label{appendix_consistency}

In this subsection we present a consistency result for functionals of a partially identified parameter related to results found in \cite{molchanov1998limit}, \cite{manski2002inference}, and \cite{chernozhukov2007estimation}. It is presented in a form that is more general than necessary for the current paper (and more general than Proposition \ref{proposition_consistency_in_text} in Section \ref{section_consistency_bias_correction_inference}), and so it may be of interest in other applications. We consider an environment where the researcher wishes to compute bounds on a functional $\E_{P}[\psi(X_{i},\tau_{1},\tau_{2})]$, where $\psi: \mathcal{X} \times \mathcal{T} \to \mathbb{R}$, $\mathcal{X}\subset \mathbb{R}^{d_{x}}$ denotes the support of the observed random vector $X_{i}$, and $\mathcal{T} = \mathcal{T}_{1}\times \mathcal{T}_{2} \subset \mathbb{R}^{d_{\tau}}$ denotes the parameter space with typical elements $\tau=(\tau_{1},\tau_{2}) \in \mathcal{T}$. The values of $(\tau_{1},\tau_{2})$ are constrained by moment inequalities of the form:
\begin{align*}
\E_{P}[m_{j}(X_{i},\tau_{1},\tau_{2})] \leq 0, \text{ for }j \in \mathcal{J}(\tau_{2}), 
\end{align*}
where $\mathcal{J}(\tau_{2})$ is a finite index set that may depend on $\tau_{2}$. Note this does not rule out moment equalities, since each moment equality can be equivalently written as a combination of two moment inequalities. In this environment, the identified set for $(\tau_{01},\tau_{02}) \in \mathcal{T}$ at the true $P$ is given by:
\begin{align*}
\mathcal{T}^{*}(P)&:= \left\{ (\tau_{1},\tau_{2}) \in \mathcal{T} : \E_{P}[m_{j}(X_{i},\tau_{1},\tau_{2})] \leq 0 \text{ for $j\in \mathcal{J}(\tau_{2})$} \right\}.
\end{align*}
In addition, the identified set for $\psi_{0}:= \E_{P}[\psi(X_{i},\tau_{01},\tau_{02})]$ is given by:
\begin{align*}
\Psi^{*}(P)&:= \left\{ \overline{\psi} \in \mathbb{R} : \exists (\tau_{1},\tau_{2}) \in \mathcal{T}^{*}(P) \text{ s.t. }  \overline{\psi}= \E_{P}[\psi(X_{i},\tau_{1},\tau_{2})] \right\}.
\end{align*}
Let us define the projection:
\begin{align*}
\mathcal{T}_{1}^{*}(\tau_{2},P)&:= \left\{ \tau_{1} \in \mathcal{T}_{1} :  \E_{P}[m_{j}(X_{i},\tau_{1},\tau_{2})] \leq 0 \text{ for $j\in \mathcal{J}(\tau_{2})$} \right\}.
\end{align*}
Under Assumption \ref{assumption_consistency} ahead, it is straightforward to show that $\Psi^{*}(P)$ can be rewritten as:
\begin{align*}
\Psi^{*}(P) = \bigcup_{\tau_{2} \in \mathcal{T}_{2}} [\Psi_{\ell b}(\tau_{2},P), \Psi_{u b}(\tau_{2},P)],
\end{align*}
where:
\begin{align*}
\Psi_{\ell b}(\tau_{2},P):= \min_{\tau_{1} \in \mathcal{T}_{1}^{*}(\tau_{2},P)} \E_{P}[\psi(X_{i},\tau_{1},\tau_{2})], && \Psi_{u b}(\tau_{2},P):= \max_{\tau_{1} \in \mathcal{T}_{1}^{*}(\tau_{2},P)} \E_{P}[\psi(X_{i},\tau_{1},\tau_{2})].
\end{align*}
We study the consistency properties of the sample analog estimator for this representation of $\Psi^{*}(P)$. In particular, define:
\begin{align*}
\E_{n}[\psi(X_{i},\tau_{1},\tau_{2})]:= \frac{1}{n} \sum_{i=1}^{n} \psi(X_{i},\tau_{1},\tau_{2}), && \E_{n}[m_{j}(X_{i},\tau_{1},\tau_{2})]:= \frac{1}{n} \sum_{i=1}^{n} m_{j}(X_{i},\tau_{1},\tau_{2}), \text{ for }j\in \mathcal{J}(\tau_{2}).
\end{align*}
Then the sample analog estimator of interest is given by: 
\begin{align*}
\Psi^{*}(\P_{n}) = \bigcup_{\tau_{2} \in \mathcal{T}_{2}} [\Psi_{\ell b}(\tau_{2},\P_{n}), \Psi_{u b}(\tau_{2},\P_{n})],
\end{align*}
where:
\begin{align*}
\Psi_{\ell b}(\tau_{2},\P_{n}):= \min_{\tau_{1} \in \mathcal{T}_{1}^{*}(\tau_{2},\P_{n})} \E_{n}[\psi(X_{i},\tau_{1},\tau_{2})], && \Psi_{u b}(\tau_{2},\P_{n}):= \max_{\tau_{1} \in \mathcal{T}_{1}^{*}(\tau_{2},\P_{n})} \E_{n}[\psi(X_{i},\tau_{1},\tau_{2})],
\end{align*}
and:
\begin{align*}
\mathcal{T}_{1}^{*}(\tau_{2},\P_{n})&:= \left\{ \tau_{1} \in \mathcal{T}_{1} : \E_{n}[m_{j}(X_{i},\tau_{1},\tau_{2})] \leq 0 \text{ for $j=1,\ldots,J$} \right\}.
\end{align*}
In the following, we define the sequence $\{\eta_{n}(\tau_{2})\}_{n=1}^{\infty}$ as:
\begin{align*}
\eta_{n}(\tau_{2}) := \max\left\{ \max_{j\in \mathcal{J}(\tau_{2})} \sup_{\tau_{1} \in \mathcal{T}_{1}} | \E_{n}[m_{j}(X_{i},\tau_{1},\tau_{2})] - \E_{P}[m_{j}(X_{i},\tau_{1},\tau_{2})]|,  \sup_{\tau_{1} \in \mathcal{T}_{1}} \left| \E_{n}[\psi(X_{i},\tau_{1},\tau_{2})]  -   \E_{P}[\psi(X_{i},\tau_{1},\tau_{2})] \right|\right\}.
\end{align*}
We impose the following assumption.
\begin{assumption}\label{assumption_consistency}
The parameter space $(\mathcal{T},\mathcal{P})$ satisfies the following: (i) $\mathcal{T}:=\mathcal{T}_{1}\times\mathcal{T}_{2} \subset \mathbb{R}^{d_{\tau}}$, where $\mathcal{T}_{1}$ is compact and convex; (ii) for each $\tau_{2} \in \mathcal{T}_{2}$, the function $\psi(\,\cdot\,,\tau_{2}): \mathcal{X} \times \mathcal{T}_{1} \to \mathbb{R}$ is measurable in $X_{i} \in \mathcal{X}\subset \mathbb{R}^{d_{x}}$ and is Lipschitz continuous in $\tau_{1}$ with a (possibly data-dependent) Lipschitz constant $C(\tau_{2})$ with $\sup_{\tau_{2} \in \mathcal{T}_{2}} C(\tau_{2})<\infty$ a.s.; (iii) for each $\tau_{2} \in \mathcal{T}_{2}$ and $j\in \mathcal{J}(\tau_{2})$, the moment function $m_{j}(\,\cdot\,,\tau_{2}): \mathcal{X} \times \mathcal{T}_{1} \to \mathbb{R}$ is measurable in $X_{i}$ and convex (and thus continuous) in $\tau_{1}$; (iv) for each $P \in \mathcal{P}$ there exists a $(\tau_{1},\tau_{2}) \in \mathcal{T}$ such that $\E_{P}[m_{j}(X_{i},\tau_{1},\tau_{2})] \leq 0, \text{ for }j\in \mathcal{J}(\tau_{2})$; (v) for each fixed $\tau_{2} \in \mathcal{T}_{2}$, we have $\eta_{n}(\tau_{2}) = O_{P}(a_{n}^{-1})$ for some sequence $a_{n} \uparrow \infty$; (vi) for each fixed $\tau_{2} \in \mathcal{T}_{2}$, there exists a sequence $b_{n} \downarrow 0$ satisfying $b_{n} \geq \eta_{n}(\tau_{2})$ w.p.a. 1; (vii) there exists a finite subset $\mathcal{T}_{2}' \subset \mathcal{T}_{2}$ such that:
\begin{align*}
&\left\{ \tau_{1} \in \mathcal{T}_{1} : \exists \tau_{2} \in \mathcal{T}_{2} \text{ s.t. } \E_{P}[m_{j}(X_{i},\tau_{1},\tau_{2})] \leq 0 \text{ for $j\in \mathcal{J}(\tau_{2})$} \right\}\\
&\qquad\qquad\qquad= \left\{ \tau_{1} \in \mathcal{T}_{1} : \exists \tau_{2} \in \mathcal{T}_{2}' \text{ s.t. } \E_{P}[m_{j}(X_{i},\tau_{1},\tau_{2})] \leq 0 \text{ for $j\in \mathcal{J}(\tau_{2})$} \right\}. 
\end{align*}
Finally, the sample $\{X_{i}\}_{i=1}^{n}$ is given by $n$ i.i.d. draws from some $P \in \mathcal{P}$.
\end{assumption}
For any $c \in \mathbb{R}$ let us define:
\begin{align*}
\mathcal{T}_{1}^{*}(\tau_{2},P,c)&:= \left\{ \tau_{1} \in \mathcal{T}_{1} : \E_{P}[m_{j}(X_{i},\tau_{1},\tau_{2})] \leq c \text{ for $j\in \mathcal{J}(\tau_{2})$} \right\},
\end{align*}
and:
\begin{align*}
\Psi^{*}(P,c) = \bigcup_{\tau_{2} \in \mathcal{T}_{2}'} [\Psi_{\ell b}(\tau_{2},P,c), \Psi_{u b}(\tau_{2},P,c)],
\end{align*}
where:
\begin{align*}
\Psi_{\ell b}(\tau_{2},P,c):= \min_{\tau_{1} \in \mathcal{T}_{1}^{*}(\tau_{2},P,c)} \E_{P}[\psi(X_{i},\tau_{1},\tau_{2})], && \Psi_{u b}(\tau_{2},\P_{n},c):= \max_{\tau_{1} \in \mathcal{T}_{1}^{*}(\tau_{2},P,c)} \E_{P}[\psi(X_{i},\tau_{1},\tau_{2})].
\end{align*}
Define the sets $\mathcal{T}_{1}^{*}(\tau_{2},P,c)$ and $\Psi^{*}(P,c)$ analogously. The following Theorem then shows that a slight enlargement of the set $\Psi^{*}(\P_{n})$ is a consistent estimator for the set $\Psi^{*}(P)$, where consistency is defined using the Hausdorff metric. 
\begin{theorem}\label{theorem_consistency}
Suppose that Assumption \ref{assumption_consistency} holds. Then $d_{H}(\Psi^{*}(\P_{n},b_{n}),\Psi^{*}(P)) = o_{P}(1)$, where $b_{n}$ is any sequence satisfying Assumption \ref{assumption_consistency}. 
\end{theorem}
\begin{proof}[Proof of Theorem \ref{theorem_consistency}]
We have:
\begin{align*}
d_{H}(\Psi^{*}(\P_{n},b_{n}),\Psi^{*}(P)) \leq \sum_{\tau_{2} \in \mathcal{T}_{2}'}d_{H}\left([\Psi_{\ell b}(\tau_{2},\P_{n},b_{n}), \Psi_{u b}(\tau_{2},\P_{n},b_{n})], [\Psi_{\ell b}(\tau_{2},P), \Psi_{u b}(\tau_{2},P)]\right).
\end{align*}
Since $\mathcal{T}_{2}'$ is finite, it suffices to show that:
\begin{align*}
d_{H}\left([\Psi_{\ell b}(\tau_{2},\P_{n},b_{n}), \Psi_{u b}(\tau_{2},\P_{n},b_{n})], [\Psi_{\ell b}(\tau_{2},P), \Psi_{u b}(\tau_{2},P)]\right)=o_{P}(1),
\end{align*}
for each $\tau_{2} \in \mathcal{T}_{2}'$. To this end, fix any $\tau_{2} \in \mathcal{T}_{2}'$. To show the previous display, it suffices to show consistency of the upper and lower bounds; i.e. that $|\Psi_{\ell b}(\tau_{2},\P_{n},b_{n}) - \Psi_{\ell b}(\tau_{2},P)| = o_{P}(1)$ and that $|\Psi_{u b}(\tau_{2},\P_{n},b_{n}) - \Psi_{u b}(\tau_{2},P)| = o_{P}(1)$. We focus on the lower bound, since the upper bound proof is symmetric. 

First recall that $\psi(X_{i},\tau_{1},\tau_{2})$ is continuous with respect to $\tau_{1}$ for every $\tau_{2}$ by Assumption \ref{assumption_consistency}, and $\mathcal{T}_{1}$ is compact. Thus, we have that $\psi(X_{i},\tau_{1},\tau_{2})$ is uniformly continuous (w.r.t. $\tau_{1}$) on $\mathcal{T}_{1}$. Thus, for every $\varepsilon>0$ there exists a $\delta(\varepsilon)>0$ such that $|\E_{n}[\psi(X_{i},\tau_{1},\tau_{2})] - \E_{n}[\psi(X_{i},\tau_{1}',\tau_{2})]|< \varepsilon$ whenever $||\tau_{1} - \tau_{1}'|| < \delta(\varepsilon)$. Now note that:
\begin{align*}
&|\Psi_{\ell b}(\tau_{2},\P_{n},b_{n}) - \Psi_{\ell b}(\tau_{2},P)|\\
&= \left|\min_{\tau_{1} \in \mathcal{T}_{1}^{*}(\tau_{2},\P_{n},b_{n})}  \E_{n}[\psi(X_{i},\tau_{1},\tau_{2})] -  \min_{\tau_{1} \in \mathcal{T}_{1}^{*}(\tau_{2},P)}  \E_{P}[\psi(X_{i},\tau_{1},\tau_{2})] \right|,\\
&\leq \left|\min_{\tau_{1} \in \mathcal{T}_{1}^{*}(\tau_{2},\P_{n},b_{n})}  \E_{n}[\psi(X_{i},\tau_{1},\tau_{2})] -  \min_{\tau_{1} \in \mathcal{T}_{1}^{*}(\tau_{2},P)}  \E_{n}[\psi(X_{i},\tau_{1},\tau_{2})] \right|\\
&\qquad\qquad\qquad+ \left|\min_{\tau_{1} \in \mathcal{T}_{1}^{*}(\tau_{2},P)}  \E_{n}[\psi(X_{i},\tau_{1},\tau_{2})] -  \min_{\tau_{1} \in \mathcal{T}_{1}^{*}(\tau_{2},P)}  \E_{P}[\psi(X_{i},\tau_{1},\tau_{2})] \right|,\\
&= \left|\max_{\tau_{1} \in \mathcal{T}_{1}^{*}(\tau_{2},P)}  - \E_{n}[\psi(X_{i},\tau_{1},\tau_{2})] - \max_{\tau_{1} \in \mathcal{T}_{1}^{*}(\tau_{2},\P_{n},b_{n})}  -\E_{n}[\psi(X_{i},\tau_{1},\tau_{2})]  \right|\\
&\qquad\qquad\qquad+ \left|\max_{\tau_{1} \in \mathcal{T}_{1}^{*}(\tau_{2},P)} -\E_{P}[\psi(X_{i},\tau_{1},\tau_{2})] -\max_{\tau_{1} \in \mathcal{T}_{1}^{*}(\tau_{2},P)} -\E_{n}[\psi(X_{i},\tau_{1},\tau_{2})]  \right|,\\
&\leq \max_{\{\tau_{1},\tau_{1}' \in \mathcal{T}_{1} : ||\tau_{1} - \tau_{1}'|| \leq d_{H}(\mathcal{T}_{1}^{*}(\tau_{2},\P_{n},b_{n}),\mathcal{T}_{1}^{*}(\tau_{2},P))\}}\left| - \E_{n}[\psi(X_{i},\tau_{1},\tau_{2})]  -  -\E_{n}[\psi(X_{i},\tau_{1}',\tau_{2})]  \right|\\
&\qquad\qquad\qquad+ \max_{\tau_{1} \in \mathcal{T}_{1}^{*}(\tau_{2},P)} \left| -\E_{n}[\psi(X_{i},\tau_{1},\tau_{2})]  -  - \E_{P}[\psi(X_{i},\tau_{1},\tau_{2})] \right|\\
&= \max_{\{\tau_{1},\tau_{1}' \in \mathcal{T}_{1} : ||\tau_{1} - \tau_{1}'|| \leq d_{H}(\mathcal{T}_{1}^{*}(\tau_{2},\P_{n},b_{n}),\mathcal{T}_{1}^{*}(\tau_{2},P))\}}\left| \E_{n}[\psi(X_{i},\tau_{1}',\tau_{2})] - \E_{n}[\psi(X_{i},\tau_{1},\tau_{2})] \right|\\
&\qquad\qquad\qquad+ \max_{\tau_{1} \in \mathcal{T}_{1}^{*}(\tau_{2},P)} \left| \E_{P}[\psi(X_{i},\tau_{1},\tau_{2})] -\E_{n}[\psi(X_{i},\tau_{1},\tau_{2})] \right|\\
&\leq \max_{\{\tau_{1},\tau_{1}' \in \mathcal{T}_{1} : ||\tau_{1} - \tau_{1}'|| \leq d_{H}(\mathcal{T}_{1}^{*}(\tau_{2},\P_{n},b_{n}),\mathcal{T}_{1}^{*}(\tau_{2},P))\}} C(\tau_{2}) \cdot ||\tau_{1} - \tau_{1}'|| + \max_{\tau_{1} \in \mathcal{T}_{1}^{*}(\tau_{2},P)} \left|\E_{P}[\psi(X_{i},\tau_{1},\tau_{2})] -  \E_{n}[\psi(X_{i},\tau_{1},\tau_{2})]  \right|\\
&\leq C(\tau_{2}) \cdot d_{H}(\mathcal{T}_{1}^{*}(\tau_{2},\P_{n},b_{n}),\mathcal{T}_{1}^{*}(\tau_{2},P))+ \max_{\tau_{1} \in \mathcal{T}_{1}^{*}(\tau_{2},P)} \left|\E_{P}[\psi(X_{i},\tau_{1},\tau_{2})] -  \E_{n}[\psi(X_{i},\tau_{1},\tau_{2})]  \right|.
\end{align*}
It suffices to show the two terms in the last line of the previous display converge to zero in probability. The second term converges in probability to zero by Assumption \ref{assumption_consistency}(vi). Furthermore, since $\sup_{\tau_{2} \in \mathcal{T}_{2}} C(\tau_{2})<\infty$ a.s., the first term converges to zero in probability if we can show that:
\begin{align*}
d_{H}(\mathcal{T}_{1}^{*}(\tau_{2},\P_{n},b_{n}),\mathcal{T}_{1}^{*}(\tau_{2},P)) = o_{P}(1). 
\end{align*}
The remainder of the proof focuses on proving this latter fact. Note that:
\begin{align*}
d_{H}(\mathcal{T}_{1}^{*}(\tau_{2},\P_{n},b_{n}),\mathcal{T}_{1}^{*}(\tau_{2},P))&= \inf\{\delta>0 : \mathcal{T}_{1}^{*}(\tau_{2},P) \subset \mathcal{T}_{1}^{*}(\tau_{2},\P_{n},b_{n})^\delta, \text{ and } \mathcal{T}_{1}^{*}(\tau_{2},\P_{n},b_{n}) \subset \mathcal{T}_{1}^{*}(\tau_{2},P)^\delta\},
\end{align*}
where:
\begin{align*}
\mathcal{T}_{1}^{*}(\tau_{2},\P_{n},b_{n})^\delta&:= \{\tau_{1} \in \mathcal{T}_{1} : B_{\delta}(\tau_{1})\cap \mathcal{T}_{1}^{*}(\tau_{2},\P_{n},b_{n})\neq \emptyset  \},\\
\mathcal{T}_{1}^{*}(\tau_{2},P)^\delta&:= \{\tau_{1} \in \mathcal{T}_{1} : B_{\delta}(\tau_{1})\cap \mathcal{T}_{1}^{*}(\tau_{2},P)\neq \emptyset  \},
\end{align*}
where $B_{\delta}(\tau_{1})$ denotes the closed ball of radius $\delta>0$ around $\tau_{1}$. The next part of the proof closely follows the proof of Theorem 2.1 in \cite{molchanov1998limit}. Define the function:
\begin{align*}
\rho(\varepsilon):= d_{H}(\mathcal{T}_{1}^{*}(\tau_{2},P,\varepsilon),\mathcal{T}_{1}^{*}(\tau_{2},P)).
\end{align*}
Since each of the moment functions are convex (and thus lower semi-continuous) in $\tau_{1}$ for each $\tau_{2}$, each of the sets $\mathcal{T}_{1}^{*}(\tau_{2},P,\varepsilon)$ and $\mathcal{T}_{1}^{*}(\tau_{2},P)$ are closed and $\rho$ is right continuous. Furthermore, $\rho$ is non-increasing for $\varepsilon<0$ and non-decreasing for $\varepsilon>0$. Now by Assumption \ref{assumption_consistency} we have with high probability:
\begin{align*}
\mathcal{T}_{1}^{*}(\tau_{2},\P_{n},b_{n}) &= \{\tau_{1} \in \mathcal{T}_{1} : \E_{n}[m_{j}(X_{i},\tau_{1},\tau_{2})] \leq b_{n} \text{ for $j\in \mathcal{J}(\tau_{2})$} \}\\
&\subset \{\tau_{1} \in \mathcal{T}_{1} : \E_{P}[m_{j}(X_{i},\tau_{1},\tau_{2})] \leq \eta_{n}(\tau_{2}) + b_{n} \text{ for $j\in \mathcal{J}(\tau_{2})$} \}\\
&\subset \mathcal{T}_{1}^{*}(\tau_{2},P,2b_{n})\\
&\subset \mathcal{T}_{1}^{*}(\tau_{2},P)^{\rho(2b_{n})}.
\end{align*}
Furthermore, by Assumption \ref{assumption_consistency} we have with high probability for large enough $n$:
\begin{align*}
\mathcal{T}_{1}^{*}(\tau_{2},P) &\subset \mathcal{T}_{1}^{*}(\tau_{2},P,b_{n}-\eta_{n}(\tau_{2}))\\
&\subset \mathcal{T}_{1}^{*}(\tau_{2},\P_{n},b_{n}).
\end{align*}
Conclude that with high probability for large enough $n$:
\begin{align*}
d_{H}(\mathcal{T}_{1}^{*}(\tau_{2},\P_{n},b_{n}),\mathcal{T}_{1}^{*}(\tau_{2},P)) \leq \rho(2b_{n}) \to 0,
\end{align*}
where the last line follows from right-continuity of the function $\rho(\,\cdot\,)$. Since $\tau_{2} \in \mathcal{T}_{2}'$ was arbitrary, this completes the proof.
\end{proof}

\subsection{The Additively Separable Case}\label{appendix_additively_separable}

In this subsection we show how our method can be applied to a model that satisfies the following assumption. 

\begin{assumption}\label{assumption_threshold_crossing} The index function $\varphi$ satisfying Assumption \ref{assumption_basic} is additively separable in $U$; i.e. we have $\varphi(X,U,\theta) = \tilde{\varphi}(X,\theta) - U$ for some function $\tilde{\varphi}$.
\end{assumption}
This is a well-studied special case of the linear model considered in the main text. In particular, much of the discussion in this section expands upon the insights of \cite{chesher2013semiparametric}. We consider two cases: (i) when the structural function $\varphi$ is linear in the parameter vector $\theta$, and (ii) when the structural function is unknown. To begin, let us consider the following simple example.

\begin{example}\label{example_threshold_crossing}
Suppose we have a scalar variable $X$ with support $\mathcal{X} = \{x_{1},\ldots,x_{m}\}$ and latent variables $U \in [-1,1]$. Consider the following additively separable threshold crossing model:
\begin{align*}
Y = \mathbbm{1}\{X\theta \geq U \},
\end{align*}
where $\theta$ is a fixed scalar coefficient. The response types in this setting are characterized by the $m\times 1$ vectors:
\begin{align*}
r(u,\theta) := 
\begin{bmatrix}
\mathbbm{1}\{x_{1}\theta \geq u \}\\
\mathbbm{1}\{x_{2}\theta \geq u \}\\
\vdots\\
\mathbbm{1}\{x_{m}\theta \geq u \}\\
\end{bmatrix}.
\end{align*}
However, the set of possible response types in this setting depends on the sign of the fixed coefficient $\theta$. In particular, when $\theta\geq0$ we have the response types $r(u,\theta) \in \{s_{1}, \ldots, s_{m+1}\}$, where:
\begin{align}
s_{1}:= 
\begin{bmatrix}
0\\
0\\
\vdots\\
0\\
0
\end{bmatrix}, && s_{2}:= 
\begin{bmatrix}
0\\
0\\
\vdots\\
0\\
1
\end{bmatrix}, &&\ldots, &&
s_{m}:= 
\begin{bmatrix}
0\\
1\\
\vdots\\
1\\
1
\end{bmatrix},&&
s_{m+1}:= 
\begin{bmatrix}
1\\
1\\
\vdots\\
1\\
1
\end{bmatrix}.\label{eq_pos_ordering}
\end{align}
No other response types are possible when $\theta>0$, and so all other response types must be assigned zero probability. Alternatively, when $\theta<0$ we have the response types $r(u,\theta) \in \{s_{1}', \ldots, s_{m+1}'\}$, where:
\begin{align}
s_{1}':= 
\begin{bmatrix}
0\\
0\\
\vdots\\
0\\
0\\
\end{bmatrix}, && s_{2}':= 
\begin{bmatrix}
1\\
0\\
\vdots\\
0\\
0
\end{bmatrix}, &&\ldots,&&
s_{m}':= 
\begin{bmatrix}
1\\
1\\
\vdots\\
1\\
0
\end{bmatrix},&&
s_{m+1}':= 
\begin{bmatrix}
1\\
1\\
\vdots\\
1\\
1
\end{bmatrix}.\label{eq_neg_ordering}
\end{align}
Again, all other response types must be assigned zero probability by the distribution of $U$. 

The reason that these particular response types arise when $\theta\geq0$ and $\theta<0$ is due to the ordering of the support of $X$ induced by the value of the scalar product $X\theta$. In particular, if we suppose $x_{1} \leq x_{2} \leq \ldots \leq x_{m}$, then when $\theta \geq0$ we have the ordering $x_{1}\theta \leq x_{2}\theta\leq \ldots \leq x_{m}\theta$. This means, for example, that it is impossible to find a value of $u \in [-1,1]$ so that:
\begin{align*}
r(u,\theta)= \begin{bmatrix}
\mathbbm{1}\{x_{1}\theta \geq u \}\\
\mathbbm{1}\{x_{2}\theta \geq u \}\\
\mathbbm{1}\{x_{3}\theta \geq u \}\\
\vdots\\
\mathbbm{1}\{x_{m}\theta \geq u \}
\end{bmatrix} = \begin{bmatrix}
0\\
1\\
0\\
\vdots\\
0
\end{bmatrix}.
\end{align*}
This shows that when $\theta \geq 0$ certain response types are not possible, and so must be assigned probability zero by the distribution of $U$. An identical intuition holds in the case when $\theta<0$. In the end, the response types that can be assigned positive probability in this example when $\theta \geq 0$ and $\theta <0$ are exactly the ones corresponding to the vectors in \eqref{eq_pos_ordering} and \eqref{eq_neg_ordering}, respectively. Figure \ref{fig_profile_eg2} provides an illustration in the case when $\mathcal{X}=\{x_{1},x_{2},x_{3}\}$.

\begin{figure}[!t]
\centering
\includegraphics[scale=0.7]{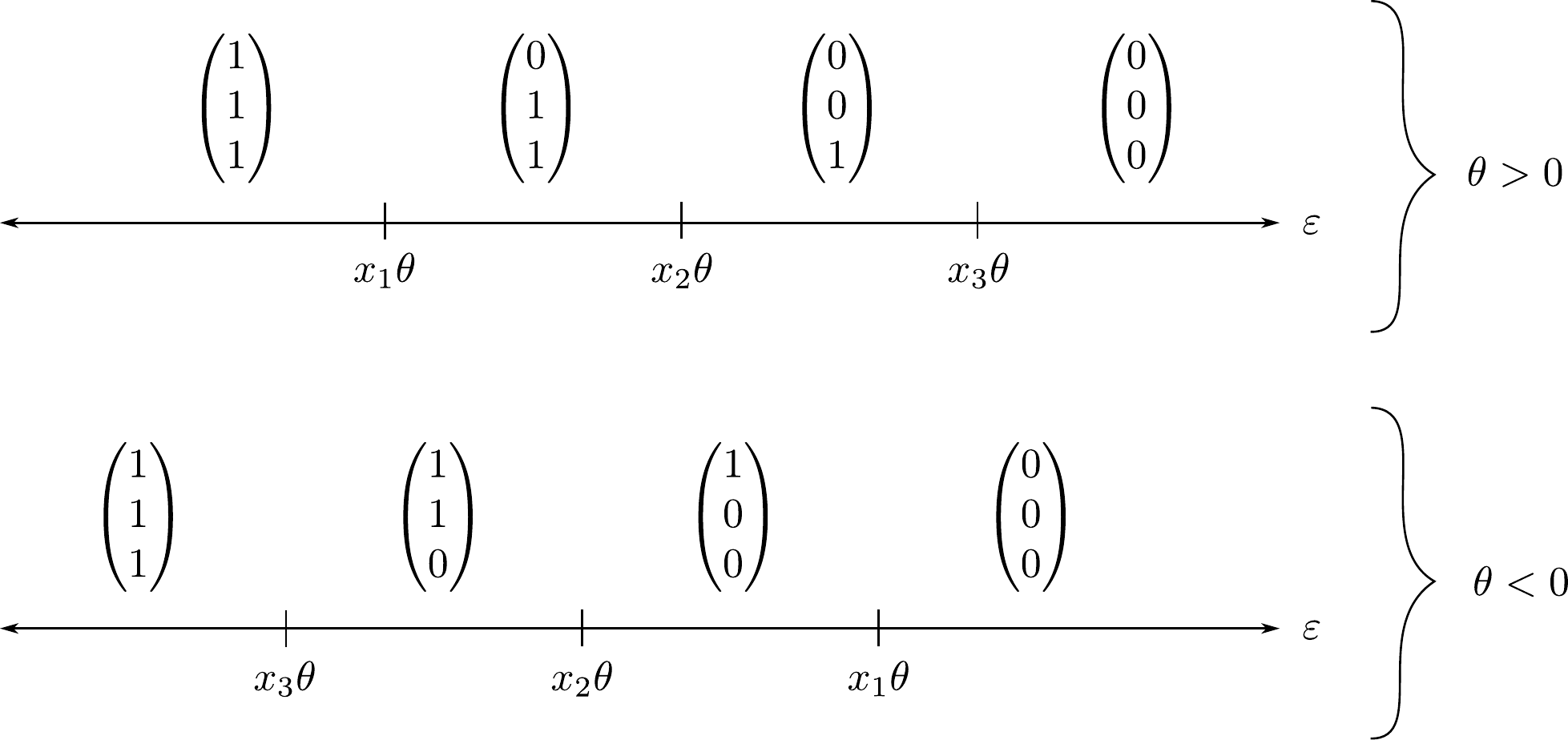}
\caption{A figure corresponding to Example \ref{example_threshold_crossing} illustrating the partition of the latent variable space according to response types in the case when the index function is additively separable in $U$ and when $\mathcal{X}= \{x_{1},x_{2},x_{3}\}$ with $x_{1} \leq x_{2} \leq x_{3}$. As indicated in the example, the feasible response types are those that correspond to a particular ordering of the points in $\mathcal{X}$ induced by the scalar product $X \theta$.} \label{fig_profile_eg2}
\end{figure}
 
\end{example}

This example illustrates the key ideas behind the implementation of our approach when the index function is additively separable in $U$, as in Assumption \ref{assumption_threshold_crossing}. In particular, given the function $\tilde{\varphi}$ from Assumption \ref{assumption_threshold_crossing}, the key is to determine the values of $\theta$ such that the function $\tilde{\varphi}(\,\cdot\,,\theta): \mathcal{X} \to \mathbb{R}$ induces a unique ordering of the points in the support $\mathcal{X}$. With a scalar $X$ variable, and $\tilde{\varphi}(X,\theta) = X \theta$, Example \ref{example_threshold_crossing} shows that only two orderings are possible, corresponding to the case when $\theta \geq 0$ and $\theta <0$. After the order is determined, we can immediately determine the set of response types that must be assigned zero probability by the distribution of $U$, and then impose these restrictions as an additional constraint in the bounding problems \eqref{eq_thm_linprog_finite_LB} and \eqref{eq_thm_linprog_finite_UB}. In particular, let $S_{\varphi}(\theta)$ denote the set of all binary vectors $s \in \{0,1\}^{m}$ corresponding to sets $\mathcal{U}(s,\theta)$ that can be assigned positive probability under Assumption \ref{assumption_threshold_crossing}, and impose the constraint:
\begin{align}
\sum_{s \in S_{\varphi}(\theta)^{c}} \pi(y,x_{j},\theta,s) = 0, \label{eq_linearity_constraint2}
\end{align}
for all $y\in \{0,1\}$ and $j=1,\ldots,m$ occurring with positive probability. Then Theorem \ref{thm_linprog_finite_case} can be extended to accommodate Assumption \ref{assumption_threshold_crossing} by simply adding the constraints \eqref{eq_linearity_constraint2} to the optimization problems \eqref{eq_thm_linprog_finite_LB} and \eqref{eq_thm_linprog_finite_UB}.

Similar to the discussion in the main text, determining the sets $\mathcal{U}(s,\theta)$ that can be assigned positive probability under Assumption \ref{assumption_threshold_crossing} poses an interesting computational problem. Although Example \ref{example_threshold_crossing} illustrates a case when there are only two orderings, in general many more orderings may be possible, even when $\tilde{\varphi}$ is linear in $\theta$. Clearly at most $m!$ orderings are possible, but when the index function is linear in $\theta$ it is possible to show that the maximum number of possible orderings is much smaller than $m!$. In particular, consider the function $\tilde{\varphi}(X,\theta) = X\theta$ where $X$ is a vector of dimension $d$. Label the support $\mathcal{X}$ as $\{x_{1},x_{2},\ldots,x_{m} \}$, and let $\Delta_{jk} := x_{j} - x_{k}$ for $1\leq j < k \leq m$. The set $H_{jk}:=\{ \theta \in \mathbb{R}^{d} : \Delta_{jk} \theta = 0\}$ defines a hyperplane through the origin that is normal to the line connecting $x_{j}$ and $x_{k}$ in $\mathbb{R}^{d}$. The set of all such hyperplanes partitions $\mathbb{R}^{d}$ into at most $Q(m,d)$ nonempty cones, where $Q(m,d)$ is defined recursively as: 
\begin{equation}
	Q(m, d) = Q(m-1, d) + (m-1) Q(m-1,d-1), \label{eq_Osets}
\end{equation}
with $Q(m,1) = 2$ for all $m \geq 2$ and $Q(2, d) = 2$ for all $d \geq 1$.  Furthermore, each these nonempty cones corresponds exactly to the equivalence class of vectors $\theta$ that induce a unique ordering of the points in $\mathcal{X}$. Thus, the value $Q(m,d)$ serves as an upper bound on the number of orderings of the points in $\mathcal{X}$ that are inducible by the function $\tilde{\varphi}(X,\theta) = X\theta$. The recursive formula from \eqref{eq_Osets} defining the upper bound $Q(m,d)$ has been independently discovered in different contexts by many authors; the earliest such account appears in \cite{Bennett1956}, although the formula was independently discovered again in \cite{cover1967number}. The upper bound $Q(m,d)$ is obtained when the collection of hyperplanes of the form $H_{jk}$ are in general position. Note that $Q(m,1)=2$ corresponds exactly to Example \ref{example_threshold_crossing}, where it was shown that only two orderings could be induced when $\tilde{\varphi}(X,\theta) = X \theta$ for scalar $X$ and $\theta$. Typically, $Q(m, d) < m!$, although some inspection of the formula shows that we always have $Q(m,d)=m!$ when $d \geq m-1$.

If we could select one value of $\theta$ from each of the cones defined by the collection of hyperplanes of the form $H_{jk}$, we could then determine the permitted orderings of the support points $\mathcal{X}$ by simply evaluating $x_{j} \theta$ for $j=1,\ldots,m$. This would then allow us to determine which sets $\mathcal{U}(s,\theta)$ must be assigned zero probability under Assumption \ref{assumption_threshold_crossing}. Note that under Assumption \ref{assumption_threshold_crossing} the latent variable $U$ obtains a value on the hyperplane $H_{jk}$ with probability zero. Thus, it suffices to select one value of $\theta$ from the interior of each of the cones defined by the collection of hyperplanes of the form $H_{jk}$. This can be done using the hyperplane arrangement algorithm described in the main text applied to the hyperplanes of the form $H_{jk}$ for $1\leq j < k \leq m$.  

Our method is also applicable to cases when $\tilde{\varphi}(X,\theta)$ may be non-linear in $\theta$. When $\tilde{\varphi}$ is not restricted by the researcher, all orderings of the support points in $\mathcal{X}$ are possible. The researcher must first fix an ordering of the support points in $\mathcal{X}$, determine the admissible response types $S_{\varphi}(\theta)$ for the fixed ordering, and run the linear programs in \eqref{eq_thm_linprog_finite_LB} and \eqref{eq_thm_linprog_finite_UB} subject to the constraint \eqref{eq_linearity_constraint2}. The researcher must then repeat the procedure for all possible orderings of the support points in $\mathcal{X}$. On each iteration of this procedure the researcher obtains an interval with endpoints determined by the values of the linear programs in \eqref{eq_thm_linprog_finite_LB} and \eqref{eq_thm_linprog_finite_UB}. The closed convex hull of the identified set for the counterfactual probability is then given by the interval whose lower endpoint is the smallest value of the linear program in \eqref{eq_thm_linprog_finite_LB} obtained across all orderings, and whose upper endpoint is the largest value of the linear program in \eqref{eq_thm_linprog_finite_UB} obtained across all orderings. There are $m!$ possible orderings for $\tilde{\varphi}(X,\theta)$ unless additional assumptions are imposed, so that considering all possible orderings can quickly become computationally demanding. 

\section{Comparison to Artstein's Inequalities}\label{appendix_artstein_comparison}

Here we briefly discuss the method proposed by \cite{chesher2013instrumental} and \cite{chesher2014instrumental}.\footnote{The general version of the approach can be found in \cite{chesher2017generalized}.} Our objective is to provide an informal comparison, and to illustrate the connections between the two approaches. To ease the comparison, we focus on the identified set for (conditional) latent variable distributions rather than counterfactual probabilities. Suppose Assumption \ref{assumption_basic} holds, and consider the correspondence:
\begin{align}
\overline{\mathcal{U}}(y,x,\theta)&:=\text{cl}\left\{ u \in \mathcal{U} : y = \mathbbm{1}\{\varphi(x,u,\theta) \geq 0  \} \right\}.\label{eq_random_set2}
\end{align}
The set $\overline{\mathcal{U}}(Y,X,\theta)$ is the closure of the set $\mathcal{U}(Y,X,\theta)$ from \eqref{eq_random_set}. Under some conditions, when the distribution of $U$ is absolutely continuous (which is assumed, for example, in both \cite{chesher2013instrumental} and \cite{chesher2014instrumental}), these two random sets are equal almost surely.\footnote{In \cite{chesher2014instrumental}, absolute continuity combined with a linear index function ensures this statement is true. \cite{chesher2013instrumental} consider a more general class of latent index functions than \cite{chesher2014instrumental}, and so also impose strict monotonicity in latent variables of the latent index function in order to ensure their analog of the set $\left\{u \in \mathcal{U} : \varphi(x,u,\theta) = 0  \right\}$ is of Lebesgue measure zero for each $(x,\theta)$.} Also define the correspondence:
\begin{align}
H(u,\theta)&:=\left\{(y,x) \in \mathcal{Y} \times \mathcal{X} : y = \mathbbm{1}\{\varphi(x,u,\theta) \geq 0  \} \right\}.\label{eq_random_set3}
\end{align}
A result due to \cite{artstein1983distributions} (also see \cite{norberg1992existence} and \cite{molchanov2017theory} Corollary 1.4.11), characterizes the set of selections of a random closed set. 
\begin{theorem}\label{theorem_artstein}
Suppose that Assumption \ref{assumption_basic} holds. Then for any $\theta \in \Theta$, the random vector $U$ can be realized as a selection of the random closed set $\overline{\mathcal{U}}(Y,X,\theta)$ if and only if:
\begin{align}
P_{U}(U \in K) \leq P_{Y,X}\left(\overline{\mathcal{U}}(Y,X,\theta)\cap K \neq \emptyset\right), \label{eq_artstein}
\end{align}
for all compact sets $K \subset \mathcal{U}$. Furthermore, for any $\theta \in \Theta$, the random vector $(Y,X)$ can be realized as a selection of the random closed set $H(U,\theta)$ if and only if:
\begin{align}
P_{Y,X}((Y,X) \in C) \leq P_{U}( H(U,\theta)\cap C \neq \emptyset), \label{eq_artstein2}
\end{align}
for all compact sets $C \subset \mathcal{Y}\times \mathcal{X}$. 
\end{theorem}
\begin{remark}
The random vector $U$ can be realized as a selection of the random closed set $\overline{\mathcal{U}}(Y,X,\theta)$ if and only if there exists a probability space and random elements $\tilde{U}$ and $\tilde{\mathcal{U}}(Y,X,\theta)$ with identical distributions to $U$ and $\overline{\mathcal{U}}(Y,X,\theta)$ such that $\tilde{U} \in \tilde{\mathcal{U}}(Y,X,\theta)$ a.s. Since $\mathcal{U} \subset \mathbb{R}^{d_{\theta}}$ is locally compact and Hausdorff, it is equivalent that \eqref{eq_artstein} hold for all open sets $G \subset \mathcal{U}$. Note that since $\mathcal{Y}\times \mathcal{X}$ is finite, all subsets are trivially compact with respect to the discrete topology. 
\end{remark}
The first part of this result is very similar to Theorem 1 in \cite{chesher2013instrumental} and Theorem 3.1 in \cite{chesher2014instrumental}, and the second part of this result is a direct corollary of Theorem 1 in \cite{chesher2017generalized}. Either \eqref{eq_artstein} and \eqref{eq_artstein2} can be used to construct the identified set of unconditional latent variable distributions, say $\mathcal{P}_{U}^{*}$; in practice, this is accomplished by first fixing a value of $\theta \in \Theta$, collecting all distributions $P_{U}$ satisfying either \eqref{eq_artstein} or \eqref{eq_artstein2}, and then taking a union (over all $\theta \in \Theta$) of the resulting collections of distributions. A similar result to Theorem \ref{theorem_artstein} can be stated after conditioning on $(Y,X)$. 

The main difficultly with using the characterization of the identified set based on Artstein's inequalities is the number of constraints that must be imposed. The resulting computational bottleneck is thus associated with a lack of short-term computer memory needed to store all of Artstein's inequalities. Most efforts to reduce the computational burden of the approach based on Artstein's inequalities are directed towards reducing the number of constraints implied by Theorem \ref{theorem_artstein}; see the discussions in \cite{galichon2011set}, \cite{chesher2017generalized}, \cite{russell2021sharp}. 

At first glance it appears that \eqref{eq_artstein} leads to a characterization of the identified set of unconditional latent variable distributions that is intractable, given the number of possible compact subsets of $\mathcal{U}$. However, following the discussion in both \cite{chesher2013instrumental} and \cite{chesher2014instrumental}, most of the inequalities of the form \eqref{eq_artstein} are redundant. For instance, when $\varphi$ is linear in $U$ and $\theta$, the set $\overline{\mathcal{U}}(y,x,\theta)$ is a closed halfspace through the origin. In this case, \cite{chesher2014instrumental} show it suffices to check the inequalities in \eqref{eq_artstein} for all sets $K$ that can be written as the intersection of halfspaces of the form $\overline{\mathcal{U}}(y,x,\theta)$. In the case with no exogenous variables and a scalar endogenous variable $X$ with $m$ points of support, \cite{chesher2014instrumental} demonstrate that there are at most $2m(2m+1)/2$ nonredundant inequalities implied by \eqref{eq_artstein}. 

On the other hand, since $r:=|\mathcal{Y}\times \mathcal{X}|$ is finite, \eqref{eq_artstein2} gives $2^r-1$ inequalities. Even with small values of $m$ the resulting number of inequalities can also be prohibitively large. We now show that, at the cost of some pre-processing, our approach leads to a simplification of the set of constraints in \eqref{eq_artstein2}: the number of constraints in our approach is proportional to $r$ rather than $2^r$. 

Consider imposing \eqref{eq_artstein2} conditional on $(Y,X)$, and assume for simplicity that each $(y,x) \in \mathcal{Y}\times\mathcal{X}$ is assigned positive probability. For any $\theta \in \Theta$, the random vector $U$ can be realized as a selection from $\overline{\mathcal{U}}(y,x,\theta)$ if and only if:
\begin{align*}
\mathbbm{1}\{(y,x) \in C\} \leq P_{U \mid Y,X}( H(U,\theta)\cap C \neq \emptyset \mid Y=y,X=x),
\end{align*}
for all compact $C \subset \mathcal{Y}\times\mathcal{X}$. For a fixed value of $(y,x)$, consider the set of all compact $C \subset \mathcal{Y}\times\mathcal{X}$ containing $(y,x)$. For all such $C$ we must have:
\begin{align*}
P_{U \mid Y,X}( H(U,\theta)\cap C \neq \emptyset \mid Y=y,X=x)=1.
\end{align*}
However, this can hold for all compact sets $C \subset \mathcal{Y}\times\mathcal{X}$ containing $(y,x)$ if and only if it holds for the singleton set $\{(y,x)\}$. Conclude that:
\begin{align*}
P_{U\mid Y,X}( H(U,\theta)\cap \{(y,x)\} \neq \emptyset \mid Y=y,X=x)=1.
\end{align*}
Some basic manipulation shows this holds if and only if:
\begin{align*}
P_{U \mid Y,X}(U \in  \overline{\mathcal{U}}(Y,X,\theta) \mid Y=y,X=x)=1.
\end{align*}
This derivation can be used to prove the following Lemma.
\begin{lemma}\label{lemma_comparision}
Suppose that Assumption \ref{assumption_basic} holds and that all $(y,x) \in \mathcal{Y}\times\mathcal{X}$ are assigned positive probability. Then Artstein's inequalities:
\begin{align}
P_{Y,X}((Y,X) \in C) \leq P_{U}( H(U,\theta)\cap C \neq \emptyset), \label{eq_artstein3}
\end{align}
hold for all compact sets $C \subset \mathcal{Y}\times \mathcal{X}$ if and only if:
\begin{align}
P_{U \mid Y,X}(U \in  \overline{\mathcal{U}}(Y,X,\theta) \mid Y=y,X=x)=1,\label{eq_artstein_equivalence}
\end{align}
for all $(y,x)$.
\end{lemma}
Notice that \eqref{eq_artstein_equivalence} is nearly identical to the condition in Definition \ref{definition_identified_set} characterizing $\mathcal{P}_{U\mid Y,X}^{*}$, connecting the approach based on Artstein's inequalities to the approach considered in this paper. Setting $r=|\mathcal{Y} \times \mathcal{X}|$, we conclude that the $2^{r}-1$ inequalities implied by \eqref{eq_artstein3} are equivalent to the (at most) $r$ equalities implied by \eqref{eq_artstein_equivalence}. For the sake of comparison, recall that in the case with no exogenous variables and a scalar endogenous variable $X$ with $m$ points of support, \cite{chesher2014instrumental} demonstrate that there are at most $2m(2m+1)/2$ nonredundant inequalities implied by \eqref{eq_artstein} when $\varphi$ is linear in $(U,\theta)$. In contrast, Lemma \ref{lemma_comparision} implies that with no exogenous variables and a scalar endogenous variable $X$ with $m$ points of support, there are at most $2m$ constraints implied by \eqref{eq_artstein_equivalence}, regardless of the functional form specified for $\varphi$.

We conclude by noting that our characterization may not always produce less constraints than the approach based on Artstein's inequalities. In particular, the characterization based on Artstein's inequalities easily handles the incorporation of instruments without increasing the number of inequality constraints (c.f. \cite{beresteanu2012partial} Proposition 2.5), where our approach may require significantly more constraints. However, in many cases we believe our approach can offer substantial simplifications. 

\section{Additional Figures}\label{appendix_figure}

This section contains two figures illustrating the intervals computed using the linear programs of the form \eqref{eq_thm_linprog_finite_LB} and \eqref{eq_thm_linprog_finite_UB} for each profiling point for certain specifications in the application section. Figure \ref{ATE_hole1} presents the intervals computed for each representative point of $\theta \in \mathbb{R}^2$ when bounding $\mu_{ate}$ for Model \eqref{eq_app_linear1} under various assumptions. Figure \ref{ATE_hole2} shows the analogous intervals for model \eqref{eq_app_linear2}.

\begin{sidewaysfigure}
	\centering 
	\includegraphics[scale=0.47]{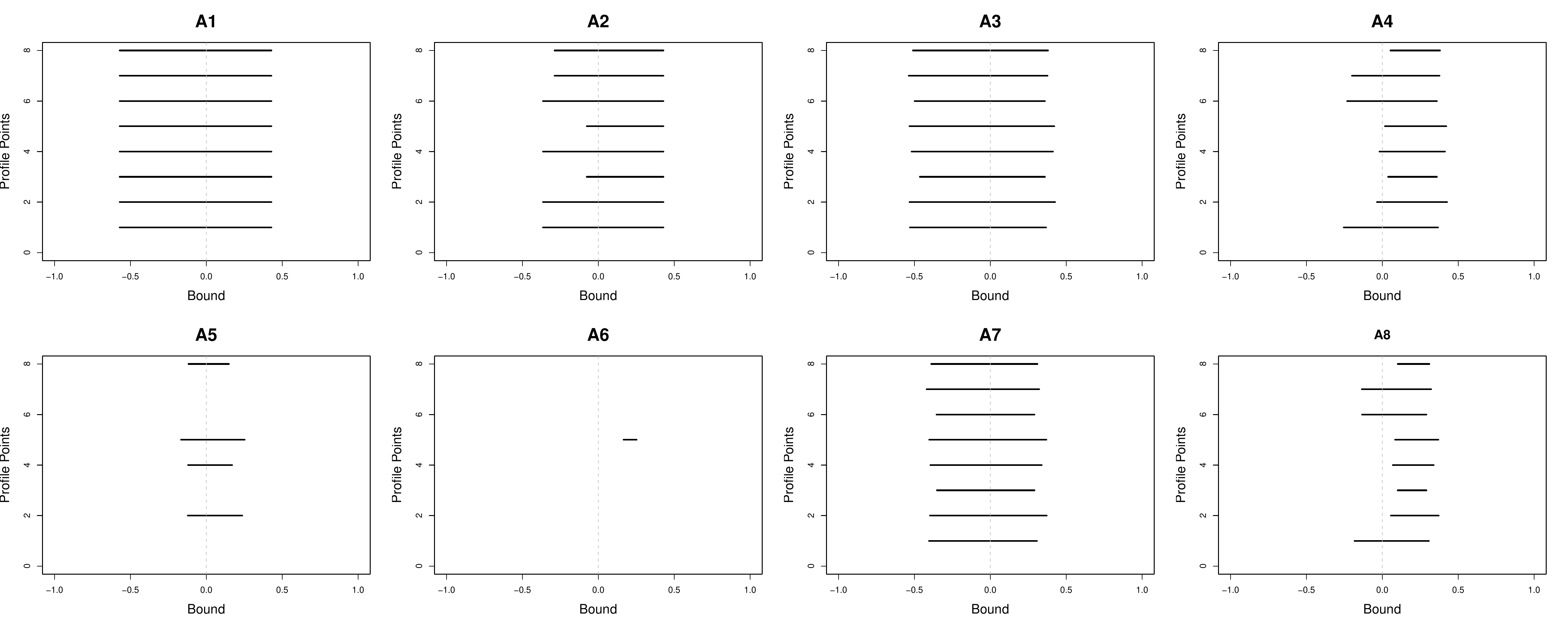}
	\caption{This figure shows the intervals computed using the linear programs of the form \eqref{eq_thm_linprog_finite_LB} and \eqref{eq_thm_linprog_finite_UB} for each representative point of $\theta \in \mathbb{R}^2$ when bounding $\mu_{ate}$ for Model \eqref{eq_app_linear1} under various assumptions. The active assumptions are given at the top of each illustration. The axes labelled ``Profile Points'' correspond to various representative points. }
	\label{ATE_hole1} 
\end{sidewaysfigure}

\begin{sidewaysfigure}
	\centering 
	\includegraphics[scale=0.47]{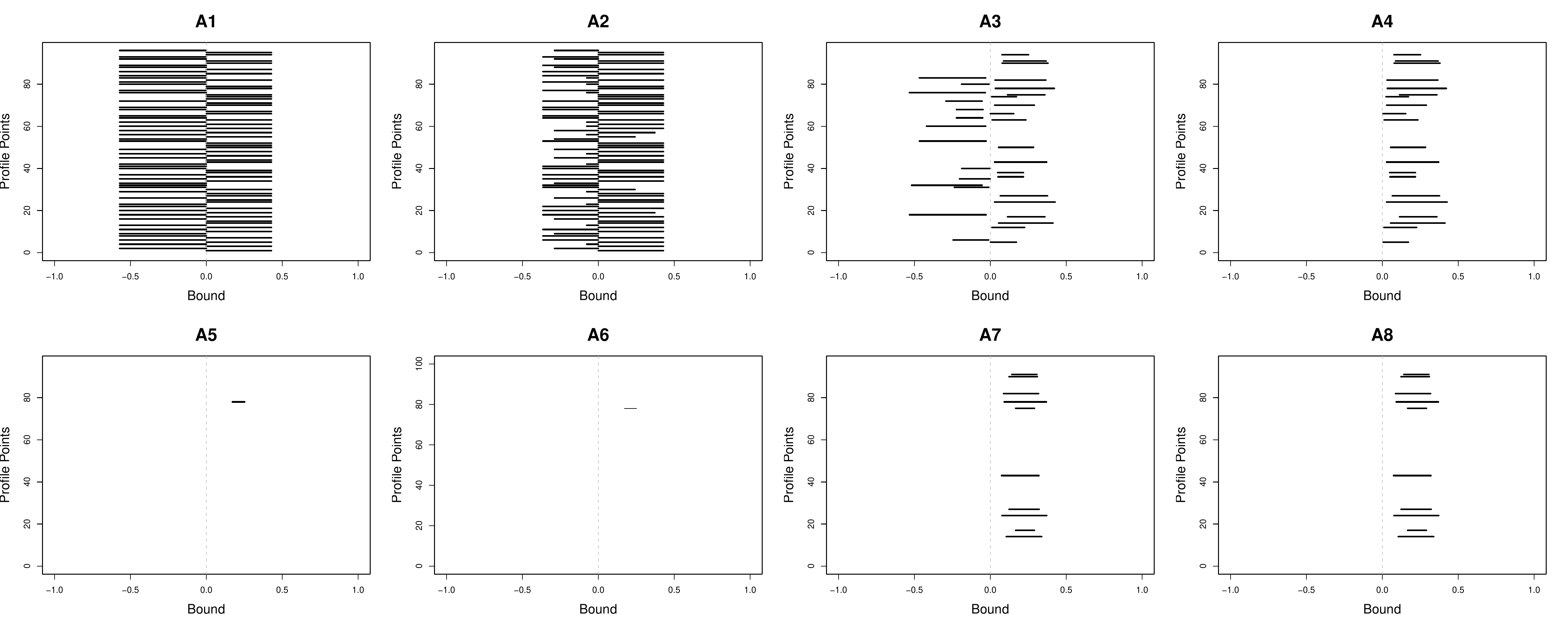}
	\caption{This figure shows the intervals computed using the linear programs of the form \eqref{eq_thm_linprog_finite_LB} and \eqref{eq_thm_linprog_finite_UB} for each representative point of $\theta \in \mathbb{R}^3$ when bounding $\mu_{ate}$ for Model \eqref{eq_app_linear2} under various assumptions. The active assumptions are given at the top of each illustration. The axes labelled ``Profile Points'' correspond to various representative points.}
	\label{ATE_hole2} 
\end{sidewaysfigure}

\end{appendix}



\end{document}